\documentclass[a4paper]{llncs}

\usepackage{makeidx}  

\usepackage{microtype}

\usepackage{graphicx}
\usepackage{epsfig}
\usepackage{xcolor}
\usepackage{colortbl}
\usepackage{paralist}
\usepackage{pslatex}
\usepackage{color}
\usepackage{marginnote}
\usepackage{url}

\usepackage[textsize=small]{todonotes}
\usepackage{multirow}
\usepackage{makecell}
\usepackage{latexsym}
\usepackage{amsmath}
\usepackage{amssymb}
\usepackage{rotating}
\usepackage{subfig}
\usepackage{caption}

\usepackage{wrapfig}

\usepackage[normalem]{ulem}
\usepackage{upgreek}

\usepackage{pgf}
\usepackage{tikz}
\usetikzlibrary{
	shapes,%
	decorations,%
	decorations.pathmorphing,%
	decorations.text,%
	decorations.pathreplacing,%
	automata,%
	arrows,%
	chains,%
	matrix,%
	scopes,%
	shadows,%
	positioning,%
	patterns,%
	fit,%
	calc,%
	through,%
	fadings,%
	backgrounds}
\tikzset{
	smallstate/.style={state,
		circle, 
		minimum size=4mm, 
		font=\itshape
	},
	initial text=,
	>=stealth',
	>/.style={ultra thick}
} 

\newcommand{\A}{A}

\newcommand{\allcompletetrees}{\mathbb C}    
\newcommand{\alltrees}{\mathbb T}    

\newcommand{\bwdirecttraceinclusion}{\xincl{\mathsf{bw}}}
\newcommand{\bwsim}{\xsim{\mathsf{bw}}}

\newcommand{\directtraceinclusion}{\xincl{\mathsf{di}}}
\newcommand{\disim}{\xsim{\mathsf{di}}}

\newcommand{\dl}{D}    
\newcommand{\dli}{\subseteq_\dl}
\newcommand{\dom}{dom}
\newcommand{\drel}{R_\mathsf{d}}		
\newcommand{\ldrel}{\hat{R}_\mathsf{d}}		

\newcommand{\dwsimlarg}{\xsimlarg{\mathsf{dw}}}
\newcommand{\dwsim}{\xsim{\mathsf{dw}}}

\newcommand{\dwtraceinclusionlarg}{\xtraceinclusionlarg{\mathsf{dw}}}
\newcommand{\dwtraceinclusion}{\xincl{\mathsf{dw}}}

\newcommand{\f}{a}

\newcommand{\id}{{\it id}} 
\newcommand{\irel}{R_i}		

\newcommand{\ksim}[1]{\sqsubseteq^{#1}}
\newcommand{\kxsim}[2]{\ksim{#1\textrm - #2}}

\newcommand{\kdwsim}[1]{\kxsim{#1}{\mathsf{dw}}}
\newcommand{\kupsim}[2]{\kxsim{#1}{\mathsf{up}}\!\!(#2)}

\newcommand{\transasymkupsim}[2]{\prec^{#1\textrm{-}\mathsf{up}}\!\!(#2)}

\newcommand{\transkdwsim}[1]{\preceq^{#1\textrm{-}\mathsf{dw}}}
\newcommand{\transkupsim}[2]{\preceq^{#1\textrm{-}\mathsf{up}}\!\!(#2)}
\newcommand{\transasymkdwsim}[1]{\prec^{#1\textrm{-}\mathsf{dw}}}

\newcommand{\transTD}[4]{\langle{#1},{#2},\vect{#3}{#4}\rangle}	

\newcommand{\makeprunerel}[2]{\prunerel({#1},{#2})}
    
\newcommand{\nat}{\mathbb{N}}

\newcommand{\node}{v}
\newcommand{\nodew}{w}
\newcommand{\nodex}{x}

\newcommand{\prunerel}{P}
\newcommand{\rank}{\#}
\newcommand{\run}{\pi} 

\newcommand{\spstate}{\psi}						
\newcommand{\strictbwdirecttraceinclusion}{\xstrictincl{\mathsf{bw}}}
\newcommand{\strictbwsim}{\strictxsim{\mathsf{bw}}}

\newcommand{\strictdirecttraceinclusion}{\xstrictincl{\mathsf{di}}}
\newcommand{\strictdisim}{\strictxsim{\mathsf{di}}}
\newcommand{\strictdwsimlarg}{\strictxsimlarg{\mathsf{dw}}}
\newcommand{\strictdwsim}{\strictxsim{\mathsf{dw}}}
\newcommand{\strictdwtraceinclusion}{\xstrictincl{\mathsf{dw}}}
\newcommand{\strictdwtraceinclusionlarg}{\xstrictincllarg{\mathsf{dw}}}

\newcommand{\strictupsim}[1]{\strictxsim{\mathsf{up}}\!\!\!\;\!(#1)}
\newcommand{\strictuptraceinclusion}[1]{\xstrictincl{\mathsf{up}}\!\!\!\;\!(#1)}
\newcommand{\strictxsim}[1]{\sqsubset^{#1}}
\newcommand{\strictxsimlarg}[1]{\sqsupset^{#1}}
\newcommand{\tickNA}{\--}
\newcommand{\tickNO}{\times}
\newcommand{\tickOK}{\checkmark}  
\newcommand{\trans}[4]{\langle{#1},{#2},\vect{#3}{#4}\rangle}		

\newcommand{\tree}{t}

\newcommand{\uli}{\subseteq_\upl}    
\newcommand{\upl}{U}    
\newcommand{\upsim}[1]{\xsim{\mathsf{up}}\!\!\!\;\!\!(#1)}
\newcommand{\upsimlarg}[1]{\xsimlarg{\mathsf{up}}\!\!(#1)}
\newcommand{\uptraceinclusion}[1]{\xincl{\mathsf{up}}\!\!\!\;\!(#1)}
\newcommand{\uptraceinclusionlarg}[1]{\xincllarg{\mathsf{up}}\!\!\!\;\!(#1)}
\newcommand{\urel}{R_\mathsf{u}}		
\newcommand{\vect}[2]{#1_1\dotsc#1_{#2}}
\newcommand{\xincl}[1]{\subseteq^{#1}}
\newcommand{\xincllarg}[1]{\supseteq^{#1}}

\newcommand{\xsim}[1]{\sqsubseteq^{#1}}
\newcommand{\xsimlarg}[1]{\sqsupseteq^{#1}}
\newcommand{\xstrictincl}[1]{\subset^{#1}}
\newcommand{\xstrictincllarg}[1]{\supset^{#1}}
\newcommand{\xtraceinclusionlarg}[1]{\supseteq^{#1}}

\newcommand{\pspace}{\texttt{PSPACE}}
\newcommand{\exptime}{\texttt{EXPTIME}}

\newcommand{\libvata}{\texttt{libvata}}

\newcommand{\level}{\mathit{level}}
\newcommand{\mypar}[1]{{\bf\noindent{#1}}}
\newcommand{\ignore}[1]{}
\newcommand{\ru}{{\it RU}}

\title{Reduction of Nondeterministic Tree Automata
\thanks{This work was supported by the Czech Science Foundation, project 16-24707Y.}}
\titlerunning{Reduction of Nondeterministic Tree Automata}
\author{Ricardo Almeida\inst{1} \and Luk\'{a}\v{s} Hol\'{i}k\inst{2} \and Richard Mayr\inst{1}}

\institute{University of Edinburgh, UK \and Brno University of Technology, Czech Republic}

\date{}

\begin{document}
\maketitle

\begin{abstract}
We present an efficient algorithm to reduce the size of nondeterministic 
tree automata, while retaining their language.
It is based on new transition pruning techniques, 
and quotienting of the state space w.r.t.\ suitable equivalences.
It uses criteria based on combinations of downward and upward simulation
preorder on trees, and the more general downward and upward language inclusions.
Since tree-language inclusion is \exptime-complete, we describe
methods to compute good approximations in polynomial time.

We implemented our algorithm as a module of the well-known \libvata\ tree automata
library, and tested its performance on a given collection of tree automata
from various applications of \libvata\ in regular model checking and shape
analysis, as well as on various classes of randomly generated tree automata.
Our algorithm yields substantially smaller and sparser automata than all
previously known reduction techniques, and it is still fast enough to handle
large instances.
\end{abstract}

\section{Introduction}\label{sec:introduction}

\mypar{Background.}
Tree automata are a generalization of word automata that accept 
trees instead of words \cite{tata2008}.
They have many applications in 
model checking
\cite{abdulla:tree,Abdulla:transducers05,Bouajjani:modelChecking2006}, 
term rewriting \cite{tool:Autowrite}, and related areas of formal software
verification, e.g., shape analysis \cite{abdulla:heapForest2013,Holik:shapeForest2013,Habermehl:forestHeap2011}.
Several software packages for manipulating tree automata have been developed, e.g.,
MONA \cite{tool:MONA}, Timbuk \cite{tool:Timbuk}, Autowrite
\cite{tool:Autowrite}
and \libvata\ \cite{tool:libvata},
on which other verification tools like Forester \cite{tool:forester} are based.

For nondeterministic automata, many questions about their languages
are computationally hard. The language universality, equivalence and inclusion
problems are \pspace-complete for word automata and \exptime-complete for tree
automata \cite{tata2008}.
However, recently techniques have been developed that can solve many practical
instances fairly efficiently. 
For word automata there are antichain techniques \cite{Abdulla:whensimulation2010},
congruence-based techniques \cite{Bonchi:bisimCongr2013} and techniques based on generalized simulation
preorders \cite{mayr:advanced2013}.
The antichain techniques have been generalized to tree automata in 
\cite{Bouajjani:antichain2008,holik:efficient2011} and 
implemented in the \libvata\ library \cite{tool:libvata}.
Performance problems also arise in computing the intersection of
several languages, since the product construction multiplies the numbers of
states.

\mypar{Automata Reduction.}
Our goal is to make tree automata more computationally tractable in practice.
We present an efficient
algorithm for the reduction of nondeterministic tree automata, in the sense of 
obtaining a smaller automaton with the same language, though not 
necessarily with the absolute minimal possible number of states. 
(In general, there is no unique nondeterministic automaton with the minimal
possible number of states for a given language, i.e., there can be several
non-isomorphic nondeterministic automata of minimal size. This holds even 
for word automata.)
The reason to perform reduction is that the smaller reduced
automaton is more efficient to handle in a subsequent computation.
Thus there is an algorithmic tradeoff between the effort for 
reduction and the complexity of the problem later considered for
this automaton. 
The main applications of reduction are the following:
(1) Helping to solve hard problems like language
universality/equivalence/inclusion.
(2) If automata undergo a long chain of manipulations/combinations by operations
like union, intersection, projection, etc., then intermediate 
results can be reduced several times on the way
to keep the automata within a manageable size.
(3) There are fixed-parameter tractable problems (e.g., in model checking
where an automaton encodes a logic formula) where the size of one automaton
very strongly influences the overall complexity, and must be kept as small as
possible.

\mypar{Our contribution.}
We present a reduction algorithm for nondeterministic tree automata.
(The tool is available for download \cite{tool:libvata-heavy}.)
It is based on a combination of new transition pruning techniques for
tree automata, and quotienting of the state space w.r.t.\ suitable
equivalences.
The pruning techniques are related to those presented for word automata in 
\cite{mayr:advanced2013}, but significantly more complex due to the
fundamental asymmetry between the upward and downward directions in trees.

Transition pruning in word automata \cite{mayr:advanced2013} is based on the
observation that certain transitions can be removed (a.k.a pruned) without changing the
language, because other `better' transitions remain.
One defines some strict partial order (p.o.) between transitions 
and removes all transitions that are not maximal 
w.r.t.\ this order.
A strict p.o.\ between transitions is called \emph{good for pruning} (GFP)
iff pruning w.r.t.\ it preserves the language of the automaton.
Note that pruning reduces not only the number of transitions, but also,
indirectly, the number of states. By removing transitions, some states may
become `useless', in the sense that they are unreachable from any initial
state, or that it is impossible to reach any accepting state from them.
Such useless states can then be removed from the automaton without changing
its language.
One can obtain computable strict p.o.\ between transitions 
by comparing the possible backward- and forward behavior of their source- and target states,
respectively.
For this, one uses
computable relations like backward/forward simulation preorder and 
approximations of backward/forward trace inclusion via lookahead- or multipebble
simulations. Some such combinations of backward/forward trace/simulation
orders on states induce strict p.o.\ between transitions that are
GFP, while others do not \cite{mayr:advanced2013}.
However, there is always a symmetry between backward and forward, since finite
words can equally well be read in either direction.

This symmetry does not hold for tree automata, because the tree branches as
one goes downward, while it might `join in' side branches as one goes upward.
While downward simulation preorder (resp.\ downward language inclusion) between
states in a tree automaton is a direct generalization of forward simulation
preorder (resp.\ forward language inclusion) on words, the corresponding
upward notions do not correspond to backward on words.
Comparing upward behavior of states in tree automata depends also
on the branches that `join in' from the sides as one goes upward in the tree.
Thus upward simulation/language inclusion is only defined {\em relative} to a given
other relation that compares the downward behavior of states `joining in' from
the sides \cite{Holik:computingSim08}. So one speaks of 
``upward simulation {\em of} the identity relation'' or
``upward simulation {\em of} downward simulation''.
When one studies strict p.o.\ between transitions in tree automata in order to
check whether they are GFP, one has combinations of three 
relations: the source states are compared by an upward relation $X(Y)$ of some
downward relation $Y$, while the target states are compared w.r.t.\ some
downward relation $Z$ (where $Z$ can be, and often must be, different from $Y$).
This yields a richer landscape, and many
counter-intuitive effects.

We provide a complete picture of which combinations of upward/downward
simulation/trace inclusions are GFP on tree automata; 
cf.\ Figure~\ref{fig:GFP_relations_trees}.
Since tree-(trace)language inclusion is \exptime-complete \cite{tata2008}, we describe
methods to compute good approximations of them in polynomial time.
Finally, we also generalize results on quotienting of tree automata
\cite{Lukas:PhD13} to larger relations, such as approximations of trace inclusion.

We implemented our algorithm \cite{tool:libvata-heavy} 
as a module of the well-known 
\libvata\ \cite{tool:libvata} tree automaton
library, and tested its performance on a given collection of tree automata
from various applications of \libvata\ in regular model checking and shape
analysis, as well as on various classes of randomly generated tree automata.
Our algorithm yields substantially smaller automata than all
previously known reduction techniques (which are mainly based on
quotienting). Moreover, the thus obtained automata are also much sparser (i.e., use fewer
transitions per state and less nondeterministic branching) 
than the originals, which yields additional performance
advantages in subsequent computations.

\section{Trees and Tree Automata} 

\mypar{Trees.}
A \emph{ranked alphabet} $\Sigma$ is a set
of symbols together with a function $\rank : \Sigma \rightarrow \nat_0$. For $\f
\in \Sigma$, $\rank{(\f)}$ is called the \emph{rank} of $\f$. For
$n \geq 0$, we denote by $\Sigma_n$ the set of all symbols of $\Sigma$ which have rank $n$. 

We define a \emph{node} as a sequence of elements of $\nat$,
where $\varepsilon$ is the empty sequence.
For a node $\node \in \nat^*$, 
any node $\node'$ s.t. $\node = \node' \node''$,
for some node $\node''$, is said to be a \emph{prefix} of $\node$,
and if $\node'' \neq \varepsilon$ then $\node'$ is a \emph{strict prefix} of $\node$.
For a node $\node \in \nat^*$, we define the $i$-th child of $v$ to be the node $vi$, for some $i \in \nat$.
Given a ranked alphabet $\Sigma$, 
a \emph{tree} over $\Sigma$ is defined as a partial mapping $\tree: \nat^* \rightarrow \Sigma$ such that 
for all $v \in \nat^*$ and $i \in \nat$, 
if $vi \in \dom(\tree)$ then 
\textbf{(1)} $\node\in\dom(\tree)$, and
\textbf{(2)} $\rank(\tree(\node))\geq i$.
In this paper we consider only finite trees.

Note that the number of children of a node $v$ 
may be smaller than $\#(t(v))$.
In this case we say that the node is \emph{open}. 
Nodes which have exactly $\#(t(v))$ children are called \emph{closed}.
Nodes which do not have any children are called \emph{leaves}.
A~tree is closed if all its nodes are closed, otherwise it is open.
By $\allcompletetrees(\Sigma)$ we denote the set of all closed trees over $\Sigma$
and by $\alltrees(\Sigma)$ the set of all trees over $\Sigma$.
A~tree $\tree$ is \emph{linear} iff every node in $\dom(\tree)$ has at most one child.

The \emph{subtree} of a tree $\tree$ at $\node$ is defined as the tree $\tree_\node$ 
such that $\dom(\tree_\node) = \{\node'\mid \node\node'\in\dom(\tree)\}$ 
and $\tree_\node(\node')=\tree(\node \node')$ for all $\node' \in \dom(\tree_\node)$. 
A~tree $\tree'$ is a prefix of $\tree$ iff $\dom(\tree')\subseteq\dom(\tree)$ and 
for all $\node\in\dom(\tree')$, $\tree'(\node) = \tree(\node)$.
For $\tree \in \allcompletetrees(\Sigma)$, 
the \emph{height of a node} $\node$ of $\tree$ is given by the function $h$:
if $\node$ is a leaf then $h(\node) = 1$, otherwise $h(\node) = 1 + max(h(\node 1)), \ldots, h(\node \#(\tree(\node))))$.
We define the height of a tree $\tree \in \allcompletetrees(\Sigma)$ as $h(\epsilon)$,
i.e., as the number of levels of $\tree$.

\mypar{Tree automata, top-down.}
A (finite, nondeterministic) \emph{top-down tree automaton} (TDTA) is a quadruple 
$\A = (\Sigma,Q,\delta, I)$ where $Q$ is a finite set of states, $I \subseteq Q$ is a
set of initial states, $\Sigma$ is a ranked alphabet, 
and $\delta \subseteq Q \times \Sigma\times Q^+ $ is the set of transition rules. 
A TDTA has an unique final state, which we represent by $\spstate$.
The transition rules satisfy that if $\langle q, \f, \spstate \rangle \in\delta$ then $\rank(\f) = 0$, 
and if $\langle q, \f, q_1 \ldots q_n \rangle \in \delta$ (with $n > 0$) then $\rank(\f) = n$.

A \emph{run} of $\A$ over a tree $\tree \in \alltrees(\Sigma)$ (or a $\tree$-run in $\A$) is a partial mapping 
$\run : \nat^*\rightarrow Q$ such that $\node\in\dom(\run)$ iff 
either $\node\in\dom(\tree)$ or $\node = \node'i$ where $\node'\in\dom(\tree)$ and $i \leq \rank(\tree(\node'))$.  
Further, for every $\node\in\dom(\tree)$, there exists either
\textbf{a)} a rule $\langle q, a, \spstate \rangle$ such that $q = \run(\node)$ and $\f = \tree(\node)$,
or \textbf{b)} a rule $\transTD q \f q n$ such that $q = \run(\node)$, 
$\f = \tree(\node)$, and $q_i = \run(\node i)$ for each $i:1\leq i\leq \rank(\f)$.
A \emph{leaf of a run} $\run$ on $\tree$ is a node $\node\in\dom(\run)$ 
such that $\node i\in\dom(\run)$ for no $i \in\nat$.
We call it \emph{dangling} if $\node\not\in\dom(\tree)$. 
Intuitively, the dangling nodes of a run over $t$ are all the nodes which are in $\run$ but are missing 
in $t$ due to it being incomplete.
Notice that dangling leaves of $\pi$ are children of open nodes of $\tree$. 
The prefix of depth $k$ of a 
run $\run$ is denoted $\run_k$. Runs are always finite since the trees we are considering are finite.

We write $\tree\stackrel{\run}{\Longrightarrow} q$ to denote
that $\run$ is a $t$-run of $\A$ such that $\run(\epsilon)=q$. 
We use $\tree\Longrightarrow q$ to denote that such run $\run$ exists.
A run $\run$ is accepting if $\tree\stackrel{\run}{\Longrightarrow} q \in I$.
The \emph{downward language of a state} $q$ in $\A$ is defined by
$D_A(q)=\{\tree\in \allcompletetrees(\Sigma)\mid \tree\Longrightarrow q\}$, 
while the \emph{language} of $\A$ is defined by $L(\A)=\bigcup_{q\in I}D_A(q)$. 
The \emph{upward language} of a state $q$ in $\A$, denoted $\upl_A(q)$, is then defined as the set of open trees $\tree$,
such that there exists an accepting $\tree$-run $\run$ with exactly one dangling leaf $\node$ s.t. 
$\run(\node) = q$.
We omit the $\A$ subscript notation when it is implicit which automaton we are considering.

%

In the related literature,
it is common to define a tree automaton bottom-up,
reading a tree from the leaves to the root~\cite{tata2008,Bouajjani:antichain2008,holik:efficient2011}.
A bottom-up tree automaton (BUTA) can be obtained from a TDTA by reversing the direction of the transition rules
and by swapping the roles between the initial states and the final states. 
See Appendix~\ref{sec:appendix_A} for an example of a tree automaton 
presented in both BUTA and TDTA form.

\section{Simulations and Trace Inclusions}\label{sec:simulations}

We consider different types of relations on states of a TDTA which under-approximate language inclusion.
Note that words are but a special case of trees where every node has only one
child, i.e., words are linear trees.
\emph{Downward} simulation/trace inclusion on TDTA
corresponds to \emph{direct forward} simulation/trace inclusion in special
case of word automata, 
and \emph{upward} corresponds to \emph{backward}~\cite{mayr:advanced2013}.

\mypar{Forward simulation on word automata.}
Let $\A = (\Sigma, Q, \delta, I, F)$ be a NFA.
A \emph{direct forward simulation} $\mathrel{D}$ is a binary relation on $Q$ such that if $q \mathrel{D} r$, 
then 
\begin{enumerate}
 \item $q \in F \Longrightarrow r \in F$, and
 \item for any $\langle q, a, q' \rangle \in \delta$,
 there exists $\langle r, a, r' \rangle \in \delta$ such that $q' \mathrel{D} r'$.
\end{enumerate}
The set of direct forward simulations on $\A$ contains $\id$ and is closed
under union and transitive closure.
Thus there is a unique maximal direct forward simulation on $\A$,
which is a preorder.
We call it \emph{the direct forward simulation preorder on $\A$} and write $\disim$.

\mypar{Forward trace inclusion on word automata.}
Let $\A = (\Sigma, Q, \delta, I, F)$ be a NFA
and $w=\sigma_1\sigma_2 \ldots \sigma_n \in \Sigma^*$ a word of length $n$.
A trace of $\A$ on $w$ (or a $w$-trace) starting at $q$ is a sequence of transitions 
$\run = q_0 \stackrel{\sigma_1}{\rightarrow} q_1 \stackrel{\sigma_2}{\rightarrow} \cdots \stackrel{\sigma_n}{\rightarrow} q_n$
such that $q_0=q$.
The \emph{direct forward trace inclusion} preorder $\xincl{di}$ is a binary relation on $Q$ such that $q \xincl{di} r$ iff
\begin{enumerate}
 \item $(q \in F \Longrightarrow r \in F)$, and
 \item for every word $w=\sigma_1\sigma_2 \ldots \sigma_n \in \Sigma^*$
and for every $w$-trace (starting at $q$) 
\\ $\run_q = q \stackrel{\sigma_1}{\rightarrow} q_1 \stackrel{\sigma_2}{\rightarrow} \cdots \stackrel{\sigma_n}{\rightarrow} q_n$,
there exists a $w$-trace (starting at $r$) 
$\run_r = r \stackrel{\sigma_1}{\rightarrow} r_1 \stackrel{\sigma_2}{\rightarrow} \cdots \stackrel{\sigma_n}{\rightarrow} r_n$
such that ($q_i \in F \Longrightarrow r_i \in F$) for each $i:1 \leq i \leq n$.
\end{enumerate}
Since $\run_r$ is required to preserve the acceptance of the states in $\run_q$, 
trace inclusion is a strictly stronger notion than language inclusion
(see Figure~\ref{fig:TrInclExample} in Appendix~\ref{sec:appendix_A}).
\\
\mypar{Downward simulation on tree automata.}
Let $\A = (\Sigma,Q,\delta, I)$ be a TDTA.
A \emph{downward simulation} $D$ is a binary relation on $Q$ such that if $q \mathrel{D} r$, then 
\begin{enumerate}
 \item $(q=\spstate \Longrightarrow r=\spstate)$, and 
 \item for any $\trans q \f q n \in \delta$, 
  there exists $\trans r \f r n \in \delta$ s.t.\ $q_i \,D\, r_i$ for $i:1\leq i\leq n$.
\end{enumerate}
Since the set of all downward simulations on $\A$ is closed under union and under reflexive and transitive closure 
(cf.~Lemma 4.1 in \cite{Lukas:PhD13}),
it follows that there is one unique maximal downward simulation on $\A$,
and that relation is a preorder. 
We call it \emph{the downward simulation preorder on $\A$} and write $\dwsim$.


\mypar{Downward trace inclusion on tree automata.} 
Let $\A = (\Sigma,Q,\delta,I)$ be a TDTA.
The \emph{downward trace inclusion} preorder $\dwtraceinclusion$ is a binary
relation on $Q$ s.t. $q \dwtraceinclusion r$ iff
for every tree $\tree \in \allcompletetrees(\Sigma)$ and for every $\tree$-run 
$\run_q$ with $\run_q(\epsilon) = q$ there exists another $\tree$-run $\run_r$
s.t.
\begin{enumerate}
 \item  $\run_r(\epsilon) = r$, and
 \item ($\run_q(\node) = \spstate \,\Longrightarrow\, \run_r(\node) = \spstate$) 
 for each leaf node $\node \in dom(\tree)$.
\end{enumerate}
Generally, one way of making downward language inclusion on the states of an automaton coincide with 
downward trace inclusion is by modifying the automaton to guarantee that 
\textbf{1)} there is one unique final state which has no outgoing transitions, 
\textbf{2)} from any other state, there is a path ending in that final state. 
Note that in a TDTA these two conditions are automatically satisfied: 
\textbf{1)} since the final state is reached after reading a leaf of the tree,
and \textbf{2)} because only complete trees are in the language of the automaton. 
Thus, in a TDTA, downward language inclusion and downward trace inclusion coincide. 

\mypar{Backward simulation on word automata.}
Let $\A = (\Sigma, Q, \delta, I, F)$ be a NFA.
A \emph{backward simulation} $\mathrel{B}$ is a binary relation on $Q$ s.t. if $q \mathrel{B} r$,
then 
\begin{enumerate}
\item ($q \in F \,\Longrightarrow\, r \in F$) and ($q \in I \,\Longrightarrow\, r \in I$), and
\item for any $\langle q', a, q \rangle \in \delta$, 
there exists $\langle r', a, r \rangle \in \delta$ s.t. $q' \mathrel{B} r'$.
\end{enumerate}
Like for forward simulation,
there is a unique maximal backward simulation on $\A$,
which is a preorder.
We call it \emph{the backward simulation preorder on $\A$} and write $\bwsim$.

\mypar{Backward trace inclusion on word automata.}
Let $\A = (\Sigma, Q, \delta, I, F)$ be a NFA
and $w=\sigma_1\sigma_2 \ldots \sigma_n \in \Sigma^*$ a word of length $n$.
A $w$-trace of $\A$ ending at $q$ is a sequence of transitions 
$\run = q_0 \stackrel{\sigma_1}{\rightarrow} q_1 \stackrel{\sigma_2}{\rightarrow} \cdots \stackrel{\sigma_n}{\rightarrow} q_n$
such that $q_n=q$.
The \emph{backward trace inclusion} preorder $\bwdirecttraceinclusion$ is a binary relation on $Q$
such that $q \bwdirecttraceinclusion r$ iff 
\begin{enumerate}
 \item ($q \in F \Longrightarrow r \in F$) and ($q \in I \Longrightarrow r \in I$), and
 \item for every word $w=\sigma_1\sigma_2 \ldots \sigma_n \in \Sigma^*$
and for every $w$-trace (ending at $q$) 
$\run_q = q_0 \stackrel{\sigma_1}{\rightarrow} q_1 \stackrel{\sigma_2}{\rightarrow} \cdots \stackrel{\sigma_n}{\rightarrow} q$,
there exists a $w$-trace (ending at $r$) 
$\run_r = r_0 \stackrel{\sigma_1}{\rightarrow} r_1 \stackrel{\sigma_2}{\rightarrow} \cdots \stackrel{\sigma_n}{\rightarrow} r$
such that ($q_i \in F \Longrightarrow r_i \in F \,\land\; q_i \in I \Longrightarrow r_i \in I$) 
for each $i:1 \leq i \leq n$.
\end{enumerate}

\mypar{Upward simulation on tree automata.}
Let $\A = (\Sigma,Q,\delta, I)$ be a TDTA.
Given a binary relation $\mathrel{R}$ on $Q$, an \emph{upward simulation} $U(R)$ induced by $\mathrel{R}$ is 
a binary relation on $Q$ such that if $q \mathrel{U(R)} r$, then 
\begin{enumerate}
\item $(q=\spstate \Longrightarrow r=\spstate)$ and $(q\in I \Longrightarrow r\in I)$, and
\item for any $\trans {q'} \f q n \in \delta$ with $q_i=q$ (for some $i:1 \leq i \leq n$), there exists 
\\ $\trans {r'} \f r n \in \delta$ such that $r_i=r$, $q'  \mathrel{U(R)} r'$ and $q_j \mathrel{R} r_j$ for 
 each $j:1\leq j\neq i\leq n$.
\end{enumerate}
Similarly to the case of downward simulation, 
for any given relation $R$, 
there is a unique maximal upward simulation induced by $R$ which is
a preorder (cf.~Lemma 4.2 in \cite{Lukas:PhD13}).
We call it \emph{the upward simulation preorder on $\A$ induced by $R$} and write $\upsim{R}$.



\mypar{Upward trace inclusion on tree automata.} 
Let $\A = (\Sigma, Q, \delta, I)$ be a TDTA.
Given a binary relation $\mathrel{R}$ on $Q$,
the \emph{upward trace inclusion preorder $\uptraceinclusion{R}$ induced by $\mathrel{R}$} 
is a binary relation on $Q$ such that $q \uptraceinclusion{\mathrel{R}}\, r$ 
iff $(q = \spstate \Longrightarrow r = \spstate)$ and the following holds:
for every tree $\tree \in T(\Sigma)$ and for every $\tree$-run $\run_q$ with $\run_q(v) = q$
for some leaf $\node$ of $\tree$,
there exists a $\tree$-run $\run_r$ s.t. 
\begin{enumerate}
 \item $\run_r(\node) = r$,
 \item for all prefixes $\node'$ of $\node$,
 $(\run_q(\node') \in I \Longrightarrow \run_r(\node') \in I)$, and
 \item if $\node'x \in \dom(\run_q)$,
 for some strict prefix $\node'$ of $\node$ and some $x \in \nat$ s.t. $\node'x$ is not a prefix of $\node$,
 then $\run_q(\node'x) \mathrel{R}\, \run_r(\node'x)$.
\end{enumerate}

\noindent
Downward trace inclusion is \exptime-complete for TDTA
\cite{tata2008},
while forward trace inclusion is \pspace-complete for word automata.
The complexity of upward trace
inclusion depends on the relation $R$ (e.g., it is \pspace-complete for $R=\id$).
In contrast, downward/upward simulation preorder is computable in polynomial
time \cite{Holik:computingSim08}, but typically yields only small under-approximations of the
corresponding trace inclusions.

\section{Transition Pruning Techniques}	\label{sec:transPruning}

We define pruning relations on a TDTA $\A = (\Sigma,Q,\delta, I)$.
The intuition is that certain transitions may be deleted without changing the
language, because `better' transitions remain.
We perform this pruning (i.e., deletion) of transitions by comparing their endpoints over the same symbol $\sigma \in \Sigma$.
Given two binary relations $\urel$ and $\drel$ on $Q$, we define the following
relation to compare transitions.
$$P(\urel,\drel) = \{ (\langle p, \sigma, r_1 \cdots r_n \rangle , \langle p', \sigma, r'_1 \cdots r'_n \rangle) \, \mid \, p \mathrel{\urel} p' \mbox{ and } (r_1 \cdots r_n) \mathrel{\ldrel} (r'_1 \cdots r'_n) \},$$
where $\ldrel$ results from lifting $\drel \subseteq Q \times Q$ to $\ldrel
\subseteq Q^n \times Q^n$, as defined below.
The function $P$ is monotone in the two arguments.
If $t\, P \,t'$ then $t$ may be pruned because $t'$ is `better' than $t$.
We want $P(\urel,\drel)$ to be a strict partial order (p.o.), i.e., 
irreflexive and transitive (and thus acyclic).
There are two cases in which $P(\urel,\drel)$ is guaranteed to be a strict p.o.:
\textbf{1)} $\urel$ is some strict p.o. $<_\mathsf{u}$ and 
$\ldrel$ is the standard lifting $\hat{\leq}_\mathsf{d}$ of some
p.o. $\leq_\mathsf{d}$ to tuples.
I.e., $(r_1 \cdots r_n) \hat{\leq}_\mathsf{d} (r'_1 \cdots r'_n)$ iff
$\forall_{1 \leq i \leq n} \ldotp r_i \leq_\mathsf{d} r'_i$.
The transitions in each pair of $P(<_\mathsf{u},{\leq}_\mathsf{d})$ depart
from different states and therefore the transitions are necessarily different.
\textbf{2)} $\urel$ is some p.o. $\leq_\mathsf{u}$ and $\ldrel$ is the
lifting $\hat{<}_\mathsf{d}$ of 
some strict p.o. $<_\mathsf{d}$ to tuples (defined below).
In this case the transitions in each pair of
$P(\leq_\mathsf{u},{<}_\mathsf{d})$ may have the same origin but 
must go to different tuples of states.  
Since for two tuples $(r_1\cdots r_n)$ and $(r'_1\cdots r'_n)$ to be different it suffices that $r_i \neq r'_i$ for 
some $1 \leq i \leq n$,
we define $\hat{<}_\mathsf{d}$ as a binary relation such that $(r_1 \cdots r_n) \hat{<}_\mathsf{d} (r'_1 \cdots r'_n)$
iff $\forall_{1 \leq i \leq n} \ldotp r_i \leq_\mathsf{d} r'_i$, and
$\exists_{1 \leq i \leq n} \ldotp r_i <_\mathsf{d} r'_i$.

Let $\A = (\Sigma,Q,\delta, I)$ be a TDTA and let $P \subseteq \delta \times \delta$ 
be a strict partial order.
The pruned automaton is defined as $Prune(\A,P) = (\Sigma,Q,\delta', I)$ where
$\delta' = \{ (p,\sigma,r) \in \delta \mid \nexists (p',\sigma,r') \in \delta \ldotp (p,\sigma,r) \, P \, (p',\sigma,r') \} $.
Note that the pruned automaton $Prune(A,P)$ is unique.
The transitions are removed without requiring the re-computation of the relation $P$, 
which could be expensive. 
Since removing transitions cannot introduce new trees in the language,
$L(Prune(\A,P)) \subseteq L(\A)$. 
If the reverse inclusion holds too (so that the language is preserved),
we say that $P$ is \emph{good for pruning} (GFP),
i.e., $P$ is GFP iff $L(Prune(A,P)) = L(A)$.

We now provide a complete picture of which combinations of
simulation and trace inclusion relations are GFP. Recall that simulations are
denoted by square symbols $\sqsubseteq$ while trace inclusions are denoted by
round symbols $\subseteq$. For every partial order $R$, the corresponding
strict p.o.\ is defined as $R\backslash R^{-1}$. 

$P(\strictbwdirecttraceinclusion, \strictdirecttraceinclusion)$ is not GFP for
word automata (see Fig.~2(a) in \cite{mayr:advanced2013} for a counterexample).
As mentioned before, words correspond to linear trees.
Thus $P(\strictuptraceinclusion{R}, \strictdwtraceinclusion)$ is not GFP for tree automata
(regardless of the relation $R$).
Figure~\ref{fig:GFP_counterexamples} presents several more counterexamples.
For word automata, 
$P(\strictbwdirecttraceinclusion,\disim)$ and
$P(\bwsim,\strictdirecttraceinclusion)$ are not GFP
(Fig.~\ref{GFPcountEx:strictbwtraceincldisim} and \ref{GFPcountEx:bwsimstrictdirecttraceinclusion}) 
even though $P(\bwdirecttraceinclusion,\strictdisim)$ 
and $P(\strictbwsim,\directtraceinclusion)$ are (cf.~\cite{mayr:advanced2013}).
Thus $P(\strictuptraceinclusion{R},\dwsim)$ 
and $P(\upsim{R},\strictdwtraceinclusion)$ are not GFP for tree automata
(regardless of the relation $R$).
For tree automata, $P(\strictupsim{\strictdwsim},\id)$ and
$P(\strictupsim{\strictdwtraceinclusion},\strictdwsim)$
are not GFP (Fig.~\ref{GFPcountEx:strictupsim(dwsim)id} and \ref{GFPcountEx:strictupsim(dwtraceinclusion)strictdwsim}).
Moreover, a complex counterexample (see Fig.~\ref{GFPcountEx:strictupsim(dwsim)strictdwtraceincl}; App.~\ref{sec:appendix_A})
is needed to show that
$P(\strictupsim{\strictdwsim},\strictdwtraceinclusion)$ is not GFP.

\begin{small}
\begin{figure}[htbp]
\begin{center}
\subfloat[
$P(\strictupsim{\strictdwsim},\id)$ is not GFP: if we remove the blue transitions, 
	the automaton no longer accepts the tree $a(c,d)$.
	We are considering $\Sigma_0 = \{c,d\}$, $\Sigma_1 = \{b\}$ and $\Sigma_2 = \{a\}$.
]{
\begin{tikzpicture}[on grid, node distance=1cm and 1.5cm]
		\tikzstyle{vertex} = [smallstate] 

		\path node [vertex,initial] (q1)  {};
		
		\path node [vertex] (q4) [below left = 2.2 and 0.3 of q1] {};
		\path node [vertex] (q5) [right = 0.7 of q4] {};
		\path node [vertex] (q3) [left = 2.5cm of q4] {};
		\path node [vertex] (q6) [right = 2.5cm of q5] {};
				
		\path node [vertex,accepting] (f1) [below right = 1.8 and 1.2cm of q3] {$\spstate$};
		\path node [vertex,accepting] (f2) [below left = 1.8 and 1.2cm of q6] {$\spstate$};

 		\path[->]
 			
 			(q1.south) edge node [below left] {} (q3.north)
 			(q1.south) edge node [above left] {} (q6.north)
 			(q1.south) edge [dashed] node [below left] {} (q4)
 			(q1.south) edge [dashed] node [above left] {} (q5)
 			
 			(q3) edge [bend left] node [above] {$b$} (q4)
 			
 			(q3) edge [blue,very thin] node [left] {$c$} (f1)
 			(q4) edge [red,very thick] node [right] {$c$} (f1)
  			(q5) edge [blue,very thin] node [left] {$d$} (f2)
  			(q6) edge [red,very thick] node [right] {$d$} (f2)
  			
  			(q5) edge [bend left] node [above] {$b$} (q6);
  			
  		\path
			(q1) -- node [below = 1.2 ] {$a$} (q1)
			(q1) -- node [below left = 0.7 and 0.1cm] {$a$} (q1)
						
			(q3) -- node [above = 0.1cm of q3] {$\strictupsim{\strictdwsim}$} (q4)
			(q3) -- node [below = 0.3cm of q3] {$\strictdwsimlarg$} (q4)
			(q5) -- node [above = 0.1cm of q3] {$\strictupsim{\strictdwsim}$} (q6)
			(q5) -- node [below = 0.3cm of q3] {$\strictdwsimlarg$} (q6);
  		
		\begin{pgfonlayer}{background}
		\end{pgfonlayer}

	\end{tikzpicture}
	\label{GFPcountEx:strictupsim(dwsim)id}
}
\end{center}		
\subfloat[$P(\strictbwdirecttraceinclusion,\disim)$ is not GFP for words: 
	    if we remove the blue transitions, the automaton no longer accepts the word $aaa$.]
{
	\begin{tikzpicture}[on grid, node distance=1cm and 1.5cm]
		\tikzstyle{vertex} = [smallstate] 

		\path node [vertex,initial] (i) {};
		\path node [vertex] (p2) [right = of i] {};
		\path node [vertex] (p1) [above = of p2] {};
		\path node [vertex] (p3) [below = 1.2cm of p2] {};
		\path node [vertex] (q1) [right = of p2] {};
		\path node [vertex] (q2) [right = of p3] {};
		\path node [vertex,accepting] (r1) [right = of q1] {};
		\path node [vertex,accepting] (r2) [right = of q2] {};
				
		\path[->]
 			(i) edge node [above left] {$b,c$} (p1)
 			(i) edge node [above] {$a$} (p2)
 			(i) edge node [below left] {$a,b$} (p3)
 			
 			(p1) edge node [above] {$a$} (q1)
 			(p2) edge [blue,very thin] node [below] {$a$} (q1)
 			(p3) edge [red,very thick] node [below] {$a$} (q2)
 			(q1) edge [red,very thick] node [below] {$a$} (r1)
 			(q2) edge [blue,very thin] node [below] {$a$} (r2);
 			
 		\path
			(p2) -- node {\rotatebox{-90}{$\strictbwdirecttraceinclusion$}} (p3)
			(q1) -- node [right = 0.2cm of q1] {\rotatebox{90}{$\strictbwdirecttraceinclusion$}} (q2)
			(q1) -- node [left  = 0.2cm of q1] {\rotatebox{-90}{$\disim$}} (q2)
			(r1) -- node {\rotatebox{90}{$\disim$}} (r2);
		
	\end{tikzpicture}
	\label{GFPcountEx:strictbwtraceincldisim}
}		\quad
\subfloat[$P(\bwsim,\strictdirecttraceinclusion)$ is not GFP for words: 
	    if we remove the blue transitions, the automaton no longer accepts the word $aaa$.]
{
	\begin{tikzpicture}[on grid, node distance=1cm and 1.5cm]
		\tikzstyle{vertex} = [smallstate] 

		\path node [vertex,initial] (i) {};
		\path node [vertex] (p2) [right = of  i] {};
		\path node [vertex] (p1) [above = 1.2cm of p2] {};
		\path node [vertex] (q1) [right = of p1] {};
		\path node [vertex] (q2) [right = of p2] {};
		\path node [vertex] (q3) [below = of q2] {};
		\path node [vertex,accepting] (f1) [right = of q1] {};
		\path node [vertex,accepting] (f2) [right = of q2] {};
		\path node [vertex,accepting] (f3) [right = of q3] {};
		
		\path[->]
 			(i) edge [blue,very thin] node [above left] {$a$} (p1)
 			(i) edge [red,very thick] node [above] {$a$} (p2)
 			(p1) edge [red,very thick] node [above] {$a$} (q1)
 			(p2) edge [blue,very thin] node [above] {$a$} (q2)
 			(p2) edge node [below left] {$a$} (q3)
 			(q1) edge node [above] {$a,b$} (f1)
 			(q2) edge node [above] {$a$} (f2)
 			(q3) edge node [above] {$b,c$} (f3);
 			
 		\path
			(i) -- node [below = 0.8cm of i] {\rotatebox{-90}{$\bwsim$}} (i)
			(p1) -- node [left = 0.2cm of p1]  {\rotatebox{90}{$\bwsim$}} (p2)
			(p1) -- node [right = 0.2cm of p1] {\rotatebox{-90}{$\strictdirecttraceinclusion$}} (p2)
			(q1) -- node [] {\rotatebox{90}{$\strictdirecttraceinclusion$}} (q2);
				
	\end{tikzpicture}
	\label{GFPcountEx:bwsimstrictdirecttraceinclusion}
}

\subfloat[
$P(\strictupsim{\strictdwtraceinclusion},\strictdwsim)$ is not GFP: if we remove the blue transitions, 
	the tree $a(a(c,c),a(c,c))$ is no longer accepted.
	We are considering $\Sigma_0 = \{c,d\}$, $\Sigma_1 = \{b\}$ and $\Sigma_2 = \{a\}$.
]{
\begin{tikzpicture}
		\tikzstyle{vertex} = [smallstate] 

		\path node [vertex,initial] (q1)  {};
		
		\path node [vertex] (q4) [below left = 1.6 and 0.5cm of q1] {};
		\path node [vertex] (q5) [right = 1.3cm of q4] {};
		\path node [vertex] (q3) [left = 1.8cm of q4] {};
		\path node [vertex] (q6) [right = 1.8cm of q5] {};
		
		\path node [vertex] (p3) [below left = 2 and 0.2cm of q3] {};
		\path node [vertex] (p2) [left = 0.8cm of p3] {};
		\path node [vertex] (p1) [left = 0.8cm of p2] {};
		\path node [vertex] (p4) [below right = 2 and 0.8cm of q3] {};
		\path node [vertex] (p5) [right = 0.8cm of p4] {};
		\path node [vertex] (p7) [below right = 2 and 0.1cm of q5 ] {};
		\path node [vertex] (p6) [left = 0.8cm of p7] {};
		\path node [vertex] (p8) [below right = 2 and 1.2cm of q5] {};
		\path node [vertex] (p9) [right = 0.8cm of p8] {};
		\path node [vertex] (p10) [right = 0.8cm of p9] {};
		
		\path node [vertex,accepting] (f1) [below = 0.9cm of p2] {$\spstate$};
		\path node [vertex,accepting] (f2) [below right = 0.9 and 0.2cm of p4] {$\spstate$};
		\path node [vertex,accepting] (f3) [below right = 0.9 and 0.2cm of p6] {$\spstate$};
		\path node [vertex,accepting] (f4) [below = 0.9cm of p9] {$\spstate$};
		
 		\path[->]
 			(q1.south) edge node [below left] {} (q3.north)
 			(q1.south) edge node [above left] {} (q6.north)
 			(q1.south) edge [dashed] node [below left] {} (q4)
 			(q1.south) edge [dashed] node [above left] {} (q5)
 			
 			(q3.south) edge [blue, very thin] node {} (p1.north)
 			(q3.south) edge [blue, very thin] node {} (p4.north)
 			(q3.south) edge [dashed,blue,very thin] node {} (p3)
 			(q3.south) edge [dashed,blue,very thin] node {} (p4.west)
 			(q4.south) edge [red,very thick] node {} (p2)
  			(q4.south) edge [red,very thick] node {} (p5)
  			
  			(q3) edge [bend left] node [above right] {$b$} (q4)
  			
  			(q5.south) edge [blue,very thin] node {} (p6.north)
 			(q5.south) edge [blue,very thin] node {} (p10.north)
 			(q5.south) edge [dashed,blue,very thin] node {} (p6.east)
 			(q5.south) edge [dashed,blue,very thin] node {} (p8)
 			(q6.south) edge [red,very thick] node {} (p7)
  			(q6.south) edge [red,very thick] node {} (p9)
  			
  			(q5) edge [bend left] node [above left] {$b$} (q6)
  			
			(p1) edge node [below left] {$c$} (f1)
			(p2) edge node [left] {$c,d$} (f1)
			(p3) edge node [below right] {$d$} (f1)
			(p4) edge node [left] {$c$} (f2)
  			(p5) edge node [left] {$c$} (f2)
  			(p6) edge node [left] {$c$} (f3)
  			(p7) edge node [left] {$c$} (f3)
  			(p8) edge node [below left] {$c$} (f4)
			(p9) edge node [left] {$c,d$} (f4)
			(p10) edge node [below right] {$d$} (f4);
			
  		\path
			(q1) -- node [below = 0.7] {$a$} (q1)
			(q1) -- node [below left = 0.6 and 0.2cm] {$a$} (q1)
			(q3) -- node [red, below right = 0.4 and 0.7cm] {$a$} (q4)
			(q3) -- node [blue, below left = 0.7 and 0.9cm] {$a$} (q4)
			(q3) -- node [blue, below left = 0.5 and 1.3cm] {$a$} (q4)
			(q5) -- node [red, below right = 0.5 and 0.8cm] {$a$} (q6)
			(q5) -- node [blue, below left = 0.6 and 0.2cm] {$a$} (q6)
			(q5) -- node [blue, below left = 0.6 and 0.8cm] {$a$} (q6)	
			(q3) -- node [] {$\strictupsim{\strictdwtraceinclusion}$} (q4)
			(q3) -- node [below = 0.1cm] {$\strictdwtraceinclusionlarg$} (q4)
			(q5) -- node [] {$\strictupsim{\strictdwtraceinclusion}$} (q6)
			(q5) -- node [below = 0.1cm] {$\strictdwtraceinclusionlarg$} (q6)
			
			(p1) -- node {$\strictdwsim$} (p2)
			(p2) -- node {$\strictdwsimlarg$} (p3)
			(p4) -- node {$\dwsim$} (p5)
			(p6) -- node {$\dwsim$} (p7)
			(p8) -- node {$\strictdwsim$} (p9)
			(p9) -- node {$\strictdwsimlarg$} (p10)
			;

	\end{tikzpicture}
	\label{GFPcountEx:strictupsim(dwtraceinclusion)strictdwsim}
}
\caption{GFP counterexamples. 
	  A transition is drawn in dashed when a different transition by the same symbol departing from the same state
	  already exists.
	  We draw a transition in thick red when it is better than another transition (drawn in thin blue).}
	\label{fig:GFP_counterexamples}
\end{figure}
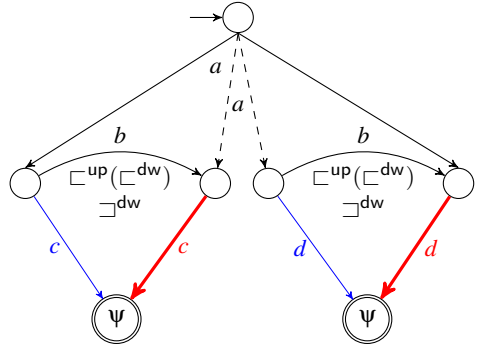
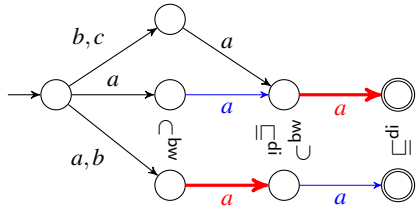
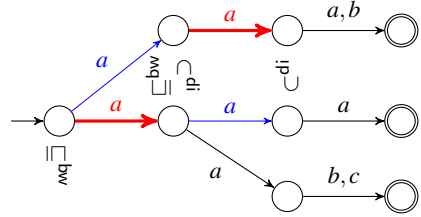
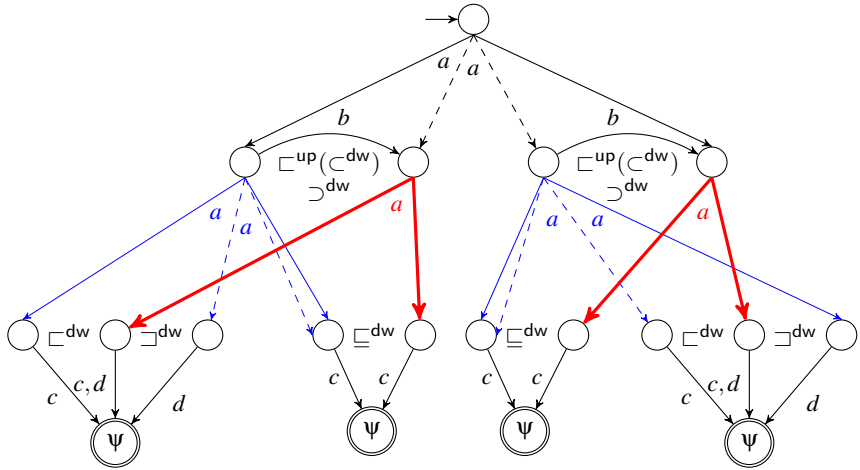
\end{small}

The following theorems and corollaries provide several relations which are GFP. 

\begin{theorem}  \label{thm:id(id)strictdwtraceinclGFP}
For every strict partial order $R \,\subset\; \dwtraceinclusion$,
it holds that $P(\id, R)$ is GFP.
\end{theorem}

\begin{corollary}
 By Theorem~\ref{thm:id(id)strictdwtraceinclGFP},
 $P(\id,\strictdwtraceinclusion)$ and $P(\id,\strictdwsim)$ are GFP.
\end{corollary}

\begin{theorem}  \label{thm:strictuptraceincl(id)idGFP}
For every strict partial order $R \,\subset\; \uptraceinclusion{\id}$,
it holds that $P(R, \id)$ is GFP.
\end{theorem}

\begin{corollary}
 By Theorem~\ref{thm:strictuptraceincl(id)idGFP}, $P(\strictuptraceinclusion{\id},\id)$ and $P(\strictupsim{\id},\id)$ are GFP.
\end{corollary}

\begin{definition}\label{def:W}
Given a tree automaton $\A$, a binary relation $W$ on its states is called a
\emph{downup-relation} iff the following condition holds: 
If $p \mathrel{W} q$ then for every
tree $\tree \in \alltrees(\Sigma)$ and accepting $\tree$-run $\run$ from $p$
there exists an accepting $\tree$-run $\run'$ from $q$ 
such that $\forall_{v \,\in\, \nat^*}\  \run(v) \upsim{W}\,\, \run'(v)$.
\end{definition}

\begin{lemma}\label{lem:relationV}
 Any relation $V$ satisfying 
 \textbf{1)} $V$ is a downward simulation, and \textbf{2)} $\id \,\subseteq\, V \,\subseteq\; \upsim{V}$
 is a downup-relation. In particular, $\id$ is a downup-relation, but
 $\dwsim$ and $\upsim{\id}$ are not.
\end{lemma}

\begin{theorem}\label{thm:strictupsim(R)dwtraceinclGFP}
For every downup-relation $W$, it holds that $P(\strictupsim{W}, \dwtraceinclusion)$ is GFP.
\end{theorem}
\begin{proof}
Let $A' = Prune(\A,P(\strictupsim{W}, \dwtraceinclusion))$.
We show $L(\A) \subseteq L(\A')$.
If $\tree \in L(A)$ then there exists an accepting 
$\tree$-run $\hat{\run}$ in $\A$.
We show that there is an accepting $\tree$-run $\hat{\run}'$ in $A'$.

For each accepting $\tree$-run $\run$ in $\A$, let 
$\level_i(\run)$ be the tuple of states that $\run$ visits at depth $i$ in the
tree, read from left to right. Formally, let 
$(x_1, \dots, x_k)$ with $x_j \in \nat^i$ be the set of all tree positions of
depth $i$ s.t.\ $x_j \in \dom(\run)$, in lexicographically increasing order.
Then $\level_i(\run) = (\run(x_1), \dots, \run(x_k)) \in Q^k$.
By lifting partial orders on $Q$ to partial orders on tuples,
we can compare such tuples w.r.t.\ $\upsim{W}$.
We say that an accepting $\tree$-run $\run$ is $i$-good iff it does not contain
any transition from $A-A'$ from any position $\node \in \nat^*$ with
$|\node| < i$.
I.e., no pruned transition is used in the first $i$ levels of the
tree.

We now define a strict partial order $<_i$ on the set of accepting
$\tree$-runs in $\A$.
Let $\run <_i \run'$ iff 
$\exists k \le i.\ \level_k(\run) \mathrel{\strictupsim{W}} \level_k(\run')$
and $\forall l < k.\, \level_l(\run) \mathrel{\upsim{W}} \level_l(\run')$.
Note that $<_i$ only depends on the first $i$ levels of the run.
Given $\A$, $\tree$ and $i$, there are only finitely many different such
$i$-prefixes of accepting $\tree$-runs. 
By our assumption that $\hat{\run}$ is an
accepting $\tree$-run in $\A$, the set of accepting $\tree$-runs
in $\A$ is non-empty.
Thus, for any $i$, there must exist some accepting $\tree$-run $\run$ in $\A$ that is maximal
w.r.t.\ $<_i$.

We now show that this $\run$ is also $i$-good, by assuming the contrary and
deriving a contradiction. 
Suppose that $\run$ is not $i$-good. Then it must contain a
transition $\langle p, \sigma, r_1 \cdots r_n \rangle$ from $\A-\A'$
used at the root of some subtree $\tree'$ of $\tree$
at some level $j < i$.
Since $A' = Prune(\A,P(\strictupsim{W}, \dwtraceinclusion))$, 
there must exist another transition 
$\langle p', \sigma, r_1' \cdots r_n' \rangle$
in $\A'$ s.t.
(1) $(r_1, \dots, r_n) \dwtraceinclusion (r_1', \dots, r_n')$
and
(2) $p \mathrel{\strictupsim{W}} p'$.

First consider the implications of (2). Upward simulation 
propagates upward stepwise (though only in non-strict form after the first step).
So $p'$ can imitate the upward path of $p$ to the root of $\tree$, maintaining
$\upsim{W}$ between the corresponding states. The states on side branches joining in 
along the upward path from $p$ can be matched by $W$-larger states in joining side branches 
along the upward path from $p'$. From Def.~\ref{def:W}
we obtain that these $W$-larger states in $p'$s joining side branches can accept
their subtrees of $\tree$ via computations that are everywhere $\upsim{W}$
larger than corresponding states in computations from $p$s joining side branches.
So there must be an accepting run $\run'$ on $\tree$ s.t.\ (3) $\run'$ is at state $p'$ at the
root of $\tree'$ and uses transition $\langle p', \sigma, r_1' \cdots r_n' \rangle$ from $p'$, and
(4) for all $v \in \nat^*$ where $\tree(v) \notin \tree'$ we have 
$\run(v) \mathrel{\upsim{W}} \run'(v)$.
Moreover, by conditions (1) and (3), $\run'$ can be extended from $r_1', \dots, r_n'$
to accept also the subtree $\tree'$. Thus $\run'$ is an accepting $\tree$-run
in $\A$.
By conditions (2) and (4) we obtain that 
$\forall l \le j.\ \level_l(\run) \mathrel{\upsim{W}} \level_l(\run')$.
By (2) we get even $\level_j(\run) \mathrel{\strictupsim{W}} \level_j(\run')$ and thus 
$\run <_j \run'$. Since $j < i$ we
also have $\run <_i \run'$ and thus $\run$ was not maximal
w.r.t.\ $<_i$. Contradiction. 
So we have shown that for every $\tree \in L(\A)$ 
there exists an $i$-good accepting run for every finite $i$.
 
If $\tree \in L(\A)$ then there exists an accepting 
$\tree$-run $\hat{\run}$ in $\A$.
Then there exists an accepting $\tree$-run $\hat{\run}'$ 
that is $i$-good, where $i$ is the height of $\tree$.
Thus $\hat{\run}'$ 
is a run in $\A'$ and $\tree \in L(\A')$.
\qed
\end{proof}

\begin{corollary}
 It follows from Lemma~\ref{lem:relationV}  and from the fact that GFP is downward closed that 
 $P(\strictupsim{V}, \dwtraceinclusion)$, $P(\strictupsim{V}, \strictdwtraceinclusion)$,
 $P(\strictupsim{V}, \dwsim)$, $P(\strictupsim{V}, \strictdwsim)$, $P(\strictupsim{V}, \id)$,
 $P(\strictupsim{\id}, \dwtraceinclusion)$, $P(\strictupsim{\id}, \strictdwtraceinclusion)$, 
 $P(\strictupsim{\id}, \dwsim)$ and $P(\strictupsim{\id}, \strictdwsim)$ are GFP.
\end{corollary}

\begin{theorem}  \label{thm:uptraceincl(dwsim)strictdwsimGFP}
 $P(\uptraceinclusion{\dwsim}, \strictdwsim)$ is GFP.
\end{theorem}
\begin{proof}
Let $\A' = Prune(\A,P(\uptraceinclusion{\dwsim}, \strictdwsim))$. 
We show $L(\A) \subseteq L(\A')$.
If $\tree \in L(\A)$ then there exists an accepting 
$\tree$-run $\hat{\run}$ in $\A$.
We show that there is an accepting $\tree$-run $\hat{\run}'$ in $\A'$.

For each accepting $\tree$-run $\run$ in $\A$, let 
$\level_i(\run)$ be the tuple of states that $\run$ visits at depth $i$ in the
tree, read from left to right. Formally, let 
$(x_1, \dots, x_k)$ with $x_j \in \nat^i$ be the set of all tree positions of
depth $i$ s.t.\ $x_j \in \dom(\run)$, in lexicographically increasing order.
Then $\level_i(\run) = (\run(x_1), \dots, \run(x_k)) \in Q^k$.
By lifting partial orders on $Q$ to partial orders on tuples 
we can compare such tuples w.r.t.\ $\dwsim$.
We say that an accepting $\tree$-run $\run$ is $i$-good if it does not contain
any transition from $\A-\A'$ from any position $\node \in \nat^*$ with
$|\node| < i$.
I.e., no pruned transitions are used in the first $i$ levels of the
tree.

We now show, by induction on $i$, the following property (C):
For every $i$ and every accepting
$\tree$-run $\run$ in $\A$ there exists an $i$-good accepting
$\tree$-run $\run'$ in $\A$ s.t.
$\level_i(\run) \dwsim \level_i(\run')$.

The base case is $i=0$. 
Every accepting $\tree$-run $\run$ in $\A$ is trivially $0$-good itself
and thus satisfies (C).

For the induction step, let $S$ be the set of all $(i-1)$-good accepting
$\tree$-runs $\run'$ in $\A$ s.t.\ $\level_{i-1}(\run) \dwsim
\level_{i-1}(\run')$.
Since $\run$ is an accepting $\tree$-run, by induction hypothesis, $S$ is
non-empty.
Let $S' \subseteq S$ be the subset of $S$ containing exactly those
runs $\run' \in S$ that additionally satisfy $\level_{i}(\run) \dwsim
\level_{i}(\run')$.
From $\level_{i-1}(\run) \dwsim \level_{i-1}(\run')$ and 
the fact that $\dwsim$ is preserved downward-stepwise, we obtain that $S'$ is
non-empty.
Now we can select some $\run' \in S'$ s.t. \ $\level_i(\run')$ is maximal,
w.r.t.\ $\dwsim$, relative to the other runs in $S'$.
We claim that $\run'$ is $i$-good and $\level_i(\run) \dwsim \level_i(\run')$.
The second part of this claim holds because $\run' \in S'$.

We show that $\run'$ is $i$-good by contraposition.
Suppose that $\run'$ is not $i$-good. Then it must contain a
transition 
$\langle p, \sigma, r_1 \cdots r_n \rangle$ 
from $\A-\A'$. Since $\run'$ is $(i-1)$-good, this transition must start at
depth $(i-1)$ in the tree.
Since $\A' = Prune(\A,P(\uptraceinclusion{\dwsim}, \strictdwsim))$,
there must exist another transition 
$\langle p', \sigma, r_1' \cdots r_n' \rangle$
in $\A'$ s.t.\ $p \mathrel{\uptraceinclusion{\dwsim}} p'$
and $(r_1, \dots, r_n) \strictdwsim (r_1', \dots, r_n')$.
From the definition of $\uptraceinclusion{\dwsim}$ we obtain that
there exists another accepting $\tree$-run $\run_1$ in $\A$
(that uses the transition $\langle p', \sigma, r_1' \cdots r_n' \rangle$)
s.t.\ $\level_{i}(\run') \strictdwsim \level_{i}(\run_1)$.
The run $\run_1$ is not necessarily $i$-good or $(i-1)$-good.
However, by induction hypothesis, there exists some
accepting $\tree$-run $\run_2$ in $\A$ that is $(i-1)$-good and satisfies
$\level_{i-1}(\run_1) \dwsim \level_{i-1}(\run_2)$.
Since $\dwsim$ is preserved stepwise, there also exists an
accepting $\tree$-run $\run_3$ in $\A$ (that coincides with $\run_2$ up-to
depth $(i-1)$), which is $(i-1)$-good and satisfies
$\level_{i}(\run_1) \dwsim \level_{i}(\run_3)$.
In particular, $\run_3 \in S'$.

From $\level_{i}(\run') \strictdwsim \level_{i}(\run_1)$
and
$\level_{i}(\run_1) \dwsim \level_{i}(\run_3)$ we obtain
$\level_{i}(\run') \strictdwsim \level_{i}(\run_3)$.
This contradicts our condition above that $\run'$ must be
$\level_i$ maximal w.r.t.\ $\dwsim$ in $S'$.
This concludes the induction step and the proof of property (C).

If $\tree \in L(\A)$ then there exists an accepting 
$\tree$-run $\hat{\run}$ in $\A$.
By property (C), there exists an accepting $\tree$-run $\hat{\run}'$ 
that is $i$-good, where $i$ is the height of $\tree$.
Therefore $\hat{\run}'$ does not use any transition from $\A-\A'$ and is 
thus also a run in $\A'$. So we obtain $\tree \in L(\A')$.
\qed
\end{proof}

\begin{corollary}
 It follows from Theorem~\ref{thm:uptraceincl(dwsim)strictdwsimGFP} and the fact that GFP is downward closed that 
 $P(\strictuptraceinclusion{\dwsim},\strictdwsim)$, $P(\upsim{\dwsim},\strictdwsim)$, $P(\strictupsim{\dwsim},\strictdwsim)$,
 $P(\uptraceinclusion{\id},\strictdwsim)$, $P(\strictuptraceinclusion{\id},\strictdwsim)$, $P(\upsim{\id},\strictdwsim)$ and $P(\id,\strictdwsim)$ are GFP.
 \end{corollary}

\section{State Quotienting Techniques}\label{sec:stateQuot}

A classic method for reducing the size of automata is state quotienting. 
Given a suitable equivalence relation on the set of states, each equivalence
class is collapsed into just one state.
From a preorder $\sqsubseteq$ one obtains an equivalence relation 
$\equiv \;:=\; \sqsubseteq \!\cap\! \sqsupseteq$.
We now define quotienting w.r.t.\ $\equiv$.
Let $\A = (\Sigma,Q,\delta,I)$ be a TDTA and let $\sqsubseteq$ be a preorder on $Q$. 
Given $q \in Q$, we denote by $[q]$ its equivalence class w.r.t $\equiv$.
For $P \subseteq Q$, $[P]$ denotes the set of equivalence classes 
$[P] = \{[p]\,|\, p \in P\}$.
We define the quotient automaton w.r.t.\ $\equiv$ 
as $\A/\!\equiv \; := \; (\Sigma, [Q], \delta_{A/\!\equiv}, [I])$,
where $\delta_{A/\!\equiv} = \{ \langle [q],\sigma,[q_1]\ldots[q_n] \rangle
\mid \langle q,\sigma,q_1\ldots q_n \rangle \in \delta_A \}$.
It is trivial that $L(A) \subseteq L(A/\!\!\equiv)$ for any $\equiv$. 
If the reverse inclusion also holds, 
i.e., if $L(A) = L(A/\!\!\equiv)$, we say that $\equiv$ is \emph{good for quotienting} (GFQ).

It was shown in \cite{Lukas:PhD13} that 
$\dwsim \!\!\cap\! \dwsimlarg$ and
$\upsim{\id} \cap \upsimlarg{\id}$ are GFQ. Here we generalize this result
from simulation to trace equivalence.
Let $\equiv^\mathsf{dw} \;:=\; \dwtraceinclusion \!\!\cap\! \dwtraceinclusionlarg$
and $\equiv^\mathsf{up}\!\!(R) :=\; \uptraceinclusion{R} \cap \uptraceinclusionlarg{R}$.

\begin{theorem}		\label{thm:dwtrinclgfq}
$\equiv^\mathsf{dw}$ is GFQ.		
\end{theorem}

\begin{theorem}		\label{thm:upidtrinclgfq}
 $\equiv^\mathsf{up}\!\!(\id)$ is GFQ.
\end{theorem}

\noindent
In Figure~\ref{fig:GFQ_counterexample} (cf.\ Appendix~\ref{sec:appendix_A})
we present a counterexample showing that 
$\equiv \,:=\, \upsim{\dwsim \!\!\cap\! \dwsimlarg} \,\cap \upsimlarg{\dwsim \!\!\cap\! \dwsimlarg}$ is not GFQ. 
This is an adaptation from the Example 5 in \cite{Lukas:PhD13},
where the inducing relation is referred to as the \emph{downward bisimulation equivalence} and the 
automata are seen bottom-up.

One of the best methods previously known for reducing TA performs state quotienting 
based on a combination of downward and upward simulation~\cite{lukas:framework2009}.
However, this method cannot achieve any further reduction on an automaton which has been previously
reduced with the techniques we described above 
(cf.~Theorem~\ref{thm:combinedpreorder} in Appendix~\ref{sec:app_combinedPreorder}).

\section{Lookahead Simulations}\label{sec:lookahead}

Simulation preorders are generally not very good under-approximations
of trace inclusion, since they are much
smaller on many automata. Thus we consider better approximations that
are still efficiently computable.

For word automata, more general \emph{lookahead simulations} were introduced
in \cite{mayr:advanced2013}. These provide a practically useful tradeoff
between the computational effort and the size of the obtained relations.
Lookahead simulations can also be seen as a particular restriction of the 
more general (but less practically useful) 
\emph{multipebble simulations}~\cite{etessami:hierarchy2002}.
We generalize lookahead simulations to tree automata in
order to compute good under-approximations of trace inclusions.

\mypar{Intuition by Simulation Games.}
Normal simulation preorder on labeled transition graphs 
can be characterized by a game between two players, Spoiler and Duplicator. 
Given a pair of states $(q_0,r_0)$, Spoiler wants to show that 
$(q_0,r_0)$ is not contained in the simulation preorder relation, while
Duplicator has the opposite goal.
Starting in the initial configuration $(q_0,r_0)$, Spoiler chooses a
transition $q_0 \stackrel{\sigma}{\rightarrow} q_1$  
and Duplicator must imitate it \emph{stepwise} by 
choosing a transition with the same symbol $r_0 \stackrel{\sigma}{\rightarrow} r_1$. 
This yields a new configuration $(q_1,r_1)$ from which the game continues.
If a player cannot move the other wins. Duplicator wins every
infinite game. Simulation holds iff Duplicator wins. 

In normal simulation, Duplicator only knows Spoiler's very next step (see above),
while in {\em $k$-lookahead simulation} Duplicator knows Spoiler's $k$ next steps
in advance (unless Spoiler's move ends in a deadlocked state - 
i.e., a state with no transitions).
As the parameter $k$ increases, the $k$-lookahead simulation relation
becomes larger and thus approximates the trace inclusion relation
better and better. 
Trace inclusion can also be characterized by a game. 
In the trace inclusion game, Duplicator knows {\em all} steps of Spoiler in
the entire game in advance.

For every fixed $k$, $k$-lookahead simulation is computable
in polynomial time, though the complexity rises quickly in $k$: it is doubly
exponential for downward- and single exponential for upward lookahead
simulation (due to the downward branching of trees).
A crucial trick 
makes it possible to practically compute it for nontrivial $k$: Spoiler's moves are
built incrementally, and Duplicator need not respond to all of 
Spoiler's announced $k$ next steps, but only to a prefix of
them, after which he may request fresh information \cite{mayr:advanced2013}.
Thus Duplicator just uses the minimal lookahead necessary to win the current step.

\mypar{Lookahead downward simulation.}
We say that a tree $\tree$ is $k$-bounded iff for all leaves $\node$ of $\tree$,
either \textbf{a)} $|\node| = k$, or \textbf{b)} $|\node|<k$ and $\node$ is closed.
Let $\A = (\Sigma,Q,\delta, I)$ be a TDTA.
A \emph{$k$-lookahead downward simulation} $L^{k-\mathsf{dw}}$ 
is a binary relation on $Q$ such that if $q \mathrel{L^{k-\mathsf{dw}}} r$, 
then $(q=\spstate \Longrightarrow r=\spstate)$ and the following holds:
Let $\run_k$ be a run on a $k$-bounded tree $\tree_k$ with $\run(\epsilon)=q$
s.t.\ every leaf node of $\run_k$ is either at depth $k$ or
downward-deadlocked (i.e., no more downward transitions exist).
Then there must exist a run $\run_k'$ over a nonempty prefix $\tree_k'$ of $\tree_k$ 
s.t. (1) $\run_k'(\epsilon)=r$, and 
(2) for every leaf $\node$ of $\run_k'$, \:\!\! $\run_k(\node) \mathrel{L^{k-\mathsf{dw}}} \run_k'(\node)$.
%
Since, for given $\A$ and $k\ge 1$, lookahead downward simulations are closed under union,
there exists a unique maximal one  that we call
\emph{the $k$-lookahead downward simulation on $\A$}, denoted by $\kdwsim{k}$.
While $\kdwsim{k}$ is trivially reflexive, it is not transitive in general 
(cf.\ \cite{mayr:advanced2013}, App. B).
Since we only use it as a means to under-approximate the transitive trace
inclusion relation $\dwtraceinclusion$ (and require a preorder to induce an
equivalence), we work with its transitive closure
$\transkdwsim{k} := (\kdwsim{k})^+$.
In particular, $\transkdwsim{k} \ \subseteq\  \dwtraceinclusion$.

\mypar{Lookahead upward simulation.}
%
%
Let $\A = (\Sigma,Q,\delta, I)$ be a TDTA.
A \emph{$k$-lookahead upward simulation} on $\A$ induced by a relation $R$ is
a binary relation $\mathrel{L^{k-\mathsf{up}}(R)}$
on $Q$ s.t.\ if $q \mathrel{L^{k-\mathsf{up}}(R)} r$, 
then $(q = \spstate \Longrightarrow r=\spstate)$ and the following holds:
Let $\run$ be a run over a tree $\tree \in \alltrees(\Sigma)$ with 
$\run(\node) = q$ for some bottom leaf $\node$ s.t.\ either
$|\node| = k$ or $0 < |\node| < k$ and $\run(\epsilon)$ is upward-deadlocked (i.e., no more
upward transitions exist).

Then there must exist $\node',\node''$
such that $\node = \node'\node''$ and $|\node''| \ge 1$ 
and a run $\run'$ over $\tree_{\node'}$
s.t.\ the following holds.
(1) $\run'(\node'') = r$,
(2) ${\run(\node')} \mathrel{L^{k-\mathsf{up}}(R)} {\run'(\epsilon)}$,
(3) $\run(\node'x) \in I \Longrightarrow \run'(x) \in I$ for all prefixes $x$
of $\node''$,
(4) If $\node'xy \in \dom(\run)$ for some strict prefix $x$ of $\node''$ and some 
$y \in \nat$ where $xy$ is not a prefix of $\node''$ then 
${\run(\node'xy)} \mathrel{R} {\run'(xy)}$.

\smallskip

\ignore{
Let $\A = (\Sigma,Q,\delta, I)$ be a TDTA.
A \emph{$k$-lookahead upward simulation} on $\A$ induced by a relation $R$ is a binary relation $\kupsim{k}{R}$ on 
$Q$ such that if $q \mathrel{\kupsim{k}{R}} r$, 
then $q\in I \implies r\in I$ and the following holds:
Let $\run$ be a run over a linear tree $\tree$ with $|\dom(\tree)| = k$ 
and let $\nodew$ be its bottom leaf with $\run(\nodew) = q$.
Then there exists a node $\node$ of $\tree$ and a run $\run'$ over the subtree $\tree_\node$ of $\tree$ such that
\begin{enumerate}
\item
$\run'(\nodew') = r$ where $\nodew'$ is the node with $\node \nodew' = \nodew$;
\item
$\run(\node) \mathrel{\kupsim{k}{R}} \run'(\epsilon)$;
\item
for all nodes $\nodex$ of $\tree_\node$, if $\run(\node\nodex)\in I$ then $\run'(\nodex) \in I$;
\item
for all leaves $\nodex$ of $\run'$ except $\nodew'$, $\run(\node\nodex) \mathrel{R} \run'(\nodex)$.
\end{enumerate}
\medskip
}

Since, for given $\A$, $k\ge 1$ and $R$, lookahead upward simulations are closed under union,
there exists a unique maximal one that we call
\emph{the $k$-lookahead upward simulation induced by $R$ on $\A$}, denoted by $\kupsim{k}{R}$.
Since both $R$ and $\kupsim{k}{R}$ are not necessarily transitive,
we first compute its transitive closure, $R^+$, 
and we then compute $\transkupsim{k}{R} := (\kupsim{k}{R^+})^+$, which 
under-approximates the upward trace inclusion $\uptraceinclusion{R^+}$.

\section{Experiments}\label{sec:experiments}

Our tree automata reduction algorithm 
(tool available \cite{tool:libvata-heavy})
combines transition pruning
techniques (Sec.~\ref{sec:transPruning}) with quotienting techniques
(Sec.~\ref{sec:stateQuot}). Trace inclusions are under-approximated
by lookahead simulations (Sec.~\ref{sec:lookahead}) where higher
lookaheads are harder to compute but yield better 
approximations.
The parameters $x,y \ge 1$ describe
the lookahead for downward/upward lookahead simulations, respectively.
Downward lookahead simulation is harder to compute than upward
lookahead simulation, since the number of possible moves is doubly exponential
in $x$ (due to the downward branching of the tree) while for upward-simulation
it is only single exponential in $y$. We use $(x,y)$ as $(1,1)$, $(2,4)$ and $(3,7)$.

Besides pruning and quotienting, we also use the operation ${\it RU}$ that
removes useless states, i.e., states that either cannot be reached from any
initial state or from which no tree can be accepted.
Let ${\it Op}(x,y)$ be the following sequence of 
operations on tree automata:
$\ru$, 
quotienting with $\transkdwsim{x}$, 
pruning with $P(\id,\transasymkdwsim{x})$,
$\ru$,
quotienting with $\transkupsim{y}{\id}$,
pruning with $P(\transasymkupsim{y}{\id}, \id)$,
pruning with $P(\strictupsim{\id},\transkdwsim{x})$,
$\ru$,
quotienting with $\transkupsim{y}{\id}$,
pruning with $P(\transkupsim{y}{\dwsim},\strictdwsim)$,
$\ru$.
It is language preserving by the Theorems of Sections~\ref{sec:transPruning}
and \ref{sec:stateQuot}. The order of the operations is chosen according to some
considerations of efficiency. (No order is ideal for all instances.)

Our algorithm ${\it Heavy}(1,1)$ just iterates ${\it Op}(1,1)$ 
until a fixpoint is reached.
For efficiency reasons, the general algorithm ${\it Heavy}(x,y)$ does not iterate
${\it Op}(x,y)$, but uses a double loop: it iterates the sequence 
${\it Heavy}(1,1){\it Op}(x,y)$ until a fixpoint is reached.

We compare the reduction performance of several algorithms.
\begin{description}
\item[RU:]
$\ru$. (Previously present in \libvata.)
\item[RUQ:]
$\ru$ and quotienting with $\dwsim$. (Previously present in \libvata.)
\item[RUQP:]
{\bf RUQ}, plus pruning with $P(\id,\strictdwsim)$.
(Not in \libvata, but simple.)
\item[Heavy:]
${\it Heavy}(1,1)$, ${\it Heavy}(2,4)$ and ${\it Heavy}(3,7)$. (New.) 
\end{description}
\vspace{-3pt}

We tested these algorithms on three sets of automata from the
\libvata\ distribution. 
The first set are 27 moderate-sized automata 
(87 states and 816 transitions on avg.)
derived from regular model checking applications.
Heavy(1,1), on avg., reduced the number of states and transitions
\emph{to} 27\% and 14\% of the original sizes, resp. 
(Note the difference between `to' and `by'.) 
In contrast, RU did not perform any reduction in any case,
RUQ, on avg., reduced the number of states and transitions
only to 81\% and 80\% of the original sizes
and RUQP reduced the number of states and transitions
to 81\% and 32\% of the original sizes;
cf.~Fig.~\ref{fig:moderate}.
The average computation times of Heavy(1,1), 
RUQP, RUQ and RU were, respectively,
0.05s, 
0.03s, 0.006s and 0.001s.

\begin{figure}[htp]
\begin{minipage}{12.2cm}
\vspace{-10pt}
\includegraphics[width=6cm]{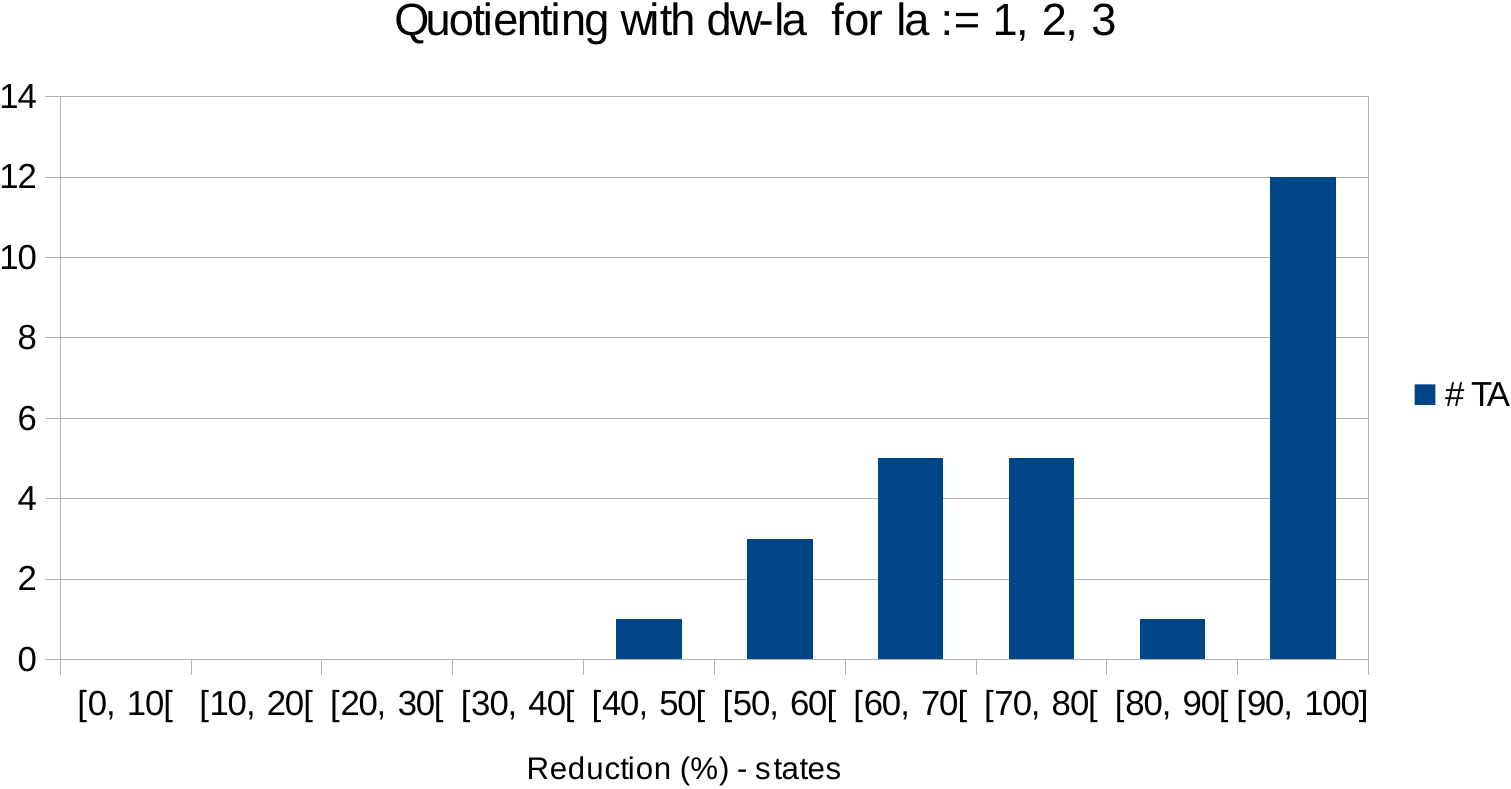}
\includegraphics[width=6cm]{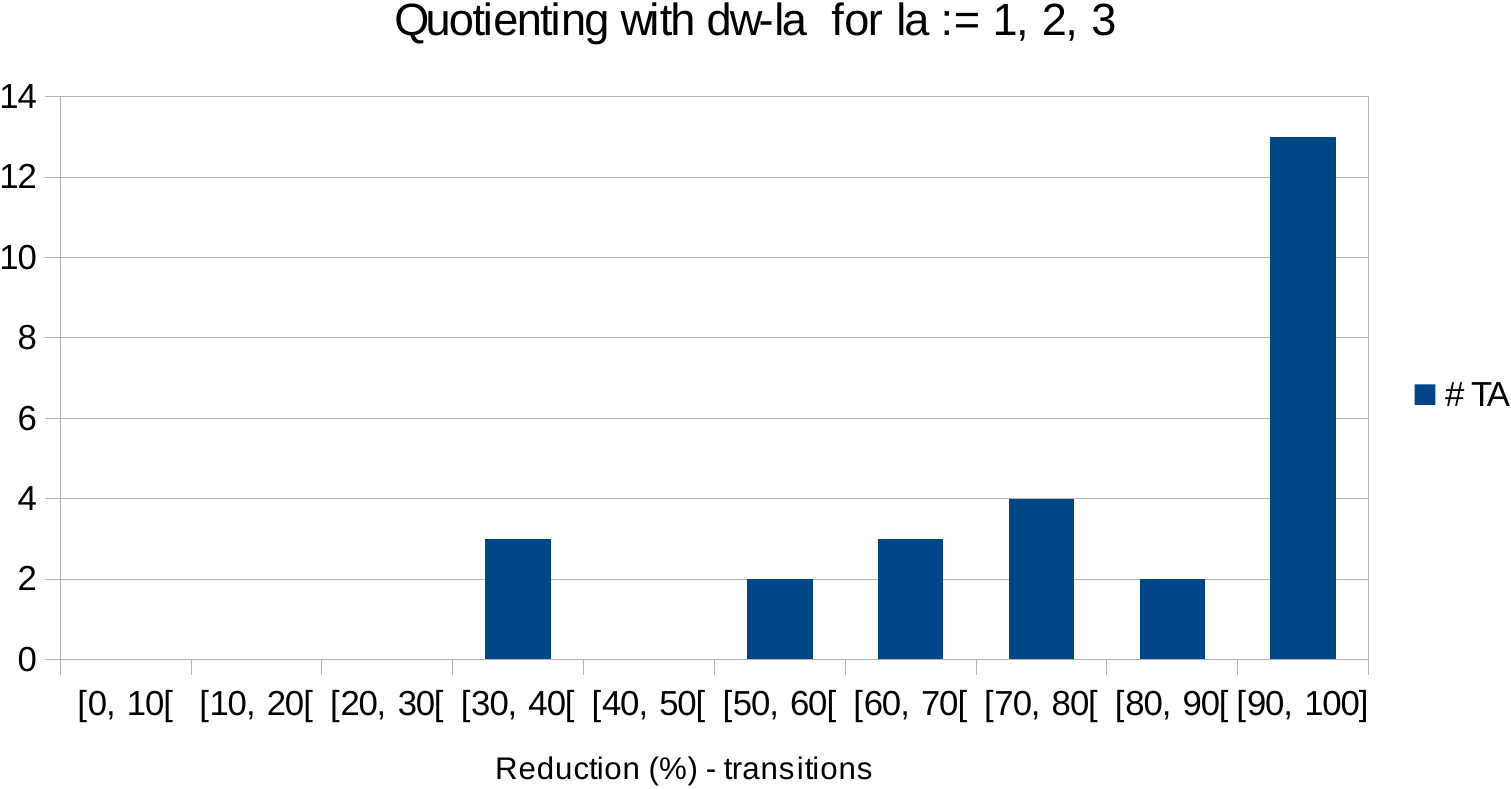}
\includegraphics[width=6cm]{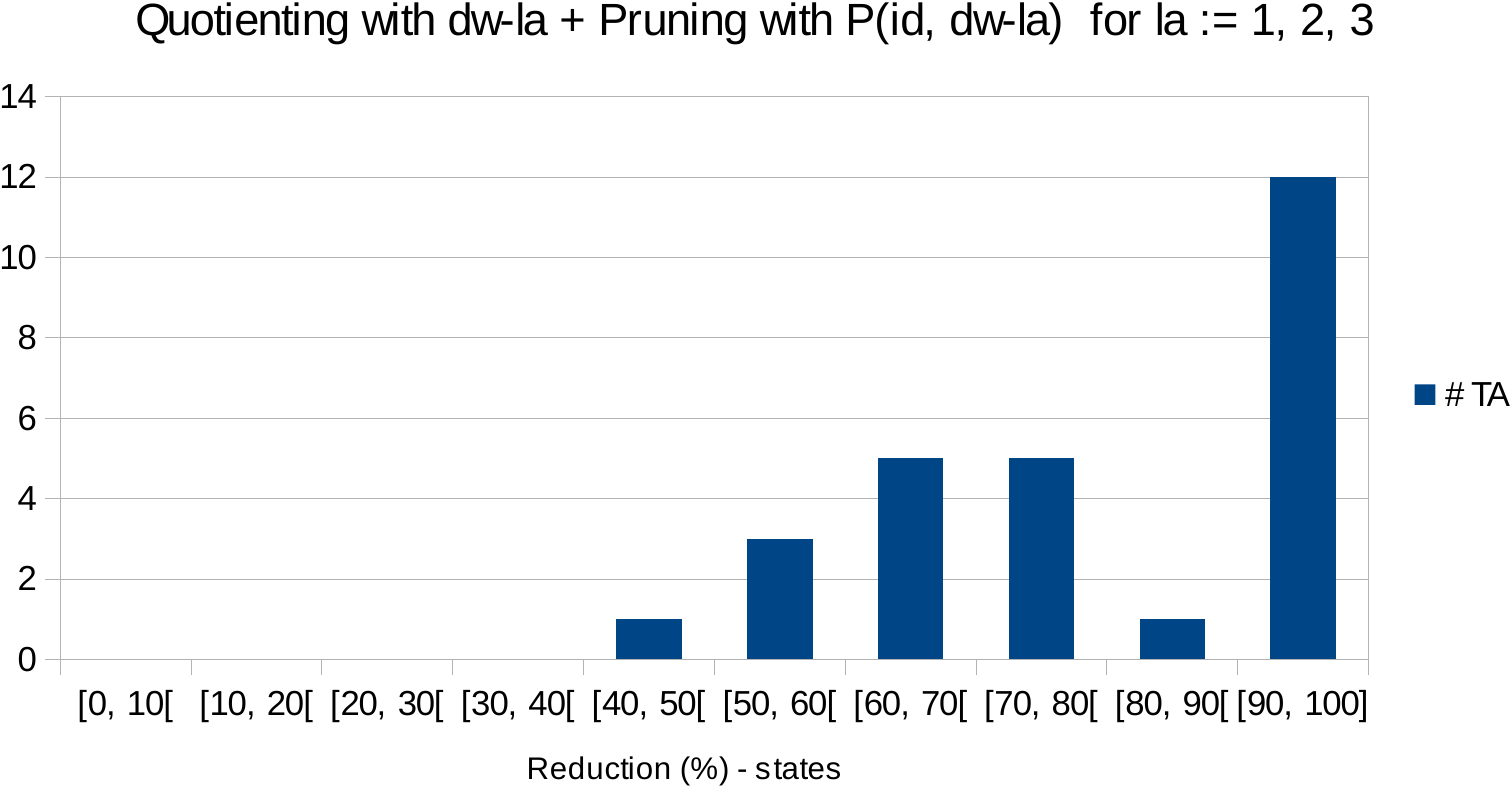}
\includegraphics[width=6cm]{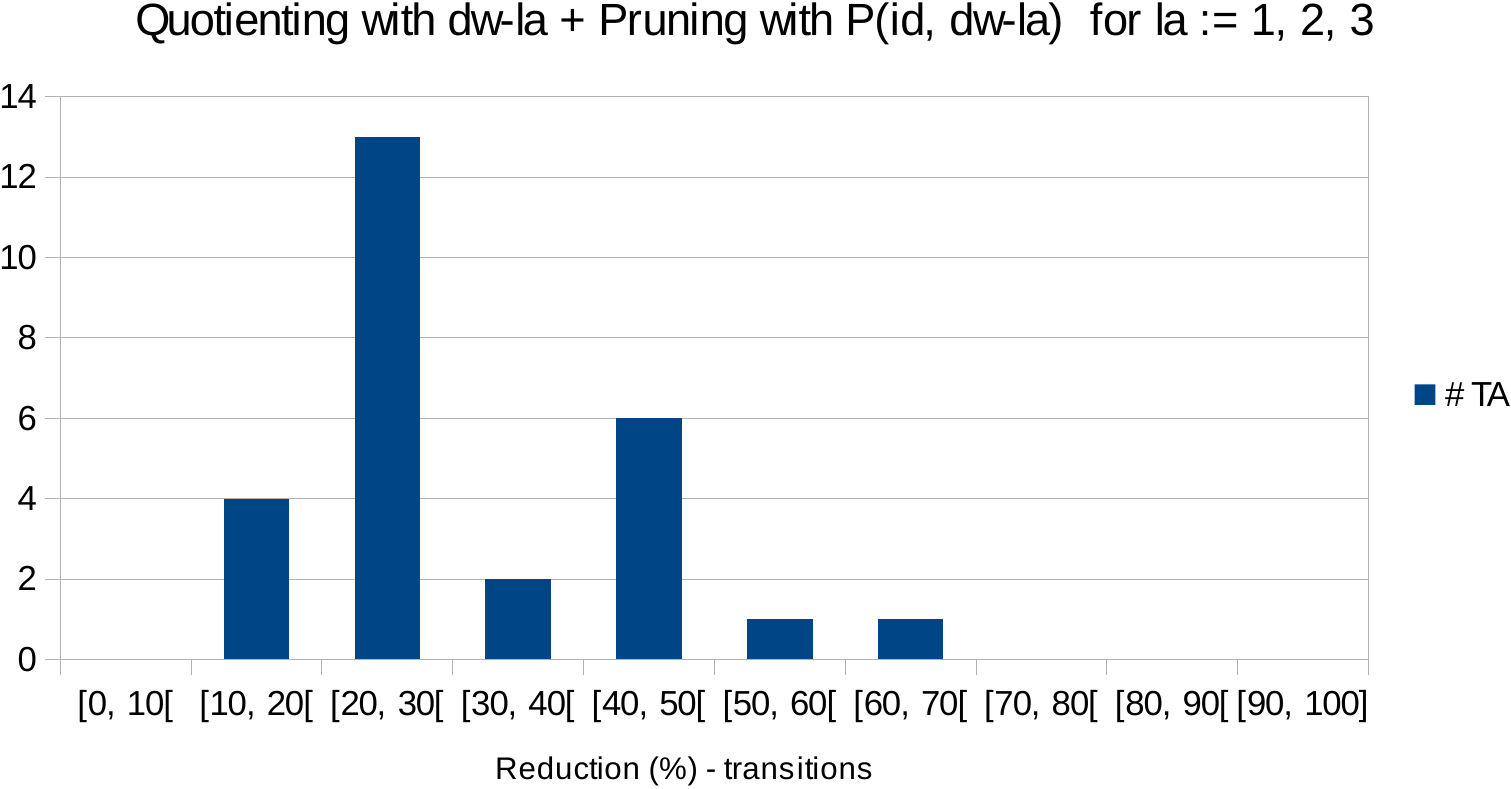}
\includegraphics[width=6cm]{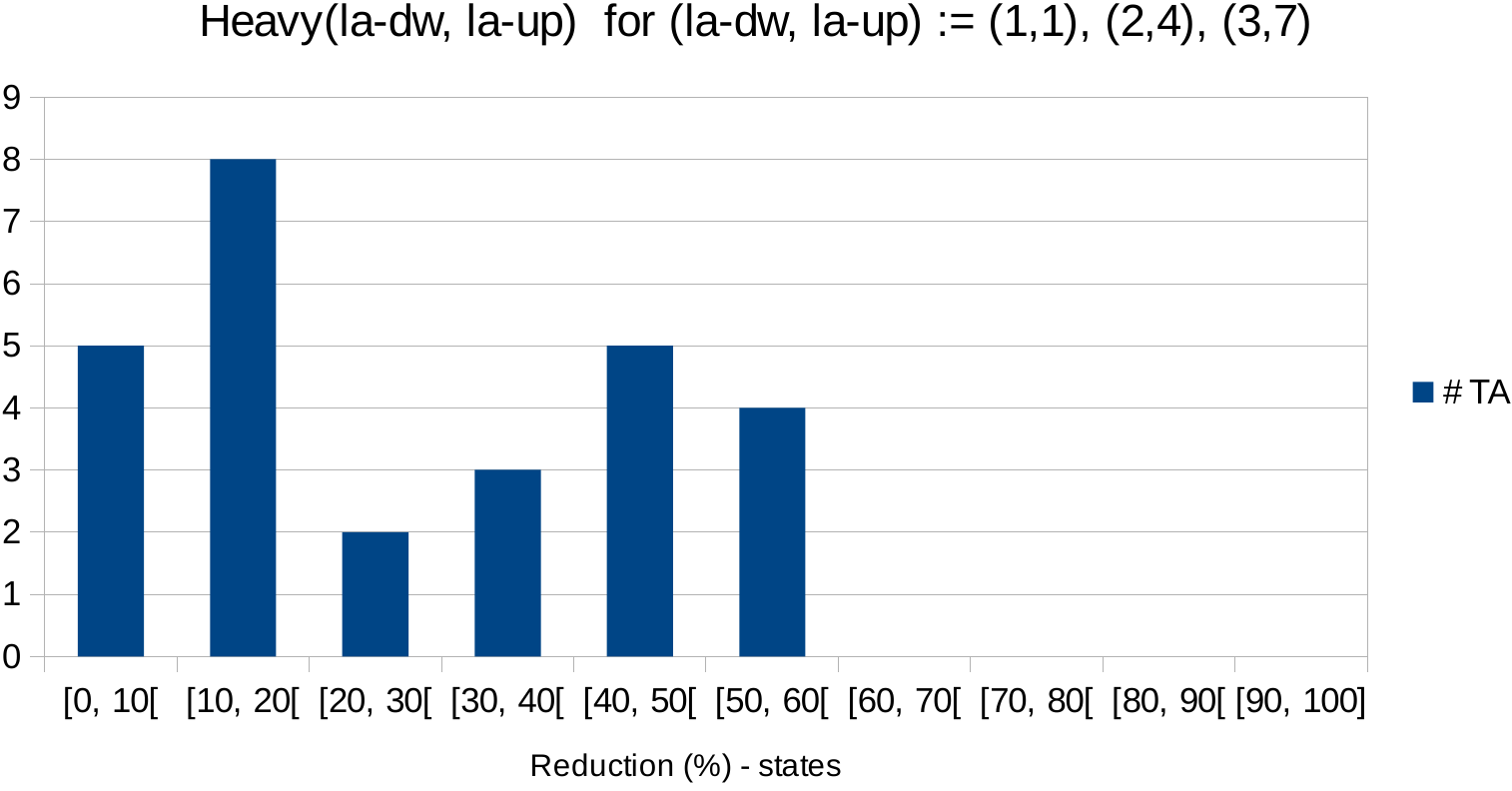}\;\,
\includegraphics[width=6cm]{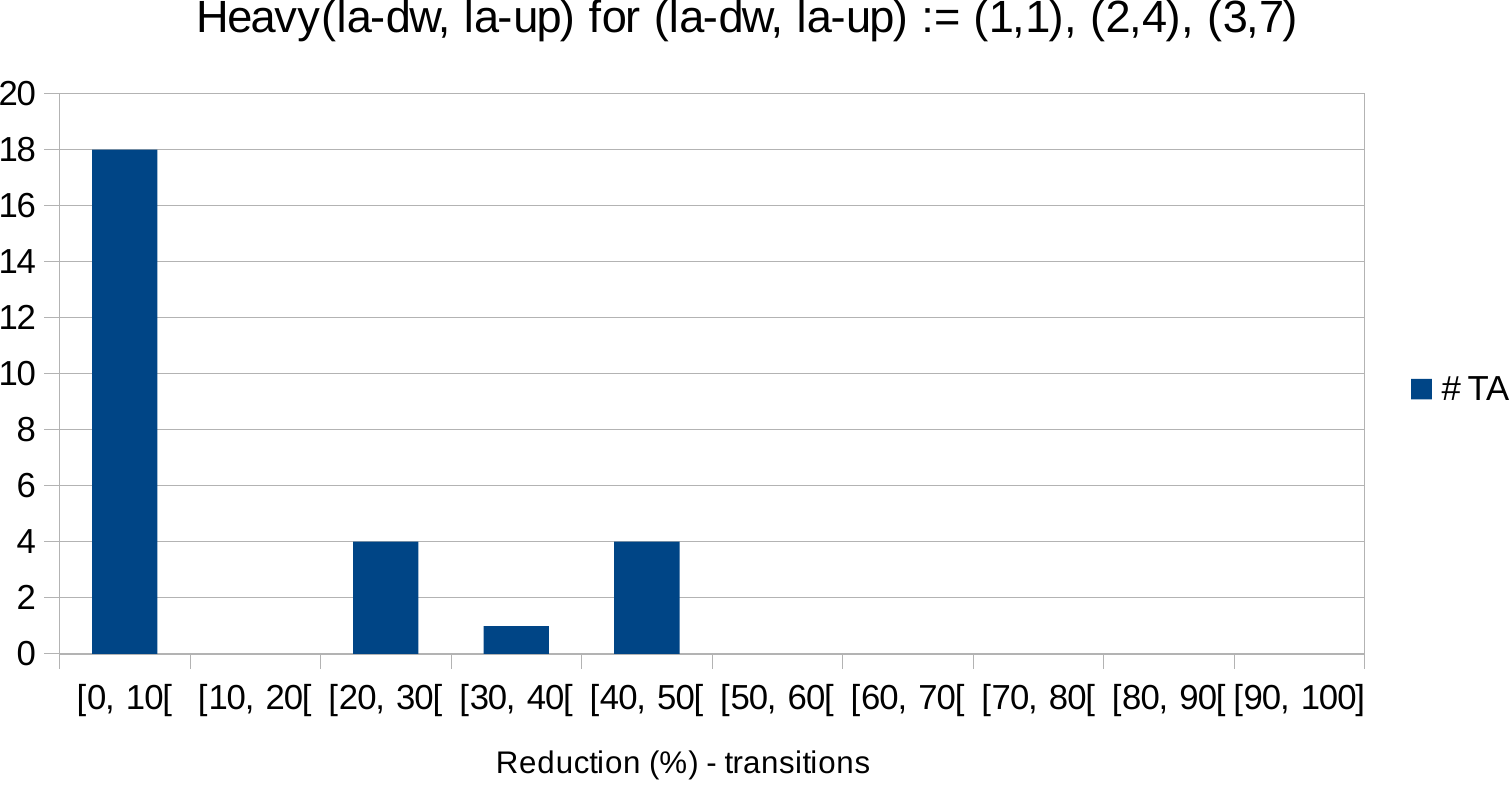}
\captionof{figure}{Reduction of 27 moderate-sized tree automata by methods RUQ (top row),
 RUQP (middle row),
 and Heavy (bottom row). 
A bar of height $h$ at an interval $[x,x+10[$ means that $h$ of the 27
    automata were reduced to a size between $x\%$ and $(x+10)\%$ of their
    original size. The reductions in the numbers of states/transitions are
    shown on the left/right, respectively.
    On this set of automata, the methods Heavy(2,4) and Heavy(3,7) gave
    exactly the same results as Heavy(1,1).}
\label{fig:moderate}
\vspace{20pt}
\end{minipage}\\

\smallskip
\begin{minipage}{12.2cm}
\includegraphics[width=12cm,height=4cm]{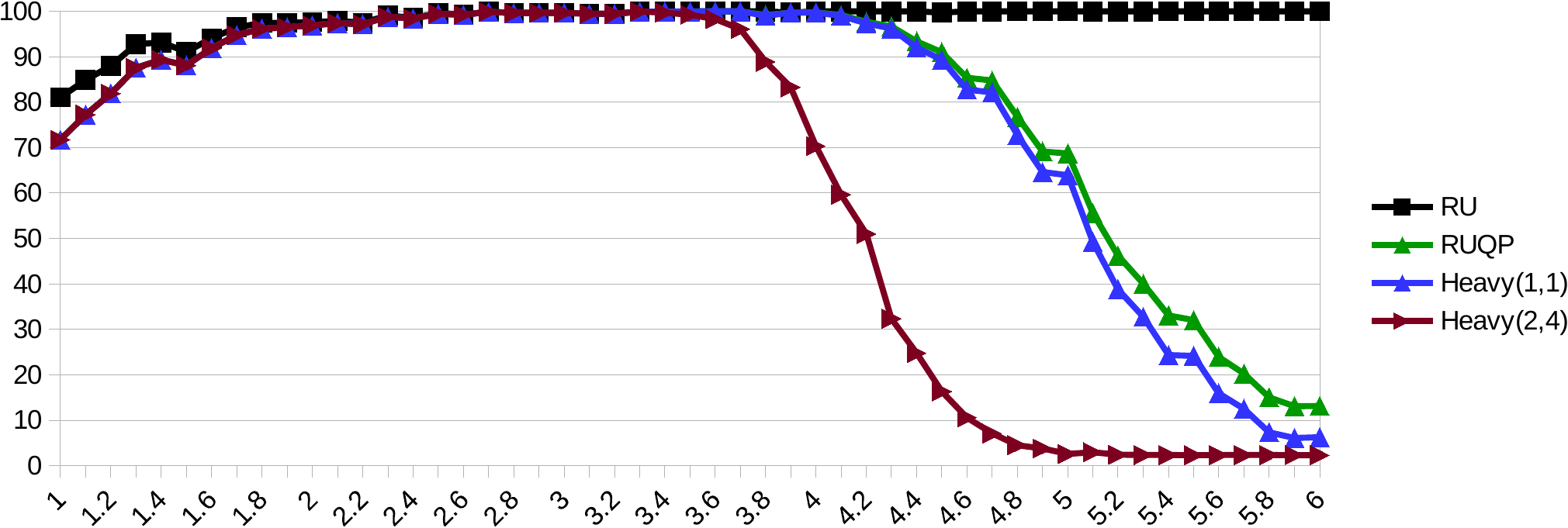}
\captionof{figure}{Reduction of Tabakov-Vardi random tree automata with $n=100, s=2$
  and ${\it ad}=0.8$. The $x$-axis gives the transition density ${\it td}$, and the
  $y$-axis gives the average number of states after reduction with the various
  methods (smaller is better). Each data point is the average of 400 random
  automata. Note that Heavy(2,4) reduces much better than Heavy(1,1) 
  for ${\it td} \ge 3.5$. Computing Heavy(x,y) for even higher $x,y$ is very slow
  on (some instances of) random automata. 
}\label{fig:random}
\end{minipage}
\end{figure}

The second set are 62 larger automata 
(586 states and 8865 transitions, on avg.)
derived from regular model checking applications.
Heavy(1,1), on avg., reduced the number of states and transitions
\emph{to} 4.2\% and 0.7\% of the original sizes.
In contrast,
RU did not perform any reduction in any case,
RUQ, on avg., reduced the number of states and transitions
to 75.2\% and 74.8\% of the original sizes
and RUQP reduced the number of states and transitions
to 75.2\% and 15.8\% of the original sizes;
cf.~Table~\ref{table:nonmoderate} in App.\ref{sec:app_experiments}.
The average computation times of Heavy(1,1), RUQP, RUQ and RU were, respectively,
2.7s, 2.1s, 0.2s and 0.02s.

The third set are 14,498 automata (57 states and 266 transitions on avg.)
from the shape analysis tool Forester \cite{tool:forester}. 
Heavy(1,1), on avg., reduced the number of states/transitions
\emph{to} 76.4\% and 67.9\% of the original, resp. 
RUQ and RUQP reduced the states and transitions only
to 94\% and 88\%, resp.
The average computation times of Heavy(1,1), 
RUQP, RUQ and RU were, respectively,
0.21s, 
0.014s, 0.004s, and 0.0006s. 

Due to the particular structure of the automata in these 3 sample sets, 
${\it Heavy}(2,4)$ and ${\it Heavy}(3,7)$ had hardly any advantage over ${\it Heavy}(1,1)$.
However, in general they can perform significantly better.

We also tested the algorithms on randomly generated tree automata, according
to a generalization of the Tabakov-Vardi model of random word automata
\cite{tabakov:model}. Given parameters $n, s$, ${\it td}$ (transition density)
and ${\it ad}$ (acceptance density), it
generates tree automata with $n$ states, $s$ symbols (each of rank 2), 
$n*{\it td}$ randomly assigned transitions for each symbol, and
$n*{\it ad}$ randomly assigned leaf rules. Figure~\ref{fig:random} 
shows the results of reducing automata of varying ${\it td}$ with different methods.

\section{Summary and Conclusion}\label{sec:conclusion}

The tables in Figure~\ref{fig:GFP_relations_trees} and
Figure~\ref{fig:GFQ_relations_trees} summarize all our results
on pruning and quotienting, respectively.
Note that negative results propagate to larger relations
and positive results propagate to smaller relations (i.e., GFP/GFQ is downward
closed).

The experiments show that our Heavy(x,y) algorithm can significantly reduce the size of
many classes of nondeterministic tree automata, and that it is sufficiently
fast to handle instances with hundreds of states and thousands of transitions.

\begin{figure}[htb]
\begin{minipage}{6.5cm}
      \begin{tabular}{c|c|ccccc}
	  \multicolumn{2}{c}{}				   & 	\multicolumn{5}{c}{$\drel$}		\\
	  \cline{3-7}	  
	  \multicolumn{2}{c|}{$\urel\backslash\irel$}	   		&	$\id$		&	$\strictdwsim$		&	$\dwsim$	& 	$\strictdwtraceinclusion$	&	$\dwtraceinclusion$	\\
	  \hline
	  $\id$					& $\id$		 	& 	$\tickNA$	&	$\tickOK$		&	$\tickNA$	&	$\tickOK$	&	$\tickNA$	\\
	  \hline
	  
	  \multirow{6}{*}{$\sqsubset^{\mathsf{up}}$}	& $\id$				& 	$\tickOK$	&	$\tickOK$	&	$\tickOK$	&	$\tickOK$	&	$\tickOK$	\\
							& $\strictdwsim$ 		& 	$\tickNO$	&	$\tickOK$	&	$\tickNO$	&	$\tickNO$	&	$\tickNO$	\\
							& $\dwsim$ 			& 	$\tickNO$	&	$\tickOK$	&	$\tickNO$	&	$\tickNO$	&	$\tickNO$	\\
							& downup-rel. 				& 	$\tickOK$	&	$\tickOK$	&	$\tickOK$	&	$\tickOK$	&	$\tickOK$	\\
							& $\strictdwtraceinclusion$ 	& 	$\tickNO$	&	$\tickNO$	&	$\tickNO$	&	$\tickNO$	&	$\tickNO$	\\
							& $\dwtraceinclusion$ 		& 	$\tickNO$	&	$\tickNO$	&	$\tickNO$	&	$\tickNO$	&	$\tickNO$	\\
	  \hline
	  
	  \multirow{6}{*}{$\sqsubseteq^{\mathsf{up}}$}	& $\id$				& 	$\tickNA$	&	$\tickOK$	&	$\tickNA$	&	$\tickNO$	&	$\tickNA$	\\
							& $\strictdwsim$ 		& 	$\tickNA$	&	$\tickOK$	&	$\tickNA$	&	$\tickNO$	&	$\tickNA$	\\
							& $\dwsim$ 			& 	$\tickNA$	&	$\tickOK$	&	$\tickNA$	&	$\tickNO$	&	$\tickNA$	\\
							& $\strictdwtraceinclusion$ 	& 	$\tickNA$	&	$\tickNO$	&	$\tickNA$	&	$\tickNO$	&	$\tickNA$	\\
	  						& $\dwtraceinclusion$ 		& 	$\tickNA$	&	$\tickNO$	&	$\tickNA$	&	$\tickNO$	&	$\tickNA$	\\
	  \hline					
	  
	  \multirow{6}{*}{$\subset^{\mathsf{up}}$}	& $\id$ 			& 	$\tickOK$	&	$\tickOK$	&	$\tickNO$	&	$\tickNO$	&	$\tickNO$	\\
							& $\strictdwsim$		& 	$\tickNO$	&	$\tickOK$	&	$\tickNO$	&	$\tickNO$	&	$\tickNO$	\\
							& $\dwsim$			& 	$\tickNO$	&	$\tickOK$	&	$\tickNO$	&	$\tickNO$	&	$\tickNO$	\\
							& $\strictdwtraceinclusion$ 	& 	$\tickNO$	&	$\tickNO$	&	$\tickNO$	&	$\tickNO$	&	$\tickNO$	\\
							& $\dwtraceinclusion$ 		& 	$\tickNO$	&	$\tickNO$	&	$\tickNO$	&	$\tickNO$	&	$\tickNO$	\\
	  \hline
	  
	  \multirow{6}{*}{$\subseteq^{\mathsf{up}}$}	& $\id$				& 	$\tickNA$	&	$\tickOK$	&	$\tickNA$	&	$\tickNO$	&	$\tickNA$	\\
							& $\strictdwsim$ 		& 	$\tickNA$	&	$\tickOK$	&	$\tickNA$	&	$\tickNO$	&	$\tickNA$	\\
							& $\dwsim$ 			& 	$\tickNA$	&	$\tickOK$	&	$\tickNA$	&	$\tickNO$	&	$\tickNA$	\\
							& $\strictdwtraceinclusion$	& 	$\tickNA$	&	$\tickNO$	&	$\tickNA$	&	$\tickNO$	&	$\tickNA$	\\
							& $\dwtraceinclusion$ 		& 	$\tickNA$	&	$\tickNO$	&	$\tickNA$	&	$\tickNO$	&	$\tickNA$	\\
      \end{tabular}

  \captionof{figure}{GFP relations $\makeprunerel{\urel(\irel)}{\drel}$ for tree automata.
  Relations which are GFP are marked with $\tickOK$, those which are not are marked with $\tickNO$
  and $\tickNA$ is used to mark relations where the test does not apply due to them being reflexive
  (and therefore not asymmetric).}
  \label{fig:GFP_relations_trees}
\end{minipage}
\hspace{4mm}
\begin{minipage}{4.7cm}
\vspace{23mm}
      \begin{tabular}{c|c|c}
	  \multicolumn{2}{c}{$R$}		&		\\
 	  \cline{1-2}
	  
	  \multicolumn{2}{c|}{$\dwtraceinclusion$}			& 	$\tickOK$	\\
	  \hline
	  
	  \multicolumn{2}{c|}{$\dwsim$}				& 	$\tickOK$	\\
	  \hline
	  
	  \multirow{6}{*}{$\sqsubseteq^\mathsf{up}$}		& $\id$		 		& 	$\tickOK$	\\
								& $\strictdwsim$ 		& 	$\tickNA$	\\
								& $\dwsim$ 			& 	$\tickNO$	\\
								& $\strictdwtraceinclusion$ 	& 	$\tickNA$	\\
								& $\dwtraceinclusion$ 		& 	$\tickNO$	\\
	  \hline
	  
	  \multirow{6}{*}{$\subseteq^\mathsf{up}$}		& $\id$		 		& 	$\tickOK$	\\
								& $\strictdwsim$ 		& 	$\tickNA$	\\
								& $\dwsim$ 			& 	$\tickNO$	\\
								& $\strictdwtraceinclusion$ 	& 	$\tickNA$	\\
								& $\dwtraceinclusion$ 		& 	$\tickNO$	\\
      \end{tabular}
  \captionof{figure}{GFQ relations $R$ for tree automata.
	    Relations which are GFQ are marked with $\tickOK$ and those which
            are not are marked with $\tickNO$.
            The relations marked with $\tickNA$ are not even reflexive in
            general (unless all transitions are linear; in this case we have
            a word automaton and these relations are the same as
            $\upsim{\id}$ and $\uptraceinclusion{\id}$, respectively).}
  \label{fig:GFQ_relations_trees}
\end{minipage}
\end{figure}


\newpage
\appendix
\section{Examples and Counterexamples for Tree Automata}\label{sec:appendix_A}

\newcommand{\TreesDist}{\quad}
\newcommand{\LevelDist}{0.8cm}
\newcommand{\SiblDist}{1cm}
\begin{small}
\begin{figure}[ht]
\begin{minipage}{12.2cm}
\subfloat[
	$A_{BU} = (\Sigma, Q, \delta_{BU}, F = \{q_1,q_2\})$.
]{
	\begin{tikzpicture}[on grid, node distance=1cm and 1.5cm]
		\tikzstyle{vertex} = [smallstate] 
		\path node [vertex,accepting] (q1) {$q_1$};
		\path node [vertex,accepting] (q2) [right = 1.5cm of q1] {$q_2$};
		\path node [vertex] (q4) [below right= 1.3 and 1cm of q2] {$q_4$};
		
		\path node [vertex] (q3) [below left = 2 and 1cm of q1] {$q_3$};
		
		\path node [vertex] (q5) [below right = 1 and 1cm of q4] {$q5$};
		
		\path node [vertex,initial] (f) [below right = 3 and 1cm of q1] {$\spstate$};
 
		\draw [-triangle 45] (q1.south) node (tipq1) {};
		\draw [-triangle 45] (q2.south) node (tipq2) {};
	    
 		\path
 			(tipq1.south) edge node {} (q3)
 			(tipq1.south) edge node {} (q4)
 			(tipq2.south) edge node {} (q3)
 			(tipq2.south) edge node {} (q4);
 			
 		\path[->] 			
 			(q3) edge [dashed] node [above left] {$c$} (q4)
 			(q5) edge node [above right] {$c$} (q4)
  			(f) edge node [above left] {$d$} (q5)  			
  			(f) edge node [above right] {$e$} (q3);
  			
  		\path
			(q1) -- node [below right = 0.5 and 0.5cm] {$b$} (q2)
			(q1) -- node [below left = 0.5 and 0.5cm] {$a$} (q2);		
					
	\end{tikzpicture}
\label{buta:notationExample}
	}
\quad
\subfloat[
	$A_{TD} = (\Sigma, Q, \delta_{TD}, I = \{q_1,q_2\})$.
]{
\begin{tikzpicture}[on grid, node distance=1cm and 1.5cm]
		\tikzstyle{vertex} = [smallstate] 
		
		\path node [vertex,initial] (q1) {$q_1$};
		\path node [vertex,initial] (q2) [right = 1.5cm of q1] {$q_2$};
		\path node [vertex] (q4) [below right= 1.3 and 1cm of q2] {$q_4$};
		
		\path node [vertex] (q3) [below left = 2 and 1cm of q1] {$q_3$};
		
		\path node [vertex] (q5) [below right = 1 and 1cm of q4] {$q5$};
		
		\path node [vertex,accepting] (f) [below right = 3 and 1cm of q1] {$\spstate$};
 
 		\path[->]
 			(q1.south) edge node {} (q3)
 			(q1.south) edge node {} (q4)
 			(q2.south) edge node {} (q3)
 			(q2.south) edge node {} (q4)
 			
 			(q4) edge [dashed] node [above left] {$c$} (q3)
 			(q4) edge node [above right] {$c$} (q5)
 			
  			(q5) edge node [above left] {$d$} (f)  			
  			(q3) edge node [above right] {$e$} (f);
  			
  		\path
			(q1) -- node [below right = 0.4 and 0.5cm] {$b$} (q2)
			(q1) -- node [below left = 0.4 and 0.5cm] {$a$} (q2);

	\end{tikzpicture}
	\label{tdta:notationExample}
}
\vspace{10pt}
\end{minipage}
\begin{minipage}{12.2cm}
\begin{center}
\subfloat[The trees accepted by $A_{BU}$ and $A_{TD}$.]
{
\begin{tikzpicture}
   [level distance=\LevelDist,
   level 1/.style={sibling distance=\SiblDist}]
  \node {a}
    child {node {e}}
    child {node {c}
      child {node {d}}};
\end{tikzpicture}	\TreesDist
\begin{tikzpicture}
   [level distance=\LevelDist,
   level 1/.style={sibling distance=\SiblDist}]
  \node {a}
    child {node {e}}
    child {node {c}
      child {node {e}}};
\end{tikzpicture}	\TreesDist
\begin{tikzpicture}	
   [level distance=\LevelDist,
   level 1/.style={sibling distance=\SiblDist}]
  \node {b}
    child {node {e}}
    child {node {c}
      child {node {d}}};
\end{tikzpicture}	\TreesDist
\begin{tikzpicture}
   [level distance=\LevelDist,
   level 1/.style={sibling distance=\SiblDist}]
  \node {b}
    child {node {e}}
    child {node {c}
      child {node {e}}};
\end{tikzpicture}
}
\end{center}
\captionof{figure}{Let $\Sigma$ be a ranked alphabet such that $\Sigma_0 = \{d,e\}$, $\Sigma_1 = \{c\}$ and $\Sigma_2 = \{a,b\}$. 
	Consider the BUTA $A_{BU}$ and the TDTA $A_{TD}$, where $Q=\{q_1,\ldots,q_5\}$ and
	$\delta_{BU} = \{ \langle\spstate,e,q_4\rangle, \langle\spstate,d,q_5\rangle, \langle q_4,c,q_3\rangle ,
	\langle q_5,c,q_3\rangle, \langle q_3q_4,a,q_1\rangle, \langle q_3q_4,b,q_2\rangle \}$.
	$A_{TD}$ is obtained from $A_{BU}$ by reversing the transition rules in $\delta_{BU}$ and by swapping the 
	roles of the accepting and the final states.
	The language accepted by the automata is $L = \{a(e,c(d)), a(e,c(e)), b(e,c(d)), b(e,c(e)) \}$,
	as represented in c).}
	\label{fig:notationExample}
\end{minipage}
\begin{minipage}{12.2cm}
\vspace{13pt}
\begin{center}
\subfloat[
	Automaton $A_1$.
]{
	\begin{tikzpicture}[on grid, node distance= .6cm and 1.4cm]
		\tikzstyle{vertex} = [smallstate] 

		\path node [vertex, initial] (i) {};
		
		\path node [vertex,accepting] (p) [right = of i] {};
		\path node [vertex,accepting] (q) [right = of p] {};
		
		\path[->]
			(i) edge node [above left] {$a$} (p)
			(p) edge node [above left] {$a$} (q);
			
		\begin{pgfonlayer}{background}
		\end{pgfonlayer}
	\end{tikzpicture}
}
$\quad \quad $
\subfloat[
	Automaton $A_2$. 
]{
	\begin{tikzpicture}[on grid, node distance= .6cm and 1cm]
		\tikzstyle{vertex} = [smallstate]

		\path node [vertex, initial] (i) {};
		\path node [vertex,accepting] (p) [above right = of i] {};
		\path node [vertex] (q) [below right = of i] {};
		\path node [vertex,accepting] (r) [right = of q] {};

		\path[->]
			(i) edge node [above left] {$a$} (p)
			(i) edge node [above right] {$a$} (q)
			(q) edge node [above left] {$a$} (r);
	\end{tikzpicture}
}
\end{center}
\captionof{figure}{An example of two NFAs for which language inclusion holds but trace inclusion does not: 
the trace for $aa$ does not preserve acceptance in the second state in $A_2$.}
\label{fig:TrInclExample}
\end{minipage}
\end{figure}
\end{small}

In Figure~\ref{fig:notationExample}, we present two examples of TA, a BUTA and a TDTA, 
where the second is obtained from the first.
We draw the automata vertically, either bottom-up or top-down (depending on if it is a BUTA or a TDTA),
to make the reading of an input tree more natural. 
The example in Figure~\ref{fig:TrInclExample} shows that language inclusion on NFAs 
(and, more generally, on TA) does not imply trace inclusion.
\vspace{-20pt}
%
%
%
%
%
%
%


\newcommand{\htabA}{2.cm}	
\newcommand{\htabB}{1.5cm}
\newcommand{\htabC}{1.5cm}
\newcommand{\htabD}{0.5cm}
\newcommand{\htabE}{1cm}
\newcommand{\vtabA}{3cm}	
\newcommand{\vtabB}{4cm}
\newcommand{\vtabC}{0.9cm}
\newcommand{\vtabD}{1.0cm}
\newcommand{\vtabE}{1.0cm}
\newcommand{\vtabF}{1.5cm}
\newcommand{\vtabG}{1.3cm}
\newcommand{\bend}{20}
\newcommand{\hspacemacro}{150pt}

\begin{tiny}
\begin{sidewaysfigure}
\begin{center}
\subfloat[
Consider the macros $T_1$ and $T_2$. 
They are used in $b)$ to introduce new transitions and new states in the third level of the automaton as here defined.
Note that each of the intermediate states in $T_2$ is smaller w.r.t
$\strictdwtraceinclusion$ than the intermediate state in $T_1$.
Thus, when the first state in $T_2$ is smaller w.r.t $\strictupsim{\strictdwsim}$ than the first state in $T_1$,
each of the $d$-transitions in $T_2$ is a blue one w.r.t. the $d$-transition in $T_1$, which is red.
Let us also consider the symbols $x$ and $y$ of rank $1$.
We can extend a macro $T$ by adding transitions from the first state to some new state by $x$, by $y$ or by both 
and we denote these by $T_x$, $T_y$ or $T_{x+y}$, respectively.
]{
\hspace{\hspacemacro}
\begin{tikzpicture}[on grid, node distance=1.5cm and 1cm]
		\tikzstyle{vertex} = [smallstate] 

		\path node [vertex] (a1) {};
		\path node [vertex] (a2) [below = of a1] {};
		\path node [vertex,accepting] (apsi) [below = of a2] {$\psi$};
		
		\path[->]
			(a1) edge node [right] {$d$} (a2)
			(a2) edge node [right] {$b,c$} (apsi);
			
		\path node (label) [left = of a2] {\textbf{$T_1$:}};
		
\end{tikzpicture} \hspace{60pt}
\begin{tikzpicture}[on grid, node distance=1.5cm and 1cm]
		\tikzstyle{vertex} = [smallstate] 

		\path node [vertex] (a1) {};
		\path node [vertex] (a2) [below left = of a1] {};
		\path node [vertex] (a3) [below right = of a1] {};
		\path node [vertex,accepting] (apsi) [below right = of a2] {$\psi$};
		
		\path[->]
			(a1) edge node [left] {$d$} (a2)
			(a1) edge node [right] {$d$} (a3)
			(a2) edge node [left] {$b$} (apsi)
			(a3) edge node [right] {$c$} (apsi);
			
		\path node (label) [left = of a2] {\textbf{$T_2$:}};
		
\end{tikzpicture}
\hspace{\hspacemacro}
}
\end{center}

\subfloat[
We consider $\Sigma_0 = \{b,c\}$, $\Sigma_1 = \{d,x,y,z\}$ and $\Sigma_2 = \{a\}$.
The dashed arrows represent transitions by $z$ from/to some new state.
Each of the six initial states has an $a$-transition to one of the states from $7$ to $12$ on the left
and one of the states from $7'$ to $12'$ on the right.
Any state $n'$ on the right side of the automaton does exactly the same downwardly as the state $n$ on the left side, 
and thus needs not be expanded in the figure.
We abbreviate $\strictupsim{\strictdwsim}$ to simply $\sqsubset^\mathsf{up}$ 
and $\upsim{\strictdwsim}$ to $\sqsubseteq^\mathsf{up}$.
]{
\begin{tikzpicture}[on grid, node distance=1.5cm and 1cm]

		\tikzstyle{vertex} = [smallstate] 
		
		\path node (invis1) {};
		\path node (invis11) [below = 0.4cm of invis1] {};
		\path node [vertex,initial] (q6) [below = 0.4cm of invis1] {$6$};
		\path node [vertex,initial] (q5) [left =\htabA of q6] {$5$};
		\path node [vertex,initial] (q4) [left =\htabA of q5] {$4$};
		\path node [vertex,initial] (q3) [left =\htabA of q4] {$3$};
		\path node [vertex,initial] (q2) [left =\htabA of q3] {$2$};
		\path node [vertex,initial] (q1) [left =\htabA of q2] {$1$};
		\path node [vertex,initial] (q1B) [right = 2.cm of q6] {$1$};

		\path node [vertex] (q12) [below left = \vtabA and 2cm of q4] {$12$};
		\path node (invis3) [left = 4.cm of q12] {};
		\path node [vertex] (q11) [left = \htabB of q12] {$11$};
		\path node (q11invis) [above right = \vtabE and 0.5cm of q11] {};
		\path node [vertex] (q10) [left = \htabC of q11] {$10$};
		\path node [vertex] (q9) [left = \htabB of q10] {$9$};
		\path node (q9invis) [above right = \vtabE and 0.5cm of q9] {};
		\path node [vertex] (q8) [left = \htabC of q9] {$8$};
		\path node (q8invis) [above right = \vtabE and 0.5cm of q8] {};
		\path node [vertex] (q7) [left = \htabB of q8] {$7$};
		\path node (q7invis) [above right = \vtabE and 0.5cm of q7] {};
		\path node [vertex] (q12B) [left = \htabC of q7] {$12$};
		\path node [vertex] (q12') [right = 1.7cm of q12] {$12'$};
		\path node [vertex] (q7') [right = \htabC of q12'] {$7'$};
		\path node (q7'invis) [below left = \vtabE and 0.5cm of q7'] {};
		\path node [vertex] (q8') [right = \htabB of q7'] {$8'$};
		\path node (q8'invis) [below left = \vtabE and 0.5cm of q8'] {};
		\path node [vertex] (q9') [right = \htabC of q8'] {$9'$};
		\path node (q9'invis) [below left = \vtabE and 0.5cm of q9'] {};
		\path node [vertex] (q10') [right = \htabB of q9'] {$10'$};
		\path node [vertex] (q11') [right = \htabC of q10'] {$11'$};
		\path node [vertex] (q12'B) [right = \htabB of q11'] {$12'$};
		\path node (q11'invis) [below left = \vtabE and 0.5cm of q11'] {};
		
		\path node [vertex] (q13) [below = \vtabA and \htabB of q12B] {$13$};
		\path node [vertex] (q14) [right = \htabC of q13] {$14$};
		\path node [vertex] (q15) [right = of q14] {$15$};
		\path node [vertex] (q16) [right = \htabC of q15] {$16$};
		\path node [vertex] (q17) [right = 1.2cm of q16] {$17$};
		\path node [vertex] (q18) [right = \htabB of q17] {$18$};
		\path node (q18invis) [above left = \vtabE and 0.5cm of q18] {};
		\path node [vertex] (q19) [right = of q18] {$19$};
		\path node [vertex] (q20) [right = \htabC of q19] {$20$};
		\path node [vertex] (q21) [right = of q20] {$21$};
		\path node [vertex] (q22) [right = \htabC of q21] {$22$};
		\path node [vertex] (q23) [right = 1.2cm of q22] {$23$};
		\path node [vertex] (q24) [right = \htabB of q23] {$24$};
		\path node (q24invis) [above right = \vtabE and 0.5cm of q24] {};
		
 		\path[->] 			
 			(q1.south) edge node {} (q7.north)
 			(q1.south) edge node {} (q12'.north)
 			(q2.south) edge node {} (q8.north)
 			(q2.south) edge node {} (q7'.north)
 			(q3.south) edge node {} (q9.north)
 			(q3.south) edge node {} (q8'.north)
 			(q4.south) edge node {} (q10.north)
 			(q4.south) edge node {} (q9'.north)
 			(q5.south) edge node {} (q11.north)
 			(q5.south) edge node {} (q10'.north)
 			(q6.south) edge node {} (q12.north)
 			(q6.south) edge node {} (q11'.north)
 			
 			(q7.south) edge [blue, very thin] node {} (q13.north)
 			(q7.south) edge [blue, very thin] node {} (q15.north)
 			(q8.south) edge [red, very thick] node {} (q14.north)
 			(q8.south) edge [red, very thick] node {} (q16.north)
 			(q9.south) edge node {} (q17.north)
 			(q9.south) edge node {} (q19.north)
 			(q10.south) edge node {} (q18.north)
 			(q10.south) edge node {} (q20.north)
 			(q11.south) edge node {} (q21.north)
 			(q11.south) edge node {} (q23.north)
 			(q12.south) edge node {} (q22.north)
 			(q12.south) edge node {} (q24.north) 		
 			
 			(q7.north west) edge [bend left = \bend,dashed] node [above] {$z$} (q7invis)
 			(q8invis) edge [bend right = \bend,dashed] node [above] {$z$} (q8.north west)
 			(q9.north west) edge [bend left = \bend,dashed] node [above] {$z$} (q9invis)
 			(q11.north west) edge [bend left = \bend,dashed] node [above] {$z$} (q11invis)
 			(q7') edge [bend left = \bend,dashed] node [left] {$z$} (q7'invis)
 			(q8'invis) edge [bend right = \bend,dashed] node [left] {$z$} (q8')
 			(q9') edge [bend left = \bend,dashed] node [left] {$z$} (q9'invis)
 			(q11') edge [bend left = \bend,dashed] node [left] {$z$} (q11'invis)
 			(q18invis) edge [bend right = \bend,dashed] node [left] {$z$} (q18)
 			(q24invis) edge [bend left = \bend,dashed] node [left] {$z$} (q24);
 					
  		\path
			
			(q1) -- node [below right = \vtabC and 0.1cm of q1] {$a$} (q1)
			(q2) -- node [below right = \vtabC and 0.1cm of q2] {$a$} (q2)
			(q3) -- node [below right = \vtabC and 0.1cm of q3] {$a$} (q3)
			(q4) -- node [below right = \vtabC and 0.1cm of q4] {$a$} (q4)
			(q5) -- node [below = \vtabC of q5] {$a$} (q5)
			(q6) -- node [below = \vtabC of q6] {$a$} (q6)
			(q7) -- node [below left = \vtabD and 0.1cm of q6, blue] {$a$} (q7)
			(q8) -- node [below left = \vtabD and 0.1cm of q8, red] {$a$} (q8)
			(q9) -- node [below right = \vtabD and 0.3cm of q9] {$a$} (q9)
			(q10) -- node [below right = 1.1cm and 0.3cm of q10] {$a$} (q10)
			(q11) -- node [below right = 1.2cm and 0.8cm of q11] {$a$} (q11)
			(q12) -- node [below right = 1.2cm and 0.8cm of q12] {$a$} (q12)
						
			(q1) -- node [left = 0.2cm of q2] {$\sqsubseteq^\mathsf{up}$} (q2)
			(q2) -- node [left = 0.2cm of q2] {$\sqsubseteq^\mathsf{up}$} (q3)
			(q3) -- node [left = 0.2cm of q2] {$\sqsubseteq^\mathsf{up}$} (q4)
			(q4) -- node [left = 0.2cm of q2] {$\sqsubseteq^\mathsf{up}$} (q5)
			(q5) -- node [left = 0.2cm of q2] {$\sqsubseteq^\mathsf{up}$} (q6)
			(q6) -- node [left = 0.2cm of q1B] {$\sqsubseteq^\mathsf{up}$} (q1B)
			
			(q12B) -- node [] {$\strictdwsim$} (q7)		
			(q7) -- node [] {$\sqsubset^\mathsf{up}$} (q8)
			(q8) -- node [] {$\strictdwsim$} (q9)
			(q9) -- node [] {$\sqsubseteq^\mathsf{up}$} (q10)
			(q10) -- node [] {$\strictdwsim$} (q11)
			(q11) -- node [] {$\sqsubseteq^\mathsf{up}$} (q12)
			(q12') -- node [] {$\strictdwsim$} (q7')
			(q7') -- node [] {$\sqsubset^\mathsf{up}$} (q8')
			(q8') -- node [] {$\strictdwsim$} (q9')
			(q9') -- node [] {$\sqsubseteq^\mathsf{up}$} (q10')
			(q10') -- node [] {$\strictdwsim$} (q11')
			(q11') -- node [] {$\sqsubseteq^\mathsf{up}$} (q12'B)
			
			(q13) -- node [] {$\strictdwtraceinclusion$} (q14)
			(q15) -- node [] {$\strictdwtraceinclusion$} (q16)
			(q17) -- node [] {$\sqsubset^\mathsf{up}$} (q18)
			(q19) -- node [] {$\strictdwsim$} (q20)
			(q21) -- node [] {$\strictdwsim$} (q22)
			(q23) -- node [] {$\sqsubset^\mathsf{up}$} (q24)		
			
			(q13) -- node [below = \vtabE of q13] {$T_1 + y$} (q13)
			(q13) -- node [below = \vtabF of q13] {\rotatebox{-90}{$\cdots$}} (q13)
			(q14) -- node [below = \vtabE of q14] {$T_2 + x + y$} (q14)
			(q14) -- node [below = \vtabF of q14] {\rotatebox{-90}{$\cdots$}} (q14)
			(q15) -- node [below = \vtabE of q15] {$T_1$} (q15)
			(q15) -- node [below = \vtabF of q15] {\rotatebox{-90}{$\cdots$}} (q15)
			(q16) -- node [below = \vtabE of q16] {$T_2 + x$} (q16)
			(q16) -- node [below = \vtabF of q16] {\rotatebox{-90}{$\cdots$}} (q16)
			(q17) -- node [below = \vtabE of q17] {$T_2 + x + y$} (q17)
			(q17) -- node [below = \vtabF of q17] {\rotatebox{-90}{$\cdots$}} (q17)
			(q18) -- node [below = \vtabE of q18] {$T_1$} (q18)
			(q18) -- node [below = \vtabF of q18] {\rotatebox{-90}{$\cdots$}} (q18)
			(q19) -- node [below = \vtabE of q19] {$T_2 + x$} (q19)
			(q19) -- node [below = \vtabF of q19] {\rotatebox{-90}{$\cdots$}} (q19)
			(q20) -- node [below = \vtabE of q20] {$T_2 + x + y$} (q20)
			(q20) -- node [below = \vtabF of q20] {\rotatebox{-90}{$\cdots$}} (q20)
			(q21) -- node [below = \vtabE of q21] {$T_1$} (q21)
			(q21) -- node [below = \vtabF of q21] {\rotatebox{-90}{$\cdots$}} (q21)
			(q22) -- node [below = \vtabE of q22] {$T_1 + y$} (q22)
			(q22) -- node [below = \vtabF of q22] {\rotatebox{-90}{$\cdots$}} (q22)
			(q23) -- node [below = \vtabE of q23] {$T_2 + x + y$} (q23)
			(q23) -- node [below = \vtabF of q23] {\rotatebox{-90}{$\cdots$}} (q23)
			(q24) -- node [below = \vtabE of q24] {$T_1$} (q24)
			(q24) -- node [below = \vtabF of q24] {\rotatebox{-90}{$\cdots$}} (q24);
  		
		
	\end{tikzpicture}
}	
\caption{$P(\strictupsim{\strictdwsim},\strictdwtraceinclusion)$ is not GFP
since the automaton presented in $b)$ cannot read the full binary tree of $a$'s with height $3$
without using a blue transition:
a run starting in state $1$ encounters a blue transition from $7$, 
as illustrated in the figure;
and since $7'$ and $8'$ do the same downwardly as $7$ and $8$, respectively,
and since $7' \sqsubset^\mathsf{up} 8'$, 
we have that there is a blue transition from $7'$ as well, and so $2$ cannot be used either;
since $17 \sqsubset^\mathsf{up} 18$ we have that, 
as explained in $a)$, there is a blue transition departing from $17$, 
thus a run starting at $3$ too cannot be used;
and since $9$ is downwardly imitated by $9'$,
a run using this state finds a blue transition as well, and so $4$ is not safe;
since $23 \sqsubset^\mathsf{up} 24$, a blue transition from $23$ exists and so $5$ cannot be used;
finally, since  $11$ is imitated by $11'$, 
we have that a run using this state encounters a blue transition as well, 
and so $6$ too is not safe.
}\label{GFPcountEx:strictupsim(dwsim)strictdwtraceincl}
\end{sidewaysfigure}
\end{tiny}

\newpage

\begin{figure}[htbp]
 \begin{center}
  \subfloat[Automaton $A$.]{
      \begin{tikzpicture}[on grid]
		\tikzstyle{vertex} = [smallstate] 

		\path node [vertex,initial] (i)  {};
		\path node [vertex] (q) [below left = 1.5cm and 0.75cm of i] {$q$};
		\path node [vertex] (p) [left = 1.5cm of q] {$p$};
		\path node [vertex] (r) [below right = 1.5cm and 0.75cm of i] {$r$};
		\path node [vertex] (s) [right = 1.5cm of r] {$s$};
		\path node [vertex,accepting] (f) [below = 2.75cm of i] {$\spstate$};
		
		\path[->]
		    (i.south west) edge node [right = 0.2cm] {$c$} (p)
		    (i.south west) edge node {} (q.north)
		    (i.south) edge node [right = 0.1cm] {$c$} (q.north east)
		    (i.south) edge node [] {} (r.north west)
		    (i.south east) edge node [right] {$c$} (r.north)
		    (i.south east) edge node {} (s)
		    (p) edge node {} (f)
		    (q) edge node {} (f)
		    (r) edge node {} (f)
		    (s) edge node {} (f)
		    ;
		    
		\path
		    (p) -- node [below right = 0.2 and 0.1cm] {$a$} (q)
		    (q) -- node [below left = 0.2 and 0.1cm] {$a$} (r)
		    (r) -- node [below left = 0.2 and 0.5cm] {$b$} (s)
		    (r) -- node [below right = 0.2 and 0.2cm] {$b$} (s)
		    ;		    
		
      \end{tikzpicture}
  }	\quad
  \subfloat[Automaton $\A/\equiv$.]{
      \begin{tikzpicture}[on grid]
		\tikzstyle{vertex} = [smallstate] 

		\path node [vertex,initial] (i)  {};
		\path node [vertex] (qr) [below = 1.5cm of i] {$\{q,r\}$};
		\path node [vertex] (p) [left = 1.5cm of qr] {$\{p\}$};
		\path node [vertex] (s) [right = 1.5cm of qr] {$\{s\}$};
		\path node [vertex,accepting] (f) [below = 3.cm of i] {$\{\spstate\}$};
		
		\path[->]
		    (i.south west) edge node [right = 0.1cm] {$c$} (p)
		    (i.south west) edge node {} (qr.north west)
		    (i.south) edge node [right] {$c$} (qr.north west)
		    (i.south) edge node [] {} (qr.north east)
		    (i.south east) edge node [right] {$c$} (qr.north east)
		    (i.south east) edge node {} (s)
		    (p) edge node {} (f)
		    (qr) edge node {} (f)
		    (s) edge node {} (f)
		    ;
		    
		    \path
		    (p) -- node [below right = 0.4 and 0.cm] {$a$} (qr)
		    (qr) -- node [below left = 0.4 and 0.1cm] {$a,b$} (s)
		    (qr) -- node [below right = 0.4 and 0.1cm] {$b$} (s)
		    ;		    
      \end{tikzpicture}
  }
 \end{center}
  \caption{$\equiv \,:=\, \upsim{\dwsim \!\!\cap\! \dwsimlarg} \,\cap \upsimlarg{\dwsim \!\!\cap\! \dwsimlarg}$ is not
    GFQ. We are considering $\Sigma_0=\{a,b\}$ and $\Sigma_2=\{c\}$. \\
	    Computing all the necessary relations to quotient $\A$ w.r.t.\ $\equiv$, 
	    we obtain $\dwsim\,=\{(p,q),(r,s)\} = \dwsimlarg$ and 
	    $\upsim{\dwsim \!\!\cap\! \dwsimlarg}=\{(q,r),(r,q)\}$.
	    Thus $\equiv \;= \{(q,r),(r,q)\}$.
	    Computing $\A/\!\!\equiv$, 
	    we verify that $c(b,a)$ is now accepted by the automaton $\A/\!\!\equiv$, while it was not in the language of $\A$.}
  \label{fig:GFQ_counterexample}
\end{figure}

\newpage
\section{Proofs of Theorems}\label{sec:app_proofs}

{\bf\noindent Theorem~\ref{thm:id(id)strictdwtraceinclGFP}.}
For every strict partial order $R \,\subset\; \dwtraceinclusion$,
it holds that $P(\id, R)$ is GFP.
\begin{proof}
 Let $\A' = Prune(\A, P(\id, R))$. We show $L(\A) \subseteq L(\A')$.
 If $\tree \in L(A)$ then there exists an accepting 
 $\tree$-run $\run$ in $\A$.
 We show that there exists an accepting $\tree$-run $\pi'$ in $\A'$.
 
We will call an accepting $\tree$-run $\tilde\run$ in $\A$ \emph{$i$-good} if its first $i$ levels use only transitions of $\A'$. Formally, for every node $\node\in\dom(\tree)$ with $|\node|< i$,
$\langle \tilde\run(\node), \tree(\node), \tilde\run(\node 1) \ldots \tilde\run(\node \#(\tree(\node))) \rangle$ is a transition of $\A'$. 
By induction on $i$, we will show that there exists an $i$-good accepting run on $\tree$ for every $i\leq h(\tree)$. In the base case $i=0$, the claim is trivially true since every accepting $\tree$-run of $\A$, and particularly $\run$, is $0$-good.  

For the induction step, let us assume that the claim holds for some $i$.
  Since $\A$ is finite, 
  for every transition ${\it trans}$ there are only finitely many
  $\A$-transitions ${\it trans}'$ such that
  ${\it trans} \mathrel{P(\id,R)} {\it trans}'$. 
  And since $P(\id,R)$ is transitive and irreflexive,
  for each transition ${\it trans}$ in $\A$ we have that either 
  \textbf{1)} ${\it trans}$ is maximal w.r.t.\ $P(\id,R)$,
  or \textbf{2)} there exists a $P(\id,R)$-larger transition ${\it trans}'$ which is 
  maximal w.r.t.\ $P(\id,R)$. 
  Thus for every state $p$ and every symbol $\sigma$, there exists a transition by $\sigma$ 
  departing from $p$ which is still in $\A'$.

  Therefore, for every $i$-good accepting run $\run^i$ on $\tree$, 
  one easily obtains an accepting run $\run^{i+1}$ which is $(i+1)$-good.
  In the first $i$ levels of $\tree$,
  $\run^{i+1}$ is identical to $\run^{i}$.
  In the $(i+1)$-th level of $\tree$,
  we have that for any transition 
  ${\it trans}=\langle \run^i(\node), \tree(\node), \run^i(\node 1) \ldots \run^i(\node \#(\tree(\node)))\rangle$,
  for $|\node|=i$,
  either ${\it trans}$ is $P(\id,R)$-maximal, 
  and so we take $\run^{i+1}(\node j) := \run^i(\node j)$ for all $1\leq j \leq \#(\tree(\node))$,
  or
  there exists a $P(\id,R)$-larger transition 
  ${\it trans}'= 
  \langle \run^i(\node),\tree(\node),q_1 \ldots
  q_{\#(\tree(\node))}\rangle$ that is $P(\id,R)$-maximal. By the definition of $P(\id,R)$, we have that 
  $(\run^i(\node 1) \ldots \run^i(\node \#(\tree(\node)))) \mathrel{\hat{R}} (q_1 \ldots q_{\#(\tree(\node))})$,
  and we take $\run^{i+1}(\node j) := q_j$ for all $1 \leq j \leq \#(\tree(\node)))$.
  Since ${R} \subset {\dwtraceinclusion}$,
  we have that for every $1\leq j \leq \#\tree(v)$, 
   there is a run $\run_j$ of $A$ such that
  $\tree_{v_j} \stackrel{\run_j} {\Longrightarrow}q_j$. 
  The run $\run^{i+1}$ on $\tree$ can hence be completed from every $q_j$ by the run $\run_j$, which concludes the proof of the induction step.
%

 Since a $h(\tree)$-good run is a run in $\A '$, the theorem is proven.
 \qed                                    
\end{proof}                              

{\bf\noindent Theorem~\ref{thm:strictuptraceincl(id)idGFP}.}
For every strict partial order $R \,\subset\; \uptraceinclusion{\id}$,
it holds that $P(R, \id)$ is GFP.
\begin{proof}
Let $\A' = Prune(\A,P(R, \id))$.
We will show that for every accepting 
run $\run$ of $\A$ on a tree $\tree$, 
there exists an accepting run $\hat\run$ of $\A'$ on $\tree$.

Let us first define some auxiliary notation. 
For an accepting run $\run$ of $\A$ on a tree $\tree$,
${\it bad}(\run)$ is the smallest subtree of $\tree$ which contains \emph{all nodes} $v$ of $\tree$  
where $\run$ uses a transition of $\A-\A'$, i.e., a transition which is not $P(R, \id)$-maximal
(where by $\run$ using a transition at node $v$ we mean that the symbol of the transition is $\tree(v)$, 
$\run(v)$ is the left-hand side of the transition, 
and the vector of $\pi$-values of children of $v$ is its right-hand side).
We will use the following auxiliary claim.
\begin{itemize}
\item[(C)]
For every accepting run $\run$ of $\A$ on a tree $\tree$ with $|{\it bad}(\run)| > 1$,
there is an accepting run $\run'$ of $\A$ on $\tree$ where ${\it bad}(\run')$ 
is a proper subtree of ${\it bad}(\run)$.
\end{itemize}
\smallskip
To prove (C),
assume that $v$ is a leaf of ${\it bad}(\run)$ labeled by a transition $\trans p \sigma r n$.
By the definition of $P(R, \id)$ and by the minimality of ${\it bad}(\run)$, 
there exists a $P(R, \id)$-maximal transition $\tau = \trans {p'} \sigma r n$
where 
$p \mathrel{\strictuptraceinclusion{\id}} p'$. 
Since $p \mathrel{\strictuptraceinclusion{\id}} p'$, it follows from the 
definition of $\strictuptraceinclusion{\id}$
that there exists a run $\run'$ of $\A$ on $\tree$ that differs from $\run$
only in labels of prefixes of $v$ (including $v$ itself) with $\run'(v) = p'$.
In other words, 
${\it bad}(\run')$ differs from ${\it bad}(\run)$ only in that it does not 
contain a certain subtree rooted by some ancestor of $v$. 
This subtree contains at least $v$ itself, since $\run'$ uses the 
$P(R, \id)$-maximal transition $\tau$ to label $v$. 
The tree ${\it bad}(\run')$ is hence a proper subtree of ${\it bad}(\run)$, 
which concludes the proof of (C).

With (C) in hand, we are ready to prove the lemma. 
By finitely many applications of (C), starting from $\run$,
we obtain an accepting run $\hat\run$ on $\tree$ where 
${\it bad}(\hat\run)$ is empty 
(we only need finitely many applications since ${\it bad}(\run)$ is a finite
tree, and every application of (C) yields a run with a strictly smaller ${\it bad}$ subtree). 
Thus 
$\hat\run$ is using only $P(R, \id)$-maximal transitions.
Since $R$ and hence also $P(R, \id)$ are strict p.o., 
$\A' = Prune(\A,P(R, \id))$ contains all $P(R, \id)$-maximal transitions of $\A$, 
which means that $\hat\run$ is an accepting run of $\A'$ on $\tree$. 
\qed
\end{proof}

{\bf\noindent Theorem~\ref{thm:dwtrinclgfq}.}
$\equiv^\mathsf{dw}$ is GFQ.		
\begin{proof}
Let $A' := A/\!\!\equiv^\mathsf{dw}$. It is trivial that $L(A) \subseteq L(A')$. 
For the reverse inclusion, we will show by induction on the height $i$ of $\tree$, that for any tree $\tree$, 
if $\tree \in D_{A'}([q])$ for some $[q] \in [Q]$, then $\tree \in D_{A}(q)$.
This guarantees $L(A') \subseteq L(A)$ since if $[q] \in [I]$ then there is some $q' \in I$
 such that $q' \equiv^\mathsf{dw} q$ and thus, by the definition of $\equiv^\mathsf{dw}$, $D_{A}(q') = D_{A}(q)$.

In the base case $i=1$, $t$ is a leaf-node $\sigma$, for some $\sigma \in \Sigma$. 
By hypothesis, $t \in L(A')$. So there exists $[q] \in [I]$ such that $t \Longrightarrow_{A'} [q]$. 
So $\langle [q], \sigma, [\spstate] \rangle \in \delta_{A'}$. 
Since $[\spstate] = \{\spstate\}$, there exists $q' \in [q]$ 
such that $\langle q', \sigma, \spstate \rangle \in \delta_A$.
Since $[q] \in [I]$ there is some $q'' \in I$ with $q'' \equiv^\mathsf{dw} q \equiv^\mathsf{dw} q'$.
We have $\tree \in D_{A}(q') = D_{A}(q'') \subseteq L(A)$.

Let us now consider $i>1$. 
Let $\sigma$ be the root of the tree $t$, and let $t_1,t_2,\ldots,t_{n}$, where $n=\#(\sigma)$,
denote each of the immediate subtrees of $t$. 
As we assume $t \in L(A')$,
there exists $[q] \in [I]$ such that $\langle [q], \sigma, [q_1][q_2]\ldots [q_n] \rangle \in \delta_{A'}$, 
for some $[q_1],[q_2],\ldots,[q_n] \in [Q]$, such that $t_i \in D_{A'}([q_i])$ for every $i$.
By the definition of $\delta_{A'}$, there are 
$q_1' \in [q_1]$, $q_2' \in [q_2]$, $\ldots$, $q_n' \in [q_n]$ and $q' \in [q]$,
such that $\langle q', \sigma, q_1'q_2'\ldots q_n' \rangle \in \delta_A$.
By induction hypothesis, we obtain $t_i \in D_A(q_i)$ for every $i$. 
Since $q_i \equiv^\mathsf{dw} q_i'$, it follows that $t_i \in D_A(q_i')$ for every $i$ 
and thus $t \in D_A(q')$. By $q \equiv^\mathsf{dw} q'$, we conclude that $t \in D_A(q)$.
\qed
\end{proof}

{\bf\noindent Theorem~\ref{thm:upidtrinclgfq}.}
$\equiv^\mathsf{up}\!\!(\id)$ is GFQ.
\begin{proof}
 Let $\equiv \;:=\; \equiv^\mathsf{up}\!\!(\id)$ and $A' := A/\!\equiv$. It is trivial that $L(A) \subseteq L(A')$.
 For the reverse inclusion, we will show, by induction on the height $h$ of $\tree$, that for any tree $t$, 
 if $t \in D_{A'}([q])$ for some $[q] \in [Q]$, 
 then $t \in D_{A}(q')$ for some $q' \in [q]$.
 This guarantees $L(A') \subseteq L(A)$ since if $[q] \in [I]$ then,
 given that $\equiv$ preserves the initial states, $q' \in I$.
 
 In the base case $h=1$, the tree is a leaf-node $\sigma$, for some $\sigma \in \Sigma$. 
 By hypothesis, $t \in L(A')$. 
 So there exists a $[q] \in [I]$ such that $t \Longrightarrow_{A'} [q]$, 
 and so $\langle [q], \sigma, [\spstate] \rangle \in \delta_{A'}$. 
 By the definition of $\delta_{A'}$ and since $[\spstate] = \{\spstate\}$ ($\equiv$ preserves acceptance), 
 we have that there exists $q' \in [q]$ such that $\langle q', \sigma, \spstate \rangle \in \delta_A$, 
 and hence $t \Longrightarrow_A q'$.
 
 Let us now consider $h>1$. 
 As we assume $\tree \in D_{A'}([q])$, there must exist a transition
 $\langle [q], \sigma, [q_1]\ldots [q_n] \rangle \in \delta_{A'}$, 
 for $n=\#\sigma$ and some $[q_1],\ldots,[q_n] \in [Q]$ such that $\tree_i \in D_{A'}([q_i])$
 for every $i:0 \leq i \leq n$,
 where the $\tree_i$s are the subtrees of $\tree$.
 We define the following auxiliary notion: 
 a transition $\langle r, \sigma, r_1 \ldots r_n \rangle$ of $\A$ satisfying $r \in [q]$
 and $\forall_{1 \leq k \leq n} \ldotp r_k \in [q_k]$ 
 is said to be $j$-good iff 
 $\forall_{1 \leq k \leq j} \ldotp \tree_k \in D_\A(r_k)$. 
 We will use induction on $j$ to show that there is a $j$-good transition for any $j$,
 which implies that there is some state $\hat r\in [q]$ such that $\tree \in
 D_{\A}(\hat r)$.

 The base case is $j=0$. By the definition of $\delta_{A'}$ and the fact that 
 $\langle [q], \sigma, [q_1] \ldots [q_n] \rangle \in \delta_{A'}$,
 there exist $q_1' \in [q_1]$, $\ldots$, $q_n' \in [q_n]$ and $q' \in [q]$ 
 such that $\langle q', \sigma, q_1'\ldots q_n' \rangle \in \delta_A$.
 This transition is trivially $0$-good. 

 To show the induction step, assume a transition $\mathit{trans} = \langle r, \sigma, r_1 \ldots r_n \rangle$ that is $j$-good for $j\geq 0$, i.e., 
 each $r_i$ is in $[q_i]$, $r\in [q]$, and $\forall_{1 \leq i \leq j}\ldotp \tree_i\in D_A(r_i)$.
 By the hypothesis of the outer induction on $h$, 
 there is $r_{j+1}' \in [r_{i+1}]$ such that $\tree_{j+1} \in D_\A(r_{j+1}')$.
 Notice that $r_{j+1} \equiv r'_{j+1}$.
 Since $\mathit{trans}$ is a transition of $A$, 
 there is a run $\run'$ of $A$ on a tree $\tree'$ of the height 1 with the root symbol $\sigma$, and where $\run'(1)=r_1,\ldots, \run'(n)=r_n$, and $\run'(\epsilon) = r$.
 Since $r_{j+1} \equiv r'_{j+1}$,
 then, by the definition of $\equiv$, there is another $\tree'$-run $\run''$ such that 
 $r' = \run''(\epsilon)\in [q]$,
 $\run''(j+1) = r'_{j+1}$,
 and $\forall_{i \neq j+1} \ldotp \run''(i) 
 = \run'(i) = r_i$.
 This run uses the transition 
 $\mathit{trans}' = \langle r', \sigma, r_1 \ldots r_j r'_{j+1} r_{j+2} \ldots r_n \rangle$ in $\A$.
 Since $\mathit{trans}$ is $j$-good and $\tree_{j+1} \in D_\A(r'_{j+1})$,
 we have that $\mathit{trans}'$ is $(j+1)$-good.
 This concludes the inner induction on $j$,
 showing that there exists an $n$-good transition.
 Hence $\tree\in D_A(\hat r)$ for some $\hat r \in [q]$, which proves the outer induction on the height $h$ of the tree, concluding  the whole proof.
 \qed
\end{proof}

\newpage
\section{Combined Preorder}\label{sec:app_combinedPreorder}

In~\cite{lukas:framework2009},
the authors introduce the notion of \emph{combined preorder} on an automaton
and prove that its induced equivalence relation is GFQ.
Let $\oplus$ be an operator defined as follows:
given two preorders $H$ and $S$ over a set $Q$,
for $x,y \in Q$, $x (H \oplus S) y$ iff
$(i)\; x(H \circ S) y$ and
$(ii)\; \forall z \in Q \,:\, y H z \,\Longrightarrow\, x (H \circ S) z$.
Let $D$ be a downward simulation preorder and $U$ an upward simulation preorder induced by $D$.
A combined preorder $W$ is defined as $W \,=\, D \,\oplus\, U^{-1}$. 
Since we have $D \,\oplus\, U^{-1} \;\subseteq\; D \,\circ\, U^{-1}$,
for any states $x,y$ such that $x (D \,\oplus\, U^{-1}) y$,
there exists a state $z$, called a mediator, 
such that $x \,D\, z$ and $y \,U\, z$.

In the following, 
Lemmas~\ref{lem:combinedpreorder_lemma1} and \ref{lem:combinedpreorder_lemma2} are used
by Theorem~\ref{thm:combinedpreorder} to show that any quotienting with the equivalence
relation induced by a combined preorder is subsumed by Heavy(1,1).
We use the maximal downward simulation $\dwsim$ and the maximal upward simulation $\upsim{\dwsim}$
in our proof.
Note that any automaton $\A$ which has been reduced with Heavy(1,1) satisfies 
$(1)\; \A \,=\, \A/(\dwsim \!\cap \dwsimlarg) \,=\, \A/(\upsim{\id} \cap \upsimlarg{\id})$
due to the repeated quotienting,
and $(2)\; \A \,=\, Prune(\A, P(\upsim{\dwsim}, \strictdwsim))$ due to the repeated pruning.

\begin{lemma}
\rm{Let $\A$ be an automaton and $p$ and $q$ two states.
If  1) $\A \,=\, \A/(\dwsim \cap \dwsimlarg)$
and 2) $\A \,=\, Prune(\A, P(\upsim{\dwsim}, \strictdwsim))$,
then $(p \dwsim q \land p (\upsim{\dwsim}) q) \,\Longrightarrow\, p=q$.}
\label{lem:combinedpreorder_lemma1}
\end{lemma}
\begin{proof}
From 1) it follows that $\dwsim$ is antisymmetric,
so if $p \dwsim q$ then $p \strictdwsim q \lor p=q$.

From $p (\upsim{\dwsim}) q$,
it follows that for any transition $\langle p', \sigma, p_1 \ldots p_i \ldots p_{\#(\sigma)} \rangle$
with $p_i = p$
there exists a transition $\langle q', \sigma, q_1 \ldots q_i \ldots q_{\#(\sigma)} \rangle$
with $q_i = q$ such that $p' (\upsim{\dwsim}) q'$ and 
$p_j \dwsim q_j$ for all $j: 1\leq j\neq i \leq \#(\sigma)$.
From $p = p_i \dwsim q_i = q$,
we have that $(p_1 \ldots p \ldots p_{\#(\sigma)}) \,\hat{\sqsubseteq}^{\mathsf{dw}}\, (q_1 \ldots q \ldots q_{\#(\sigma)})$.
From 2) it follows that there is no $k: 1 \leq k \leq \#(\sigma)$ such that $p_k \strictdwsim q_k$.
In particular, $\neg(p \sqsubset^{\mathsf{dw}} q)$. 
Thus we conclude that $p=q$.
 \qed 
\end{proof}

\begin{lemma}\label{lem:combinedpreorder_lemma2}
\rm{Let $\A$ be an automaton and $p$ and $q$ two states.
If $\A \,=\, \A/(\dwsim \cap \dwsimlarg)$, 
then $(p \upsim{\dwsim} q) \land (q \upsim{\dwsim} p) \Longrightarrow (p \upsim{\id} q) \land (q \upsim{\id} p)$.}
\end{lemma}
\begin{proof}
Since $\A = \A/(\dwsim \cap \dwsimlarg)$,
for any two states $x$ and $y$ we have that $(x \dwsim y) \,\Longrightarrow\, (x \sqsubset^{\mathsf{dw}} y \lor x=y)$.

Let $p$ and $q$ be states s.t. $p \upsim{\dwsim} q$ and $q \upsim{\dwsim} p$.
By the definition of $\upsim{\dwsim}$ it follows that for any transition
$\langle p', \sigma, p_1 \ldots p_i \ldots p_{\#(\sigma)} \rangle$ with $p_i = p$ there exists
a transition $\langle q', \sigma, q_1 \ldots q_i \ldots
q_{\#(\sigma)} \rangle$ with 
$p' \upsim{\dwsim} q'$ and
$q_i = q$
such that for any $j:1 \leq j\neq i \leq \#(\sigma) \cdot p_j \dwsim q_j$,
and vice-versa. We can thus construct an infinite sequence of matching
transitions where, for every index $j\neq i$, the sequence of states at component
$j$ is $\dwsim$-increasing.  
However, since $\A$ only has a finite number of states (and transitions),
all these sequences must converge to some equivalence class w.r.t.\
$\dwsim \cap \dwsimlarg$.
Thus, for any transition
$\langle p', \sigma, p_1 \ldots p_i \ldots p_{\#(\sigma)} \rangle$ with $p_i = p$ there exists
a transition $\langle q', \sigma, q_1 \ldots q_i \ldots
q_{\#(\sigma)} \rangle$ with 
$p' \upsim{\dwsim} q'$ and
$q_i = q$
such that for any $j:1 \leq j\neq i \leq \#(\sigma) \cdot 
p_j \dwsim q_j \wedge q_j \dwsim p_j$,
and vice-versa.
However, since $\A = \A/(\dwsim \cap \dwsimlarg)$, we obtain that $p_j = q_j$
for $j:1 \leq j\neq i \leq \#(\sigma)$.
By repeating the same argument for the new pair of states $p'$ and $q'$, we get
that $(p \upsim{\id} q) \land (q \upsim{\id} p)$ as required.
Hence $(\upsim{\dwsim} \cap (\upsim{\dwsim})^{-1}) \;\subseteq\; (\upsim{\id} \cap (\upsim{\id})^{-1})$.
\qed 
\end{proof}

\begin{theorem}
\rm{ Let $\A$ be an automaton such that: 
\\ $(1)\; \A \,=\, \A/(\dwsim \!\cap \dwsimlarg) \,=\, \A/(\upsim{\id} \cap \upsimlarg{\id})$,
and \\ $(2)\; \A \,=\, Prune(\A, P(\upsim{\dwsim}, \strictdwsim))$.
\\ Then $\A = \A / (W \cap W^{-1})$, where $W =\;\; \dwsim \oplus (\upsim{\dwsim})^{-1}$.}
\label{thm:combinedpreorder}
\end{theorem}
\begin{proof}
 We show that $(p \,W\, q) \land (q \,W\, p) \;\Longrightarrow\; p=q$,
 which implies $\A \,=\, \A/(W \cap W^{-1})$.
 Let $p \,W\, q$ and $q \,W\, p$, 
 then by the definition of $W$,
 there exist mediators $r$ such that $p \dwsim r$ and $q \upsim{\dwsim} r$ and $s$ such that
 $q \dwsim s$ and $p \upsim{\dwsim} s$.
 By the definition of $W$,
 we have that $p \,(\dwsim \circ (\upsim{\dwsim})^{-1})\, s$ 
 and $q \,(\dwsim \circ (\upsim{\dwsim})^{-1})\, r$.
 Thus, there exist mediators $t$ such that 
 $p \dwsim t$ and $s \upsim{\dwsim} t$ and $u$ such that $q \dwsim u$ and $r \upsim{\dwsim} u$.
 By the transitivity of $\upsim{\dwsim}$ we obtain that $p \upsim{\dwsim} t$ and $q \upsim{\dwsim} u$.
 From 1), 2) and Lemma~\ref{lem:combinedpreorder_lemma1} we obtain that $p=t$ and $q=u$.
 So we have $s \upsim{\dwsim} p$ and $r \upsim{\dwsim} q$.
 By Lemma~\ref{lem:combinedpreorder_lemma2} we obtain that $s \upsim{id} p$ and $p \upsim{id} s$, 
 and $r \upsim{id} q$ and $q \upsim{id} r$.
 Thus by (1) we obtain that $p=s$ and $q=r$.
 Since $p \dwsim r$ and $q \dwsim s$, we conclude that $p=q$.
 \qed 
\end{proof}

\newpage
\section{More Data from the Experiments}\label{sec:app_experiments}

\begin{table}[htbp]
\begin{tabular}{lrrrrrrr}
TA name	& $\#Q_i$	& $\#Delta_i$ &  $\#Q_f$ & $\#Delta_f$ & 	Q reduction &	Delta
reduction & Time(s) \\ \hline
A0053	 & 	54	 & 	159	 & 	27	 & 	66	 & 	50	 & 	41.509434	 & 	0.015	\\
A0054	 & 	55	 & 	241	 & 	28	 & 	93	 & 	50.909088	 & 	38.589211	 & 	0.024	\\
A0055	 & 	56	 & 	182	 & 	27	 & 	73	 & 	48.214287	 & 	40.10989	 & 	0.017	\\
A0056	 & 	57	 & 	230	 & 	24	 & 	55	 & 	42.105263	 & 	23.913044	 & 	0.017	\\
A0057	 & 	58	 & 	245	 & 	24	 & 	58	 & 	41.379311	 & 	23.67347	 & 	0.020	\\
A0058	 & 	59	 & 	257	 & 	25	 & 	65	 & 	42.372883	 & 	25.291828	 & 	0.019	\\
A0059	 & 	60	 & 	263	 & 	24	 & 	59	 & 	40	 & 	22.43346	 & 	0.022	\\
A0060	 & 	61	 & 	244	 & 	32	 & 	111	 & 	52.459015	 & 	45.491802	 & 	0.034	\\
A0062	 & 	63	 & 	276	 & 	32	 & 	112	 & 	50.793655	 & 	40.579708	 & 	0.029	\\
A0063	 & 	64	 & 	571	 & 	11	 & 	23	 & 	17.1875	 & 	4.028021	 & 	0.027	\\
A0064	 & 	65	 & 	574	 & 	11	 & 	23	 & 	16.923077	 & 	4.006969	 & 	0.024	\\
A0065	 & 	66	 & 	562	 & 	11	 & 	23	 & 	16.666668	 & 	4.092527	 & 	0.026	\\
A0070	 & 	71	 & 	622	 & 	11	 & 	23	 & 	15.492958	 & 	3.697749	 & 	0.016	\\
A0080	 & 	81	 & 	672	 & 	26	 & 	58	 & 	32.098763	 & 	8.630952	 & 	0.043	\\
A0082	 & 	83	 & 	713	 & 	26	 & 	65	 & 	31.325302	 & 	9.116409	 & 	0.047	\\
A0083	 & 	84	 & 	713	 & 	26	 & 	65	 & 	30.952381	 & 	9.116409	 & 	0.048	\\
A0086	 & 	87	 & 	1402	 & 	26	 & 	112	 & 	29.885057	 & 	7.988588	 & 	0.103	\\
A0087	 & 	88	 & 	1015	 & 	12	 & 	23	 & 	13.636364	 & 	2.26601	 & 	0.060	\\
A0088	 & 	89	 & 	1027	 & 	12	 & 	23	 & 	13.483146	 & 	2.239532	 & 	0.063	\\
A0089	 & 	90	 & 	1006	 & 	12	 & 	21	 & 	13.333334	 & 	2.087475	 & 	0.064	\\
A0111	 & 	112	 & 	1790	 & 	11	 & 	42	 & 	9.821428	 & 	2.346369	 & 	0.139	\\
A0117	 & 	118	 & 	2088	 & 	25	 & 	106	 & 	21.186441	 & 	5.076628	 & 	0.177	\\
A0120	 & 	121	 & 	1367	 & 	12	 & 	21	 & 	9.917356	 & 	1.536211	 & 	0.068	\\
A0126	 & 	127	 & 	1196	 & 	11	 & 	23	 & 	8.661418	 & 	1.923077	 & 	0.083	\\
A0130	 & 	131	 & 	1504	 & 	11	 & 	23	 & 	8.396947	 & 	1.529255	 & 	0.044	\\
A0172	 & 	173	 & 	1333	 & 	11	 & 	23	 & 	6.358381	 & 	1.725431	 & 	0.098	\\
A0177	 & 	178	 & 	1781	 & 	26	 & 	58	 & 	14.606741	 & 	3.256597	 & 	0.085	\\
\hline
Average & 87.07 & 816.04 & 19.78 & 53.59 & 26.97 & 13.94 & 0.052
\end{tabular}

\caption{Results on reducing the 27 moderate-sized tree automata (from
  \libvata's regular model checking examples) with our ${\it Heavy}(1,1)$ algorithm.
The columns give the name of each automaton, $\#Q_i$: its original number of
states, $\#Delta_i$: its original number of transitions,
$\#Q_f$: the number of states after reduction, $\#Delta_f$: the number of
transitions after reduction, the reduction ratio for states in percent $100*\#Q_f/\#Q_i$ (smaller is
better), the reduction ratio for transitions in percent $100*\#Delta_f/\#Delta_i$ (smaller is
better), and the computation time in seconds.
Note that the reduction ratios for transitions are smaller than the ones for
states, i.e., the automata get not only smaller but also sparser.
Experiments run on Intel 3.20GHz i5-3470 CPU.  
}\label{table:moderate}
\end{table}

Tables~\ref{table:moderate} and ~\ref{table:nonmoderate} show the results of reducing 
two automata samples from \libvata's regular model checking examples with our ${\it Heavy}(1,1)$ algorithm.
The first sample (Table~\ref{table:moderate}) contains 27 automata of moderate size while 
the second one (Table~\ref{table:nonmoderate}) contains 62 larger automata.
In both tables the columns give the name of each automaton, 
$\#Q_i$: original number of states, 
$\#Delta_i$: original number of transitions,
$\#Q_f$: states after reduction, 
$\#Delta_f$: transitions after reduction, 
the reduction ratio for states in percent $100*\#Q_f/\#Q_i$ (smaller is better), 
the reduction ratio for transitions in percent $100*\#Delta_f/\#Delta_i$ (smaller is better), 
and the computation time in seconds.
Note that the reduction ratios for transitions are smaller than the ones for states, i.e., 
the automata get sparser.
The experiments were run on Intel 3.20GHz i5-3470 CPU.

\begin{table}[htbp]
\vspace*{-28mm}
{\tiny
\begin{footnotesize}
\begin{tabular}{lrrrrrrr}
TA name	& $\#Q_i$	& $\#Delta_i$ &  $\#Q_f$ & $\#Delta_f$ & 	Q reduction &	Delta
reduction & Time(s) \\ \hline
A246	&	247	&	2944	&	11	&	42	&	4.45	&	1.43	&	0.40	\\
A301	&	302	&	4468	&	12	&	21	&	3.97	&	0.47	&	0.29	\\
A310	&	311	&	3343	&	24	&	52	&	7.72	&	1.56	&	0.59	\\
A312	&	313	&	3367	&	11	&	23	&	3.51	&	0.68	&	0.21	\\
A315	&	316	&	3387	&	24	&	52	&	7.59	&	1.54	&	0.58	\\
A320	&	321	&	3623	&	26	&	65	&	8.10	&	1.79	&	0.56	\\
A321	&	322	&	3407	&	24	&	52	&	7.45	&	1.53	&	0.62	\\
A322	&	323	&	3651	&	35	&	100	&	10.84	&	2.74	&	0.67	\\
A323	&	324	&	6199	&	26	&	112	&	8.02	&	1.81	&	1.48	\\
A328	&	329	&	3517	&	26	&	58	&	7.90	&	1.65	&	0.50	\\
A329	&	330	&	5961	&	24	&	100	&	7.27	&	1.68	&	1.36	\\
A334	&	335	&	3936	&	11	&	23	&	3.28	&	0.58	&	0.72	\\
A335	&	336	&	3738	&	26	&	58	&	7.74	&	1.55	&	0.56	\\
A339	&	340	&	5596	&	12	&	21	&	3.53	&	0.38	&	0.49	\\
A348	&	349	&	3681	&	11	&	23	&	3.15	&	0.62	&	0.27	\\
A354	&	355	&	3522	&	24	&	52	&	6.76	&	1.48	&	0.70	\\
A355	&	356	&	3895	&	25	&	55	&	7.02	&	1.41	&	0.45	\\
A369	&	370	&	4134	&	24	&	52	&	6.49	&	1.26	&	0.31	\\
A387	&	388	&	4117	&	24	&	52	&	6.19	&	1.26	&	0.51	\\
A390	&	391	&	5390	&	11	&	23	&	2.81	&	0.43	&	1.15	\\
A400	&	401	&	5461	&	11	&	23	&	2.74	&	0.42	&	1.36	\\
A447	&	448	&	7924	&	12	&	23	&	2.68	&	0.29	&	2.55	\\
A483	&	484	&	5592	&	25	&	55	&	5.17	&	0.98	&	0.51	\\
A487	&	488	&	4891	&	16	&	28	&	3.28	&	0.57	&	0.33	\\
A488	&	489	&	8493	&	12	&	21	&	2.45	&	0.25	&	2.86	\\
A489	&	490	&	8516	&	12	&	21	&	2.45	&	0.25	&	2.93	\\
A491	&	492	&	8708	&	12	&	21	&	2.44	&	0.24	&	3.03	\\
A493	&	494	&	7523	&	12	&	21	&	2.43	&	0.28	&	0.69	\\
A494	&	495	&	8533	&	12	&	21	&	2.42	&	0.25	&	2.97	\\
A496	&	497	&	8618	&	12	&	21	&	2.41	&	0.24	&	2.81	\\
A498	&	499	&	8612	&	12	&	21	&	2.40	&	0.24	&	3.10	\\
A501	&	502	&	8632	&	12	&	21	&	2.39	&	0.24	&	2.95	\\
A532	&	533	&	8867	&	12	&	23	&	2.25	&	0.26	&	3.20	\\
A569	&	570	&	8351	&	26	&	58	&	4.56	&	0.69	&	0.98	\\
A589	&	590	&	9606	&	12	&	21	&	2.03	&	0.22	&	3.20	\\
A620	&	621	&	9218	&	12	&	21	&	1.93	&	0.23	&	1.45	\\
A646	&	647	&	6054	&	19	&	34	&	2.94	&	0.56	&	0.65	\\
A667	&	668	&	8131	&	26	&	58	&	3.89	&	0.71	&	1.12	\\
A670	&	671	&	11021	&	34	&	76	&	5.07	&	0.69	&	5.80	\\
A673	&	674	&	11157	&	25	&	55	&	3.71	&	0.49	&	5.38	\\
A676	&	677	&	11043	&	34	&	76	&	5.02	&	0.69	&	5.85	\\
A678	&	679	&	11172	&	26	&	56	&	3.83	&	0.50	&	5.32	\\
A679	&	680	&	11032	&	34	&	76	&	5.00	&	0.69	&	5.88	\\
A689	&	690	&	11207	&	31	&	71	&	4.49	&	0.63	&	5.59	\\
A691	&	692	&	11047	&	34	&	76	&	4.91	&	0.69	&	5.61	\\
A692	&	693	&	11066	&	34	&	76	&	4.91	&	0.69	&	6.10	\\
A693	&	694	&	11188	&	34	&	76	&	4.90	&	0.68	&	6.05	\\
A694	&	695	&	11191	&	34	&	76	&	4.89	&	0.68	&	6.09	\\
A695	&	696	&	11070	&	34	&	76	&	4.89	&	0.69	&	5.80	\\
A700	&	701	&	11245	&	36	&	81	&	5.14	&	0.72	&	6.13	\\
A701	&	702	&	11244	&	36	&	83	&	5.13	&	0.74	&	6.00	\\
A703	&	704	&	11255	&	34	&	76	&	4.83	&	0.68	&	6.09	\\
A723	&	724	&	9376	&	26	&	58	&	3.59	&	0.62	&	1.28	\\
A728	&	729	&	11903	&	12	&	21	&	1.65	&	0.18	&	2.97	\\
A756	&	757	&	8884	&	26	&	58	&	3.43	&	0.65	&	1.34	\\
A837	&	838	&	13038	&	11	&	23	&	1.31	&	0.18	&	5.34	\\
A881	&	882	&	15575	&	12	&	21	&	1.36	&	0.13	&	3.36	\\
A980	&	981	&	21109	&	12	&	21	&	1.22	&	0.10	&	4.64	\\
A1003	&	1004	&	21302	&	12	&	21	&	1.20	&	0.10	&	3.99	\\
A1306	&	1307	&	19699	&	25	&	55	&	1.91	&	0.28	&	2.88	\\
A1404	&	1405	&	18839	&	24	&	52	&	1.71	&	0.28	&	3.09	\\
A2003	&	2004	&	30414	&	24	&	52	&	1.20	&	0.17	&	6.98	\\ \hline
Average	&	586.21	&	8865.85	&	21.32	&	47.74	&	4.19	&	0.72	&	2.69	\\
\end{tabular}
\end{footnotesize}

}
\caption{Results on reducing the 62 larger automata (those not called
  moderate-sized) from \libvata.}\label{table:nonmoderate}
\end{table}

\begin{figure}[htbp]
\includegraphics[width=6cm]{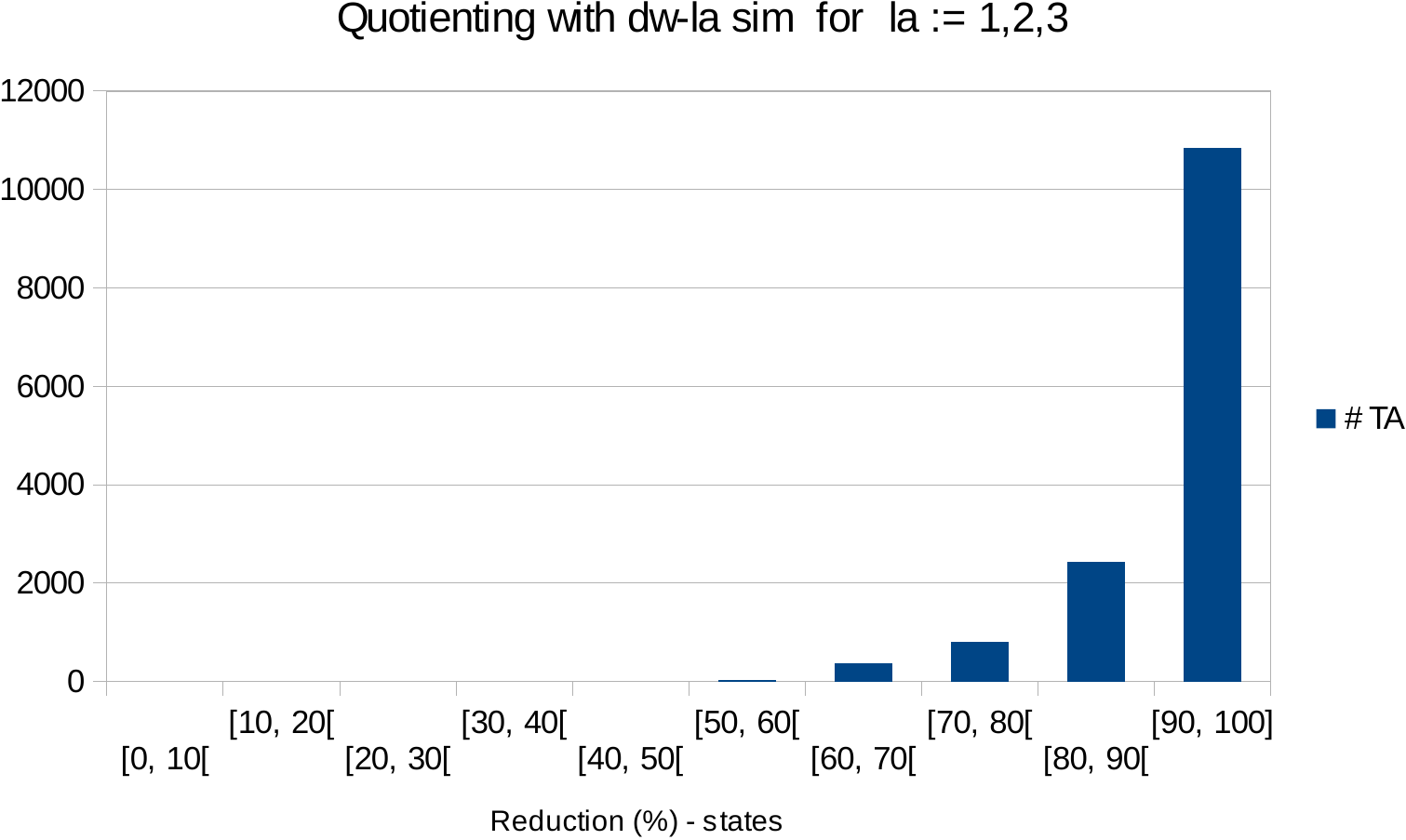}
\includegraphics[width=6cm]{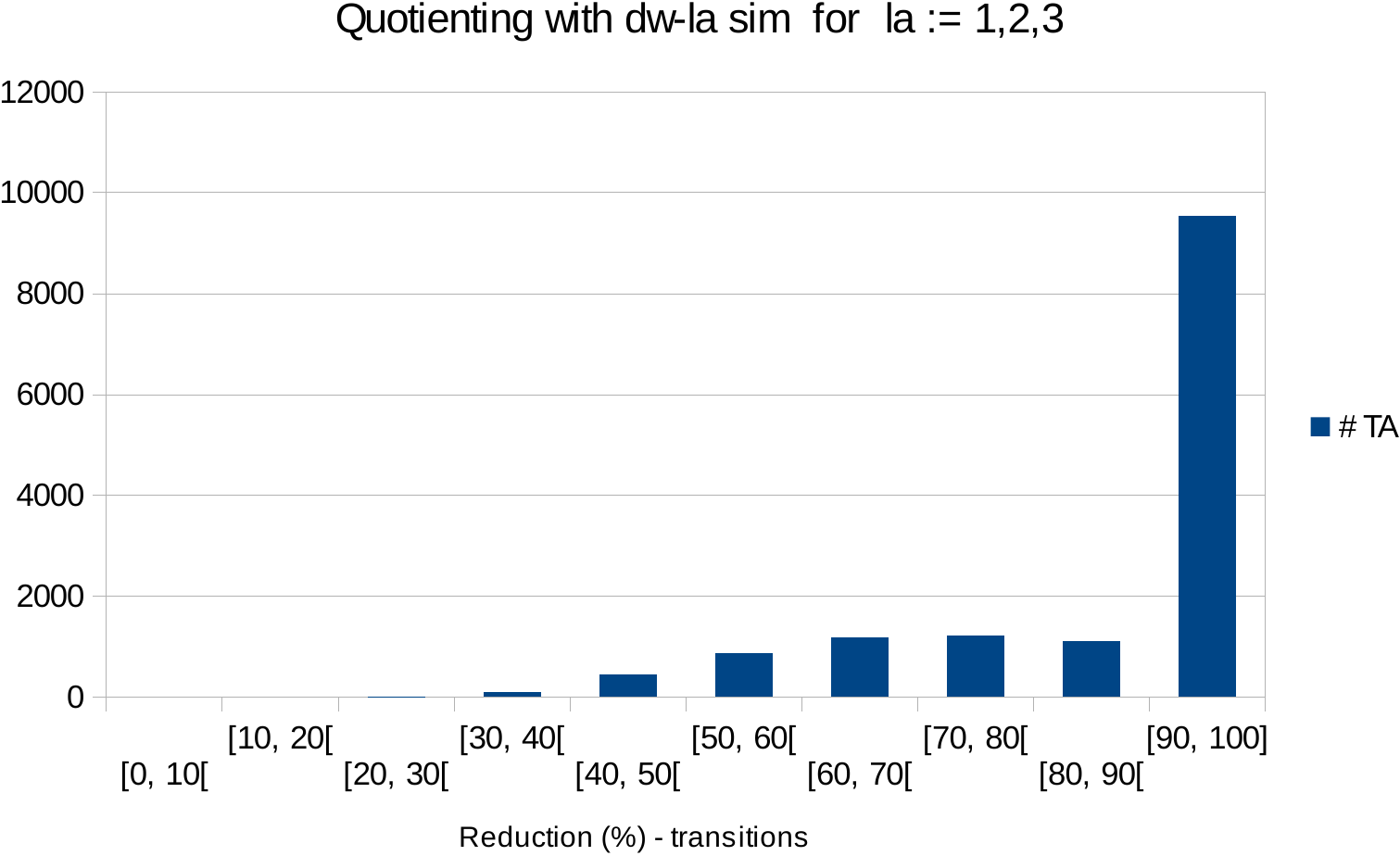}
\includegraphics[width=6cm]{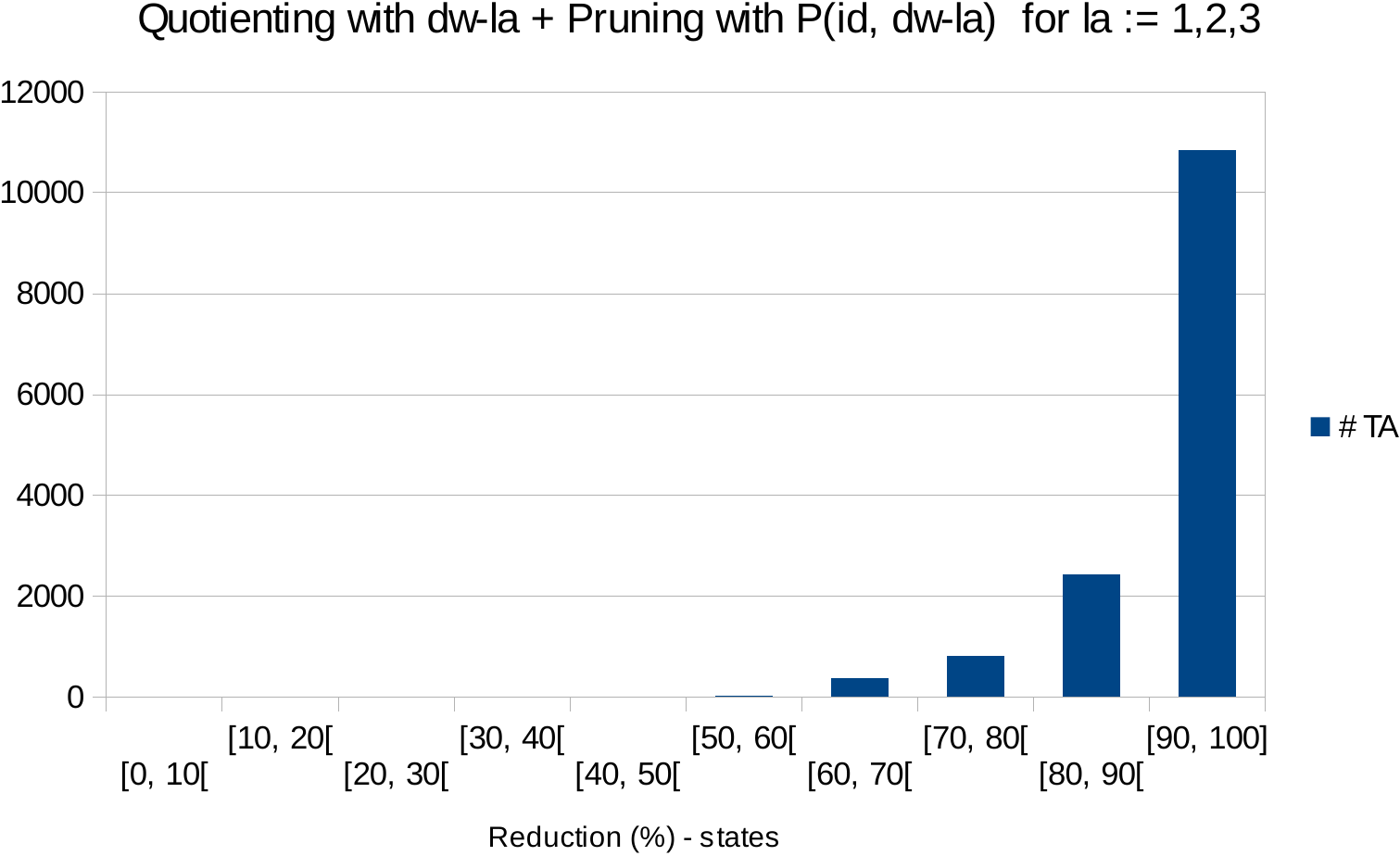} \,
\includegraphics[width=6cm]{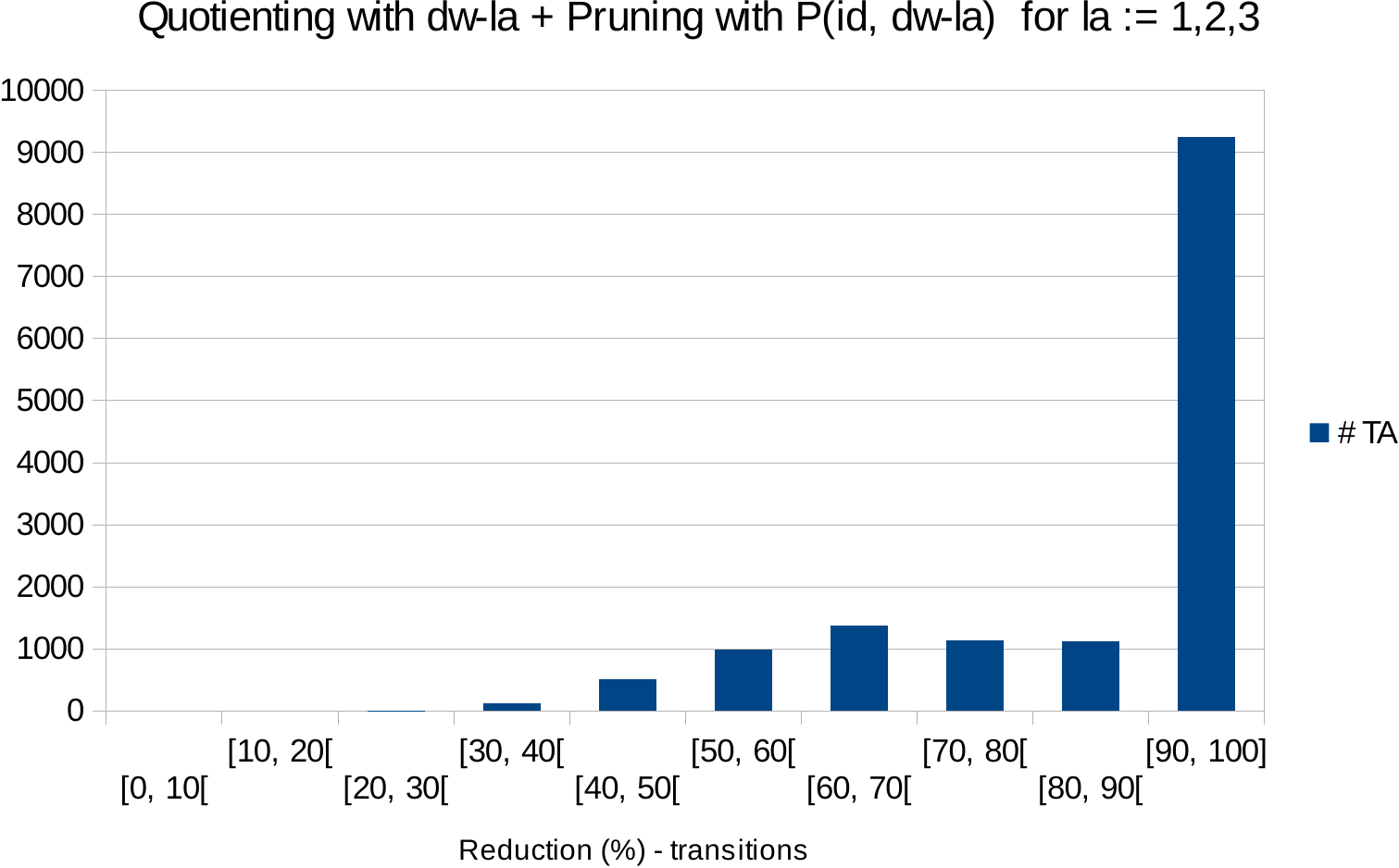}
\includegraphics[width=6cm]{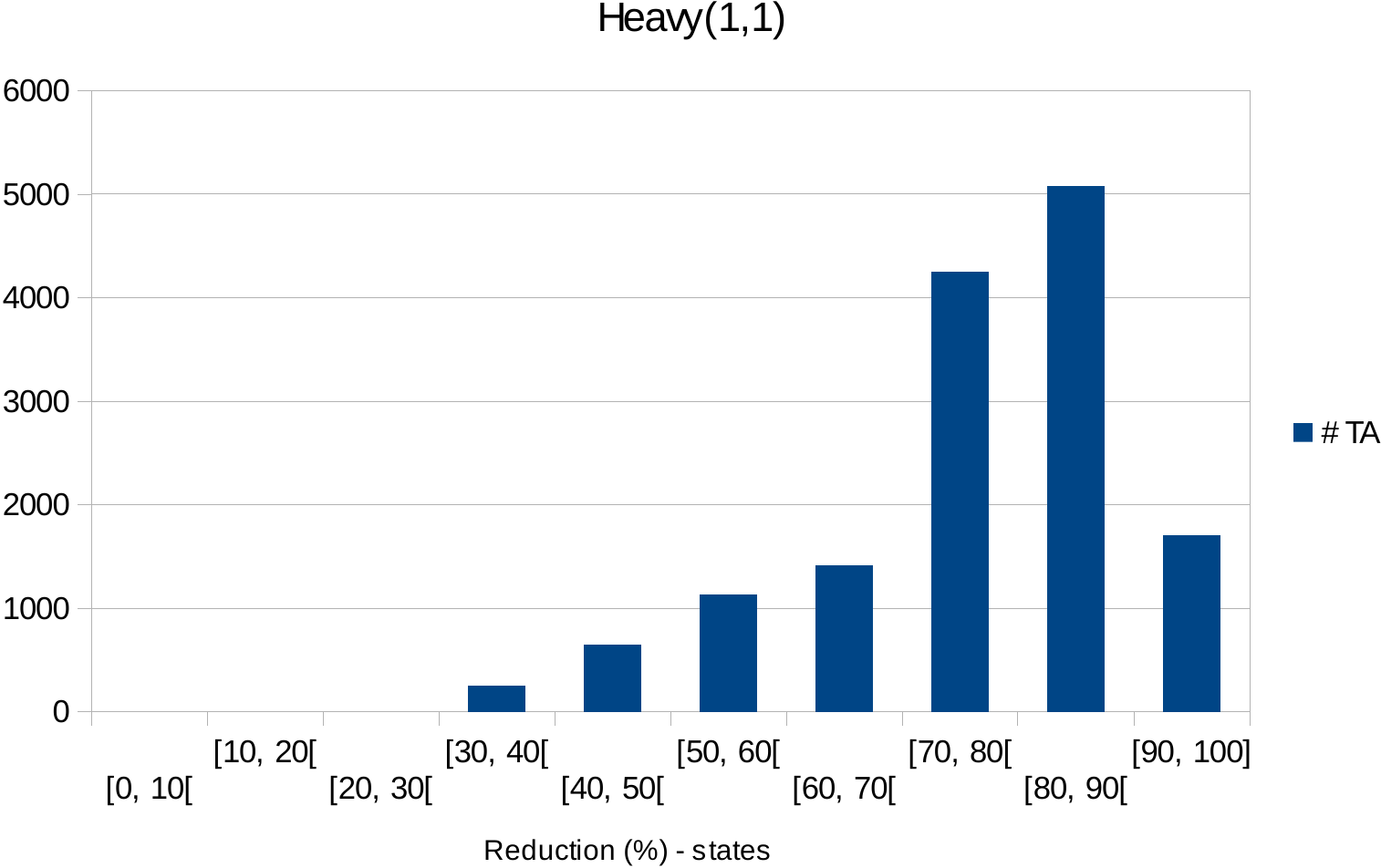} \;
\includegraphics[width=6cm]{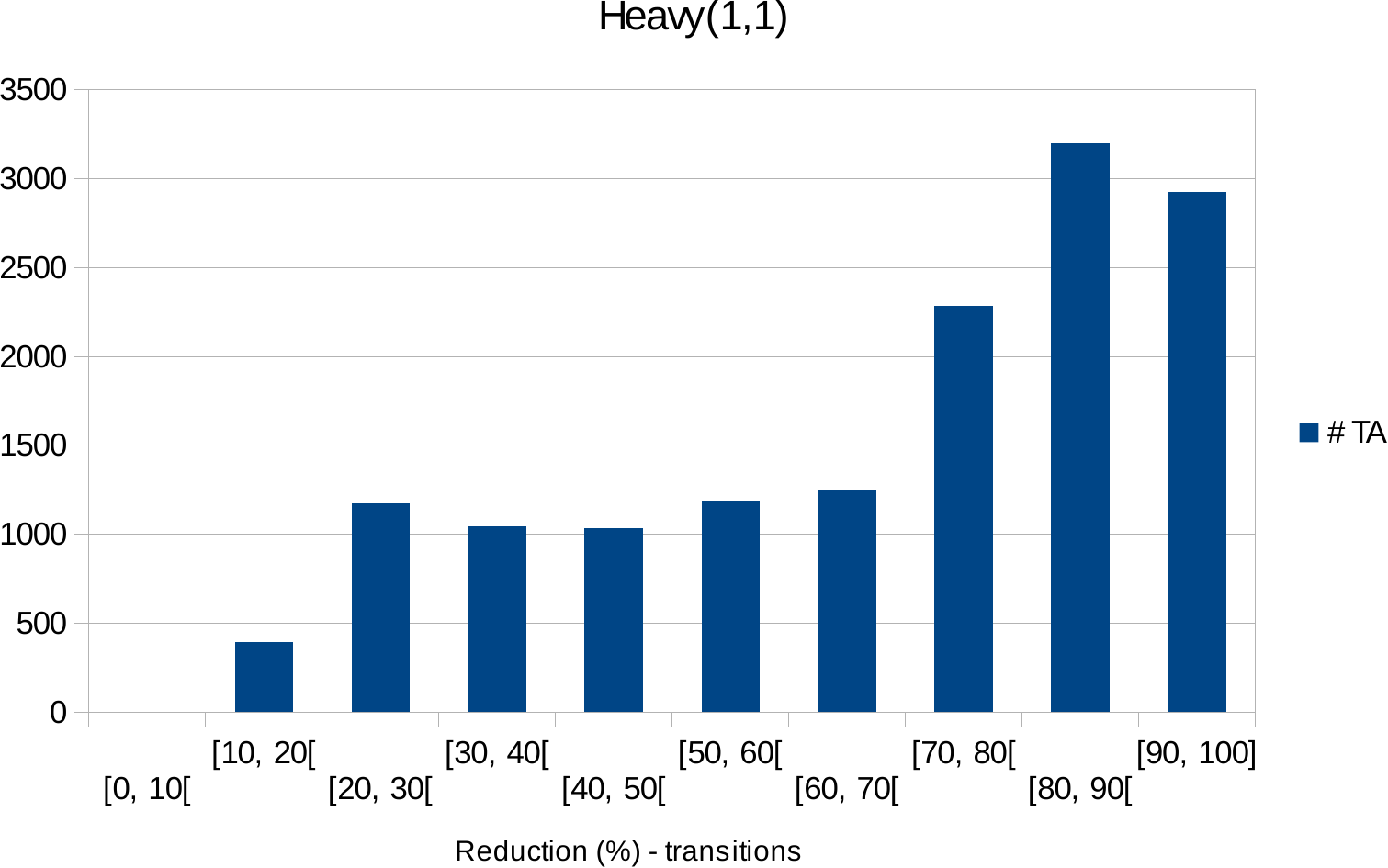}
\caption{Reduction of 14498 tree automata from the Forester tool
  \cite{tool:forester}, by methods RUQ (top row),
 RUQP (middle row),
 and Heavy (bottom row). 
A bar of height $h$ at an interval $[x,x+10[$ means that $h$ of the 14498
    automata were reduced to a size between $x\%$ and $(x+10)\%$ of their
    original size. The reductions in the numbers of states/transitions are
    shown on the left/right, respectively.
    Heavy(1,1) performed significantly better than RUQ and RUQP
    .
    Using 
    lookaheads higher than 1 made hardly any difference in this sample set.}
\label{fig:forester}
\end{figure}

\begin{figure}[htbp]
 \subfloat[This chart illustrates how the average number of transitions after
reduction (in percent of the original number) with each method ($y$-axis)
varied with the transition density $td$ of the sample ($x$-axis) being used (smaller is better).
Note that Heavy(2,4) reduces much better than Heavy(1,1) for ${\it td} \ge 3.5$.]
 {
 \includegraphics[width=12cm,height=4cm]{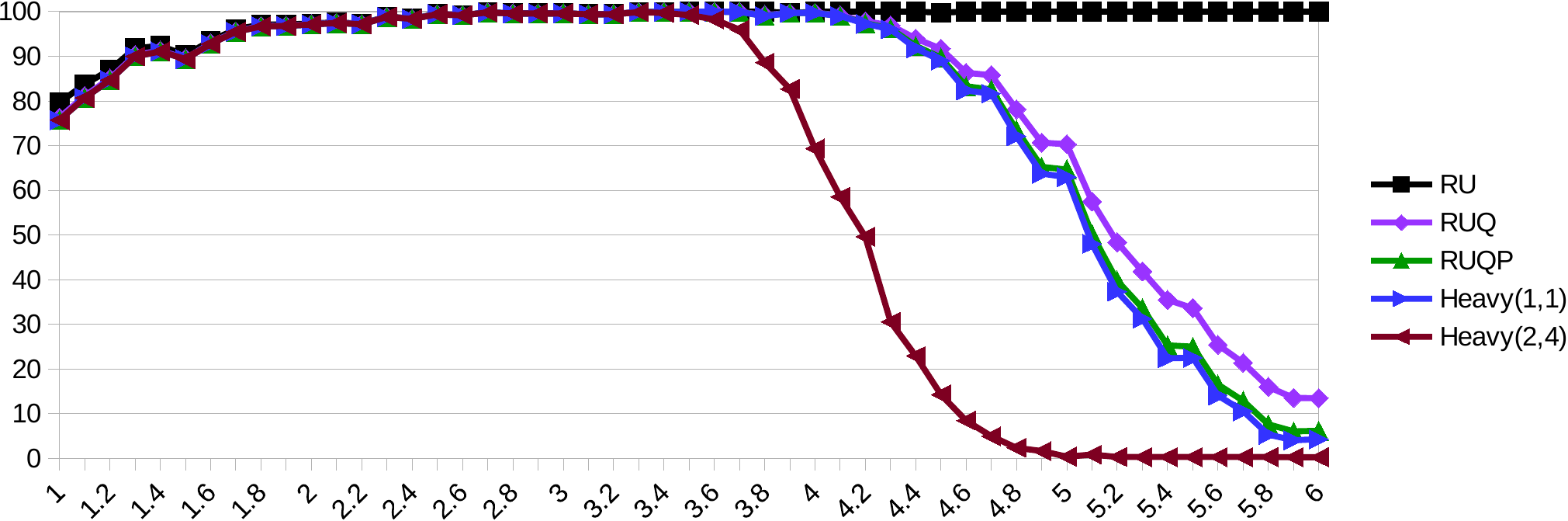}
 }
 \vspace{5mm}
  \subfloat[This chart illustrates how the average time (in seconds) taken by each method ($y$-axis)
varied with the transition density ${\it td}$ of the sample ($x$-axis) being used. 
Heavy(1,1) is significantly faster than Heavy(2,4), which has its time peak at ${\it td}=4.3$.
RUQP is slightly faster than Heavy(1,1) and at ${\it td}=4.5$ it has its highest average value ($0.13$s).]
 {
 \includegraphics[width=12cm,height=4.5cm]{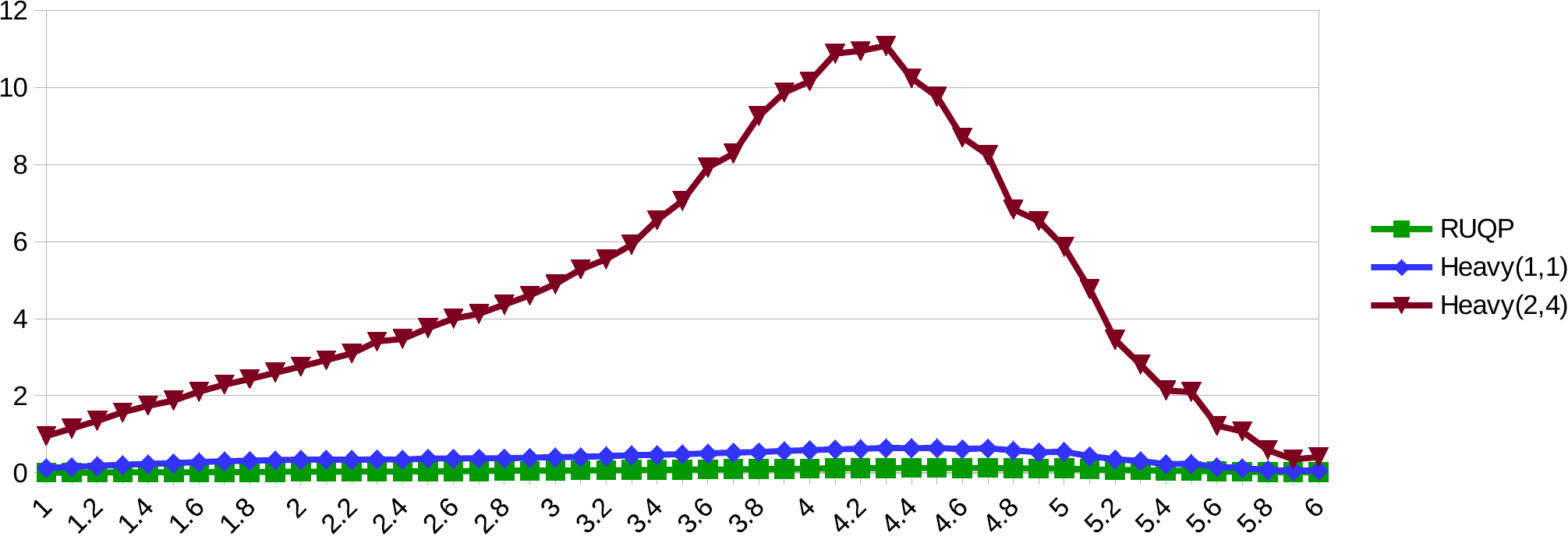}
 }
 \caption{Reduction of Tabakov-Vardi random tree automata with $n=100$, $s=2$ and ${\it ad}=0.8$. 
The top chart shows the average reduction in terms of number of transitions obtained with the various methods,
while the bottom chart shows how long they took.
Each data point in the charts is the average of 400 random automata.}
\end{figure}

\end{document}